\documentclass[11pt,letterpaper]{article}

\usepackage[utf8]{inputenc}
\usepackage[english]{babel}

\usepackage[dvipsnames]{xcolor}
\usepackage[colorlinks=true,pdfpagemode=UseNone,urlcolor=RoyalBlue,linkcolor=RoyalBlue,citecolor=OliveGreen,pdfstartview=FitH]{hyperref}
\definecolor{hkured}{HTML}{EE4123}
\definecolor{hkublue}{HTML}{009BD4}
\definecolor{hkugreen}{HTML}{00B38C}

\usepackage{tikz}
\usetikzlibrary{shapes, arrows.meta, positioning}
\usepackage{pgfplots}

\usepackage{subcaption}
\usepackage{tabto}

\usepackage[margin=1in]{geometry}
\usepackage{amsfonts,amsmath,amsthm,amssymb}
\usepackage{mathtools,thmtools}
\usepackage{bm}

\usepackage{nicefrac}

\declaretheorem{theorem}
\declaretheorem[sibling=theorem]{lemma}
\declaretheorem[sibling=theorem]{corollary}

\declaretheorem[style=definition]{definition}

\usepackage{todonotes}
\usepackage{tcolorbox}
\usepackage{booktabs}
\usepackage{caption}
    
\usepackage{multirow}

\usepackage{enumitem}

\usepackage[numbers, sort]{natbib}

\newcommand{\dif}{\mathrm{d}}

\newcommand{\ALG}{\mathrm{ALG}}
\newcommand{\OPT}{\mathrm{OPT}}

\newcommand{\E}{\mathbf{E}}

\newcommand{\Distribution}{F}
\newcommand{\mass}{f}
\newcommand{\matchrate}{\mu}

\renewcommand{\Pr}[2][]{\mbox{\rm\bf Pr}_{#1}\left[#2\right]}%

\title{Stochastic Online Correlated Selection}

\author{
	Ziyun Chen
	\thanks{University of Washington. Email: ziyuncc@cs.washington.edu. This work was done when the author was at Tsinghua University.}
	\and
	Zhiyi Huang
	\thanks{The University of Hong Kong. Email: zhiyi@cs.hku.hk, sunenze@connect.hku.hk.}
	\and
	Enze Sun
	\footnotemark[2]
}

\date{August 2024}

\begin{document}

\begin{titlepage}
    \thispagestyle{empty}
    \maketitle
    \begin{abstract}
        \thispagestyle{empty}
        We initiate the study of Stochastic Online Correlated Selection (SOCS), a family of online rounding algorithms for the general Non-IID model of Stochastic Online Submodular Welfare Maximization and its special cases such as unweighted and vertex-weighted Online Stochastic Matching, Stochastic AdWords, and Stochastic Display Ads.
At each time step, the algorithm sees the type of an online item and a fractional allocation of the item, then immediately allocates the item to an agent.
We propose a metric called the convergence rate that measures the quality of SOCS algorithms in the above special cases.
This is cleaner than most metrics in the related Online Correlated Selection (OCS) literature and may be of independent interest.

We propose a Type Decomposition framework that reduces the design of SOCS algorithms to the easier special case of two-way SOCS.
First, we sample a surrogate type whose fractional allocation is half-integer.
The rounding is trivial for a one-way surrogate type fully allocated to one agent.
For a two-way surrogate type split equally between two agents, we round it using a two-way SOCS.
We design the distribution of surrogate types to get two-way types as often as possible, while respecting the original fractional allocation in expectation.

Following this framework, we make progress on numerous problems including two open questions related to AdWords:
\begin{itemize}
\item \emph{Online Stochastic Matching:~} 
	We improve the state-of-the-art $0.666$ competitive ratio for unweighted and vertex-weighted matching by Tang, Wu, and Wu (STOC 2022) to $0.69$.
\item \emph{Query-Commit Matching:~} 
	We further enhance the above competitive ratio to $0.705$ in the random-order relaxation.
    Using a known reduction, we get that same ratio in the Query-Commit model.
	This result improves the best previous ratios $0.696$ for unweighted matching by Mahdian and Yan (STOC 2011) and $0.662$ for vertex-weighted matching by Jin and Williamson (WINE 2021).
\item \emph{Stochastic AdWords:~}
	We give a $0.6338$ competitive algorithm for Stochastic AdWords, breaking the $1-\frac{1}{e}$ barrier for the first time.
	This answers a decade-old open question from the survey by Mehta (2013) about breaking this barrier in the IID special case.
\item \emph{AdWords:~}
	The framework of Type Decomposition can also be applied to the adversarial model if the two-way rounding algorithm is oblivious to the distribution of future items.
    From the two-way algorithm's viewpoint, the fixed adversarial sequence of items is a non-IID distribution that is a point mass, and the stochasticity comes from sampling surrogate types.
	Following this framework, we get the first multi-way OCS for AdWords, addressing an open question in the OCS literature.
	This further leads to a $0.504$ competitive ratio for AdWords, improving the previous $0.501$ ratio by Huang, Zhang, and Zhang (FOCS 2020).
\item \emph{Stochastic Display Ads:~}
	We design a $0.644$ competitive online algorithm for Stochastic Display Ads, breaking the $1-\frac{1}{e}$ barrier for the first time.
\end{itemize}

    \end{abstract}
\end{titlepage}

\pagestyle{empty}
\setcounter{tocdepth}{2}
\tableofcontents
\clearpage
\pagestyle{plain}
\setcounter{page}{1}

\section{Introduction}

Lots of practical scenarios involve solving matching problems under uncertainty of the underlying graphs.
For example, online advertising service providers like Google Ads and Microsoft Ads match user impressions to advertisers without accurate knowledge of the impressions that may come next.
Ride-hailing mobile apps including Uber and Lyft match riders and drivers under uncertainty of the ride requests and the drivers' geographical information in the future.
Kidney exchange platforms such as the National Kidney Registry match donor-patient pairs with each other, with only partial information about their compatibility.

One way to model uncertainty is to consider a maximin approach, where the algorithm needs to make online decisions that are robust to the worst-case future input as if it would be chosen by an adversary.
We will refer to this as the \emph{adversarial model}.
A classical example is the \emph{Online Bipartite Matching} problem introduced by \citet*{KarpVV:STOC:1990}

Another popular way is to capture the uncertainty in optimization problems stochastically.
For instance, consider online advertising and treat the user impressions and advertisers as the vertices on the two sides of a bipartite graph.
The advertisers register beforehand and hence are known to the platform.
User impressions arrive one by one and the platform needs to match each impression upon its arrival.
Therefore, we will refer to the advertisers and user impressions as offline and online vertices respectively.
While the platform cannot know in advance which impression will come next, it can predict based on data that the next user is drawn from a prior distribution, e.g., the uniform distribution over the user population.
Depending on the business model, we may want to maximize the cardinality of the matching, or the sum of the values of the matched edges, etc.
This is the \emph{Online Stochastic Matching} problem posed by \citet*{FeldmanMMM:FOCS:2009}.

The prior distribution may be time-dependent in some scenarios.
In the ride-hailing example, the pattern of ride requests in the morning when people are commuting to work could be different from the pattern in the evening when people are going home.
In other words, the prior distributions are no longer independently and identically distributed (IID) like in \emph{Online Stochastic Matching}.
\citet*{TangWW:STOC:2022} recently considered such a non-IID variant of \emph{Online Stochastic Matching}.

Both examples above consider vertex arrivals and reveal the incident edges as vertices arrive.
In other applications such as kidney exchange, however, one needs to proactively probe an edge to find whether it exists.
Before probing an edge $(i,j)$, we can only estimate the probability $0 \le p_{ij} \le 1$ that it presents.
For instance, a hospital needs to evaluate compatibility through blood tests, tissue typing, etc., before matching two donor-patient pairs.
Further, once we find that two donor-patient pairs are compatible, i.e., that the edge exists, we must commit to matching them.
This is the \emph{Query-Commit} model by \citet*{ChenIKMR:ICALP:2009}.

In the past few years, a lot of progress has been made on several long-standing open questions related to online matching, through the study of online rounding algorithms.
Most related to this paper is the concept of \emph{Online Correlated Selection (OCS)} introduced by \citet*{FahrbachHTZ:FOCS:2020} for the \emph{Display Ads} problem posed by \citet{FeldmanKMMP:WINE:2009}, which is an edge-weighted generalization of the \emph{Online Bipartite Matching} problem.
Online matching problems usually become easier if we relax them by allowing matching vertices fractionally.
For example, \citet{FeldmanKMMP:WINE:2009} solved the fractional variant of Display Ads by giving an optimal $1-\frac{1}{e}$ competitive algorithm.
By contrast, it had been open for more than a decade since then, until the work of \citet{FahrbachHTZ:FOCS:2020}, whether there is an online algorithm for the (integral) Display Ads problem with a competitive ratio better than $0.5$, the baseline set by the greedy algorithm.
In a nutshell, OCS is an online rounding algorithm that rounds the fractional matching decisions made by an online algorithm for the relaxed problem, by making randomized integral decisions.
Importantly, it will negatively correlate the decisions regarding any fixed offline vertex, because the baseline independent rounding algorithm only leads to the trivial $0.5$ competitive ratio for the online matching problems.

\citet{FahrbachHTZ:FOCS:2020} designed an OCS with provable negative correlation, and used it to achieve a $0.508$ competitive ratio for Display Ads.
Since then, the technique of OCS has been improved in a series of follow-up research by \citet{ShinA:ISAAC:2021}, \citet{GaoHHNYZ:FOCS:2021}, and finally \citet{BlancC:FOCS:2021}, who gave the state-of-the-art $0.536$-competitive algorithm for Display Ads.
Further, \citet*{HuangZZ:FOCS:2020} designed an OCS for the \emph{AdWords} problem posed by \citet*{MehtaSVV:JACM:2007}, and gave the first algorithm that breaks the $0.5$ barrier in the general case, answering another decade-old open question in the literature of online matching.

Finally, there were sporadic attempts to apply the OCS technique to the stochastic models of online matching.
\citet*{HuangSY:STOC:2022} modified the OCS and analysis by \citet{GaoHHNYZ:FOCS:2021} to get an OCS in the Poisson arrival model of unweighted and vertex-weighted matching.
Based on their Poisson OCS and an asymptotic equivalence between the IID Online Stochastic Matching and Poisson arrival models, they gave the state-of-the-art $0.716$-competitive algorithm.
\citet*{TangWW:STOC:2022} analyzed the OCS algorithm by \citet{GaoHHNYZ:FOCS:2021} in the model of Non-IID Online Stochastic Matching, and obtained a $0.666$ competitive ratio, breaking the $1-\frac{1}{e}$ barrier for the first time.

\subsection{Conceptual Contribution: Stochastic Online Correlated Selection (SOCS)}

This paper initiates the study of \emph{Stochastic Online Correlated Selection} (SOCS), formalizing the concept of online rounding algorithms for stochastic online matching problems.
We highlight two differences compared to the previous studies on OCS.
Compared to \citet{HuangSY:STOC:2022} and \citet{TangWW:STOC:2022} who also studied OCS in the stochastic model, this paper will consider the whole family of online matching problems under a unified framework in the general non-IID model, rather than only focusing on unweighted and vertex-weighted matching.
Further, we will measure the quality of SOCS algorithms through a unified metric called the \emph{convergence rate}.
Given a SOCS algorithm's convergence rate, the competitive ratio for the corresponding online matching problem will follow easily as a corollary.
This is cleaner than the counterparts in previous studies, which usually considered complicated metrics of rounding quality (e.g., the concept of consecutive steps in Display Ads~\cite{FahrbachHTZ:FOCS:2020, ShinA:ISAAC:2021, GaoHHNYZ:FOCS:2021, BlancC:FOCS:2021} and the panorama view in AdWords~\cite{HuangZZ:FOCS:2020}) and often required a nested analysis of the OCS and the underlying online matching problem.

We now introduce the concept of SOCS.
Consider a general problem known as \emph{Online Submodular Welfare Maximization}, which captures all aforementioned online matching problems as its special cases. 
Let there be a set of offline agents (offline vertices) and a set of online items (online vertices).
Each agent $j$ has a non-negative, non-decreasing, and submodular value function $v_j$ over subsets of items.
The items arrive one by one in $T$ discrete time steps.
The item at each time $t$ is sampled independently (but not identically in general) from a distribution $\Distribution^t$.
In the beginning, the online algorithm knows the distributions but not the realization of items.
Further, the algorithm needs to immediately allocate each item to an agent upon the item's arrival.
We want to maximize the sum of the agents' values for the subsets of items allocated to them, known as the \emph{social welfare}.

In each time step $t$ in the SOCS setting, the algorithm further receives a fractional allocation $\matchrate^t = \big( \matchrate^t_j \big)$ where $\matchrate^t_j$ is the fraction of the item allocated to agent $j$.
Naturally, the fractional allocation ensures that  $\matchrate^t_j \ge 0$ and $\sum_j \matchrate^t_j = 1$.
We further assume that the fractional allocation satisfies a set of linear constraints that would be satisfied by the allocation statistics of any offline allocation rule, i.e., if $\matchrate^t_j$ was the probability that the offline allocation rule would allocate item $t$ to $j$ over the random realization of the other items.
The SOCS algorithm will then decide how to allocate the item at time $t$, guided by the fractional allocation.
For example, we consider independent rounding as the baseline SOCS, which samples an agent $j$ independently in time step $t$ treating $\matchrate^t$ as a distribution.
This paper will introduce better SOCS algorithms for different matching problems, intuitively by introducing negative correlation in the decisions related to any fixed agent $j$.
We will next elaborate on the metrics for evaluating SOCS algorithms.

\paragraph{Unweighted and Vertex-Weighted Matching.}
These are the special case when each offline vertex/agent $j$'s value function is the indicator for whether it receives an item that it has an edge with, scaled by the vertex-weight if applicable.
We consider the total amount of items allocated to agent $j$ according to the fractional allocation, which we will denote as:
\[
	y_j = \sum_t \matchrate^t_j
	~.
\]
Here, we assume without loss of generality that $\matchrate^t_j$ is positive only when there is an edge between the vertices.
If $y_j$ is small, then the algorithm can match offline vertex $j$ with little probability.
After all, it is supposed to be guided by the fractional allocation.
On the flip side, if $y_j$ is large, then it must match $j$ with a large probability.
Hence, we will quantify the quality of SOCS by the function $g(y_j)$ that upper bounds the probability that $j$ stays unmatched.
This is similar to how \citet{GaoHHNYZ:FOCS:2021} measured the quality of OCS in the adversarial model of unweighted and vertex-weighted matching.
We will call function $g$ the \emph{convergence rate} of the SOCS algorithm.

\paragraph{AdWords.}
This is the special case when the offline vertex/agent $j$'s value is a budget-additive function $v_j(S) = \min \big\{ \sum_{t \in S} b^t_j, B_j \big\}$ where $b^t_j$ is agent $j$'s bid for the item at time $t$, and $B_j$ is agent $j$'s budget.
Observe that the unweighted matching problem is the special case of AdWords when $b^t_j \in \{0, 1\}$ and $B_j = 1$.
Further, from any agent $j$'s point of view, the SOCS's guarantee shall be invariant to scaling its budget $B_j$ and bids $b^t_j$'s by the same multiplicative factor.
Therefore, it is natural to generalize the definition therein to consider:
\[
	y_j = \sum_t \matchrate^t_j \cdot \frac{b^t_j}{B_j}
	~.
\] 
Correspondingly, an algorithm's convergence rate $g(y_j)$ upper bounds an agent $j$'s expected portion unused budget, i.e., agent $j$'s expected value for the rounded allocation is at least $\big(1 - g(y_j) \big) B_j$.

\paragraph{Display Ads.}
This is the special case when offline vertex/agent $j$'s value equals the maximum edge-weight allocated to it.
Following the now standard approach originally proposed by \citet{DevanurHKMY:TEAC:2016}, we will account for each agent $j$'s contribution to the objective by considering different weight-levels.
More precisely, for any weight-level $w > 0$, consider the total amount of edges with weights at least $w$ that is allocated to $j$, which we will denote as:
\[
	y_j(w) = \sum_{t : w^t_j \ge w} \matchrate^t_j
	~.
\]

Agent $j$'s value for the fractional allocation is then:
\[
	\sum_t w^t_j \matchrate^t_j = \int_0^\infty y_j(w) \:\dif w
	~.
\]

We require a SOCS to have a uniform convergence rate at every weight-level. 
That is, for the same function $g$, and for every weight-level $w > 0$, the probability of \emph{not} allocating any edges with weight at least $w$ to agent $j$ is at most $g\big(y_j(w)\big)$.

\subsection{Conceptual Contribution: General SOCS from Two-Way SOCS}

The literature on OCS suggests that the two-way special case is easier than the general case.
A two-way OCS/SOCS instance's $\matchrate^t$ is half-integer in every time step $t$, i.e., $\matchrate^t_j \in \{0, \frac{1}{2}\}$.
In other words, the fractional allocation provides two choices in each step for the OCS/SOCS to choose from, and has no preference between the two choices.
For example, \citet{GaoHHNYZ:FOCS:2021} gave the optimal two-way OCS for unweighted and vertex-weighted matching.
By contrast, the optimal multi-way counterpart remains elusive.
For AdWords, \citet{HuangZZ:FOCS:2020} only provided a two-way (Panoramic) OCS.
It was open before this paper how to design a non-trivial multi-way OCS for AdWords.
Finally, the first OCS for Display Ads by \citet{FahrbachHTZ:FOCS:2020} was only applicable to two-way instances.
The multi-way counterpart was only introduced a year later by \citet{BlancC:FOCS:2021}.

Due to the above observation, we originally only focused on the two-way special case of SOCS.
To our surprise, however, we found that one can convert any two-way SOCS into a general SOCS through a general method which we called \textbf{Type Decomposition}.
Given the fractional allocation $\matchrate^t$ at time step $t$, we randomly sample either a \emph{two-way surrogate type} with two choices of offline agents $\{j, k\}$, or a \emph{one-way surrogate type} with only one choice $j$.
In the former case, we let the two-way SOCS select between agents $j$ and $k$ and allocate the item to that agent.
In the latter case, we allocate the item to the only choice $j$.

What is an appropriate distribution over the surrogate types?
We argue the following two factors are important.
First, the expected fractional allocation to any offline vertex $j$ shall equal the original fractional allocation $\matchrate^t$, so that the $y_j$ defined in the previous subsection remains the same in expectation.
Second, we prefer two-way surrogate types over the one-way counterpart because we would like to exploit the power of the two-way SOCS algorithm to the maximum extent possible.
As we will show in the next few subsections, one of the main challenges in our analyses will be upper bounding the influence of one-way surrogate types.

We next demonstrate the distribution driven by these two factors through two examples.
The first example considers four offline agents $\{1, 2, 3, 4\}$ and fractional allocation $\matchrate^t = (0.1, 0.2, 0.3, 0.4)$.
A distribution driven by the above factors would only sample two-way surrogate types, e.g., $\{1, 3\}$ with probability $0.2$, $\{2, 4\}$ with probability $0.4$, and $\{3, 4\}$ with probability $0.4$.
On the other hand, observe that the expected allocation to an agent $j$ cannot exceed half if we only sample two-way surrogate types.
Hence, if the original fractional allocation $\matchrate^t$ assigns more than half of the item to an agent $j$, then we must sample one-way type $j$ with positive probability in order to satisfy the invariant about the expected allocation.
The second example considers a fractional allocation $\matchrate^t = (0.1, 0.1, 0.1, 0.7)$ over the same four offline agents.
The unique distribution determined by the above factors would sample a one-way surrogate type $j = 4$ with probability $0.4$, and two-way surrogate types $\{1, 4\}, \{2, 4\}, \{3, 4\}$ each with probability $0.2$.
In general, we will use a sampling algorithm~\cite{JailletL:MOR:2014, HuangS:STOC:2021} originally designed for Online Stochastic Matching to get an appropriate right distribution of surrogate types. See Subsection~\ref{sec:socs-reduction} for details.

This approach is quite natural and almost obvious in hindsight, but was either overlooked or underestimated in previous works.
Surprisingly, it is powerful enough to improve the state-of-the-art of many stochastic online matching problems.
In particular, we will use it to break the $1-\frac{1}{e}$ barrier for AdWords and Display Ads in the Non-IID model.
This answers a decade-old open question from the survey by \citet{Mehta:FTTCS:2013} about breaking the barrier for AdWords in the more restrictive IID model.
We will further use this approach to get the first multi-way OCS for AdWords, answering an open question from the OCS literature.
We will elaborate on the results and techniques in different settings in the next few subsections.
See Table~\ref{tab:summary} for a summary of our results. 

Due to its wide range of applications, we consider this reduction from general SOCS (and also OCS in the case of AdWords) to two-way SOCS as another conceptual contribution of this paper.

\begin{table}[t]
\centering
\caption{Summary of Results}
\label{tab:summary}
\renewcommand{\arraystretch}{1.5}
\begin{tabular}{|l|c|c|c|c|}
    \hline
    & Adversarial & IID & Non-IID & Query-Commit \\
    \hline
    Unweighted & \multirow{2}{*}{$1-\frac{1}{e}$~\cite{KarpVV:STOC:1990, AggarwalGKM:SODA:2011}} & \multirow{2}{*}{$0.716$~\cite{HuangSY:STOC:2022}} & \multirow{2}{*}{$0.666$~\cite{TangWW:STOC:2022} $\to$ $\bm{0.69}$} & $0.696$~\cite{MahdianY:STOC:2011} $\to$ $\bm{0.705}$ \\
    \cline{1-1} \cline{5-5}
    Vertex-Weighted & & & & $0.662$~\cite{JinW:WINE:2021} $\to$ $\bm{0.705}$ \\
    \hline
    AdWords & $0.501$~\cite{HuangZZ:FOCS:2020} $\to$ $\bm{0.504
    }$ & \multicolumn{2}{|c|}{$1-\frac{1}{e}$~\cite{KapralovPV:SODA:2013} $\to$ $\bm{0.6338}$} & n.a. \\
    \hline    
    Display Ads & $0.536$~\cite{BlancC:FOCS:2021} & $0.706$~\cite{HuangSY:STOC:2022} & $1-\frac{1}{e}$~\cite{KapralovPV:SODA:2013} $\to$ $\bm{0.644}$ & n.a. \\
    \hline
\end{tabular}
\end{table}

\subsection{Results and Techniques: Unweighted and Vertex-Weighted Matching}

\paragraph{Optimal Two-Way SOCS.}
Our first result (Theorem~\ref{thm:stochastic-ocs}) is a two-way SOCS for unweighted and vertex-weighted matching with convergence rate:
\[
	g(y_j) = e^{-2y_j} (1+y_j)
	~.
\]
This is tight in the sense that for any $y_j > 0$ and any two-way SOCS, there is an instance and an offline vertex $j$ therein with this value of $y_j$ which stays unmatched with probability $g(y_j)$.

\paragraph{Summary of Techniques.}
We now describe the algorithm and the main ideas in its analysis. 
Consider selecting between two choices $j$ and $k$ in some time step $t$.
If only one of the two choices is still unmatched, we will obviously select that vertex.
In other words, the only non-trivial case is when both $j$ and $k$ are still unmatched.
To select between them, we will consider the expectation sum of fractional allocation to $j$ and $k$ so far \emph{before time step $t$}, which we will abuse notation and also denote as $y_j$ and $y_k$ in this discussion.
We will scale the sampling weights of $j$ and $k$ by $w_j = e^{2 y_j}$ and $w_k = e^{2 y_k}$, i.e., giving higher priority to an offline vertex if it has more opportunities to get matched in the past in expectation (but is still unmatched).
Intuitively, this compensates for its misfortune in the past.
The sampling probability of $j$ can then be written as:
\[
	\frac{\frac{1}{2} \cdot w_j \cdot \mathbf{1}_\text{$j$ unmatched}}{\frac{1}{2} \cdot w_j \cdot \mathbf{1}_\text{$j$ unmatched} \,+\, \frac{1}{2} \cdot w_k \cdot \mathbf{1}_\text{$k$ unmatched}}
	~.
\]

Although the idea of scaling the sampling weights based on the vertices' cumulative fractional allocation is from the existing literature~\cite{GaoHHNYZ:FOCS:2021, HuangSY:STOC:2022}, we stress that our choice of $w_j$ and $w_k$ is more aggressive than what the existing analysis could allow, and based on a new analysis that may be of independent interest. 
The choice of $w_j$ and $w_k$ in the existing algorithms is driven by the invariant that $w_j \cdot \mathbf{E} [ \mathbf{1}_\text{$j$ unmatched} ] \le 1$ so that each choice's expected contribution to the denominator is at most $\frac{1}{2}$, the intended fractional allocation of this step.
For example, the Poisson OCS~\cite{HuangSY:STOC:2022} lets $w_j = e^{y_j}$ based on $\mathbf{E} [ \mathbf{1}_\text{$j$ unmatched} ] = \Pr{\text{$j$ unmatched}} \le e^{-y_j}$.
By contrast, we let $w_j = e^{2 w_j}$ even when $j$ could stay unmatched with probability more than $e^{-2y_j}$ according to our convergence rate.

Why is our choice of $w_j$ and $w_k$ still feasible then?
To demonstrate the idea, we consider the following recurrence about the relation among the probabilities that $j$ stays unmatched before and after time $t$, denoted as $u_j^{t-1}$ and $u_j^t$, and the probability that both $j$ and $k$ are unmatched before time $t$, denoted as $u_{jk}^{t-1}$:
\[
	u_j^t
	~=~ 
	\Pr{\text{\,$j$ is not a choice\,}} \cdot u_j^{t-1}
	+ 
	\sum_{k \ne j} \Pr{\text{\,$j, k$ are the choices\,}} \cdot u_{jk}^{t-1} \cdot  \frac{\frac{1}{2} \cdot w_k}{\frac{1}{2} \cdot w_j \,+\, \frac{1}{2} \cdot w_k}
	~.
\]

Had we followed the logic in the old analysis~\cite{HuangSY:STOC:2022}, $k$'s contribution to $u_{jk}^{t-1}$, e.g., the mentioned baseline bound of $e^{-y_k}$, would need to cancel $w_k$ so that we could derive a new recurrence that is only about $j$ but not the other vertices.
Our new insight is that the last term scales proportional to the squared root of $w_k$ rather than linearly, by applying AM-GM to the denominator.
Hence, we can use the more aggressive $w_k = e^{2 y_k}$ to get the optimal convergence rate.

Another difference compared to the existing analyses~\cite{GaoHHNYZ:FOCS:2021, HuangSY:STOC:2022} is that we avoid applying the Jensen inequality to aggregate the terms on the right-hand-side of the recurrence, exploiting the simpler structure of the two-way special case.
See Subsection~\ref{sec:optimal-two-way-socs} for the detailed analysis.

\paragraph{Improved Algorithms in Non-IID and Query-Commit Models.}
Next, we apply the type decomposition to obtain a general SOCS for unweighted and vertex-weighted matching.
Moreover, we let the fractional allocation be the optimal solution of a linear program (LP) relaxation of Non-IID Online Stochastic Matching.
By rounding this fractional allocation using the general SOCS, we improve the state-of-the-art of these two problems in Non-IID Online Stochastic Matching and the Query-Commit model.
For unweighted and vertex-weighted Non-IID Online Stochastic Matching, we get a $0.69$-competitive algorithm (Corollary~\ref{cor:vertex-weighted-matching}), improving the existing $0.666$-competitive algorithm~\cite{TangWW:STOC:2022}.
Further, we achieve a better $0.705$ competitive ratio in the \emph{Random-Order} model (Corollary~\ref{cor:vertex-weighted-matching-random-order}).
Through a known reduction~\cite{CostelloTT:ICALP:2012, GamlathKS:SODA:2019} (see also Appendix~\ref{app:lossless-simulation}), we get an algorithm that obtains the same competitive ratio $0.705$ in the Query-Commit model (Corollary~\ref{cor:vertex-weighted-matching-query-commit}), which improves the previous $0.696$-competitive and $0.662$-competitive algorithms for unweighted~\cite{MahdianY:STOC:2011} and vertex-weighted matching~\cite{JinW:WINE:2021}.

\paragraph{Summary of Techniques.}
Recall that we can only exploit the power of SOCS on two-way surrogate types.
Hence, the main challenge is to upper bound the influence of one-way surrogate types.
Similar issues were handled in the IID case through an asymptotic equivalence between IID Online Stochastic Matching and the (homogeneous) Poisson arrival model~\cite{JailletL:MOR:2014, HuangS:STOC:2021, HuangSY:STOC:2022}.
If there was such an asymptotic equivalence in our problem, the expected allocation of one-way surrogate types to an offline vertex would be at most $1-\ln 2 < 0.307$, according to the Converse Jensen Inequality~\cite{HuangS:STOC:2021} (more precisely, its straightforward generalization to the non-homogeneous Poisson arrival model).
See Subsection~\ref{sec:jl-demonstration} for a demonstration of the subsequent analysis.

Unfortunately, Non-IID Online Stochastic Matching is \emph{not} asymptotically equivalent to the (non-homogeneous) Poisson arrival model. 
This paper provides two methods to overcome the absence of asymptotic equivalence.
For unweighted and vertex-weighted matching, we modify the analysis of SOCS to explicitly capture the influence of one-way types in the recurrence. 
Further, we manage to transform the recurrence in both the Non-IID model and its random-order relaxation, so that resulting coefficients are related to the left-hand-side of the following LP constraints:
\begin{equation}
	\label{eqn:intro-matching-constraint}	
	\prod_{t \le t'} \bigg( 1 - \sum_{i \,:\, \matchrate^t_{ij} > \frac{1}{2}} \mass_i^t \bigg) \le 1 - \sum_{t \le t'} \sum_{i \,:\, \matchrate^t_{ij} > \frac{1}{2}} \mass^t_i \cdot \matchrate_{ij}^t
	~.
\end{equation}
Here, $\mass^t_i$ is the probability of realizing an online vertex of type $i$ at time  $t$, and $\matchrate^t_{ij}$ is the fractional (LP) allocation of the online vertex at time $t$ to offline vertex $j$, when the online type is $i$.
Hence, the left-hand-side is the probability that, up to time $t'$, we do not have any online types $i$ more than half of which are allocated to $j$ by the LP solution.
The inequality holds because the right-hand-side is the probability of \emph{not} matching $j$ to such online types.
Recall that these are the only online types for which the type decomposition would sample a one-way surrogate type with a positive probability.
After bounding the coefficients by the LP constraints, we get the stated competitive ratios by solving the relaxed recurrence.
See Subsections~\ref{sec:unweighted-socs} and \ref{sec:unweighted-ro} for the detailed analyses.

\subsection{Results and Techniques: AdWords}

\paragraph{Breaking the $\bm{1-\frac{1}{e}}$ Barrier in Stochastic Model.}
We give a two-way SOCS for AdWords with a convergence rate $g(y_j)$ strictly better than the baseline $e^{-y_j}$ (Theorem~\ref{thm:two-way-socs-adwords}).
Further, we introduce a novel LP relaxation for the stochastic model of AdWords.
By letting the fractional allocation be the optimal LP solution and rounding it through type decomposition and the two-way SOCS, we get an algorithm for Non-IID Stochastic AdWords that is better than $1-\frac{1}{e}$ competitive.

As mentioned earlier, it was unknown before this paper how to break the $1-\frac{1}{e}$ barrier even in the IID model, which has been open for at least a decade since Mehta's survey~\cite{Mehta:FTTCS:2013}.
Following the framework of SOCS, we directly break the barrier in the more general Non-IID model.

\paragraph{Summary of Techniques.}
It is well known that the allocation of larger bids is the crux of the AdWords problem.
For example, when agents' bids for any item are at most half their budgets, \citet*{DevanurSA:EC:2012} already gave an online algorithm that is $0.73 > 1 - \frac{1}{e}$ competitive for IID Stochastic AdWords.
Similarly, we observe that independent rounding already achieves a convergence rate strictly better than $e^{-y_j}$ in the small-bid regime.
Hence, we only need to introduce negative correlations among the selections of large bids. 

We use a different threshold based on our analysis to define large and small bids.
Consider an agent $j$'s bid for an online item of type $i$, denoted as $b_{ij}$.
We say that it is small if $b_{ij} \le \frac{2}{3} B_j$, and is large if $b_{ij} > \frac{2}{3} B_j$.
Further, we say that a two-way surrogate type gets a large bid from agent $j$ if (1) agent $j$ is one of the two choices and (2) agent $j$'s bid for this type of item is large.

This paper gives a simple two-way algorithm that merely injects a mild amount of negative correlation to the decisions regarding large bids, based on an idea similar to the first OCS algorithm by \citet{FahrbachHTZ:FOCS:2020}.
We consider it a proof of concept and leave it for future research to design better SOCS algorithms for AdWords.
Our two-way SOCS works as follows (see Subsection~\ref{sec:two-way-socs-adwords} for the formal definition and its analysis):

\begin{itemize}
\item If the two-way surrogate type at time $t$ gets a large bid from agent $j$, mark the time step with $j$ with probability half (reserving the other half for the other agent in the two-way type).
\item For any offline agent $j$, the two-way SOCS makes the opposite selections randomly in the first two time steps marked with $j$, i.e., it either selects $j$ in the first such time step, and selects the choice other than $j$ in the second, or the other way around, each with probability half.
\item In the other time steps, the algorithm selects independently and uniformly at random.
\end{itemize}

To utilize this two-way SOCS, we again face the challenge of upper bounding the influence of one-way surrogate types, in particular, those that get large bids for agent $j$.
The existing LP for Stochastic AdWords~\cite{DevanurSA:EC:2012} is insufficient because it lacks constraints that correspond to Equation~\eqref{eqn:intro-matching-constraint} for unweighted and vertex-weighted matching.
By contrast, we introduce a new LP by including the following constraints: 
\begin{equation}
	\label{eqn:intro-adwords-constraint}
	\sum_t \sum_{i \in L} \,\mass^t_i \cdot \matchrate_{ij}^t \cdot b_{ij}
	~\le~
	\E \bigg[ \min \bigg\{ \sum_t \sum_{i \in L} \,X^t_i \cdot b_{ij} \,,\, B_j \bigg\} \bigg]
	~.
\end{equation}

Here, $L$ is a subset of online types, intuitively those that get large bids from agent $j$.
Further, $X^t_i$ is the indicator that the online vertex realized at time $t$ is of type $i$.
The left-hand-side is the agent $j$'s expected budget spent on items of types $i \in L$ according to the fractional allocation $\matchrate^t$'s;
the right-hand-side is the maximum budget that could be spent on these items, even if we allocated all realized items of these types to agent $j$.
If we were in the non-homogeneous Poisson arrival model, then these constraints would give a Converse Jensen Inequality similar to the counterpart for matching.
In particular, it would imply that the one-way surrogate types that are large bids for agent $j$ can contribute at most $0.36$ to $y_j$.

To break the $1-\frac{1}{e}$ barrier in the Non-IID model without an asymptotic equivalence to the (non-homogeneous) Poisson arrival model, we prove an approximate Converse Jensen Inequality for the Non-IID model, which may be of independent interest and find further applications in the Non-IID stochastic models of other online algorithms.
We need this alternative approach because we cannot apply the previous method for unweighted and vertex-weighted matching, which relies on a clean recurrence that we cannot replicate in the more complicated AdWords problem. 
See Subsection~\ref{sec:adwords-general-socs} for details.

\paragraph{Multi-Way OCS and Improved Algorithm in Adversarial Model.}
Note that the above two-way SOCS for AdWords does not rely on distributional information. 
As a result, we are able to further use it to get the first multi-way OCS for AdWords (Theorem~\ref{thm:ocs-adwords}), where the fractional allocation $\matchrate^t$'s are chosen by an adversary.
Given the fractional allocation $\matchrate^t$ at time step $t$, we use the type decomposition to sample a one-way or two-way surrogate type, and then select an offline vertex using the two-way SOCS.
As a corollary, we get a $0.504$-competitive algorithm for AdWords in the adversarial model (Corollary~\ref{cor:adwords}), improving the previous $0.501$-competitive algorithm~\cite{HuangZZ:FOCS:2020}.

Our multi-way OCS for AdWords achieves a convergence rate strictly better than $e^{-y_j}$ in the adversarial model.
This guarantee is simpler and more direct than the (two-way) Panorama OCS by \citet{HuangZZ:FOCS:2020}.
We believe the simpler unified metric for OCS/SOCS by their convergence rates will lead to further improvements for AdWords in both stochastic and adversarial models.

\paragraph{Summary of Techniques.}
The main ingredient, beyond what we have already explained in the stochastic model, is a new argument for bounding the influence of one-way surrogate types that are large bids, tailored for the adversarial model.
Consider a time step $t$ in which (1) agent $j$'s bid is large, and (2) with a positive probability the type decomposition samples a one-way surrogate type with $j$ as the only choice.
Recall that our type decomposition would sample such a one-way type only when $\matchrate^t$ allocates more than half of the item to agent $j$, i.e., $\matchrate^t_j \in (\frac{1}{2}, 1]$.
The intuition behind the new argument is best demonstrated on the two extremes, when $\matchrate^t_j$ is close to either $\frac{1}{2}$ or $1$.
In the former case, the one-way surrogate type $j$ in this step only makes a negligible contribution compared to the two-way surrogate types, which also include $j$ as one of the two choices.
Hence, we get sufficient improvement from the two-way SOCS.
In the latter case, we allocate a large bid to $j$ and consume at least $\frac{2}{3}$ of its budget, \emph{almost with certainty}.%
\footnote{The (almost) certainty is critical and the reason why the same argument does not work in the stochastic model.}
In other words, agent $j$'s expected value is sufficiently large due to this step alone.
More efforts are needed to extend the argument to all cases of $\matchrate^t_j$;
see Subsection~\ref{sec:adwords-ocs} for details.

\subsection{Results and Techniques: Display Ads}

\paragraph{Breaking the $\bm{1-\frac{1}{e}}$ Barrier in Non-IID Model.}
We consider a two-way SOCS for Display Ads similar to the counterpart for AdWords:
each time step is marked with one of the two agents uniformly at random, and for each agent $j$ the algorithm makes the opposite selections in the first two time steps marked with agent $j$.
Our analysis gives a convergences rate that is equal to the baseline $e^{-y_j}$ when $0 \le y_j \le \theta$ for some threshold $\theta < 1$, and is strictly better when $\theta < y_j \le 1$ (Theorem~\ref{thm:display-ads-two-way}).
This is sufficient for solving the Non-IID Stochastic Display Ads problem.
By letting the fractional allocation be the optimal LP solution, and rounding it through type decomposition and the two-way SOCS, we get a $0.644$ competitive online algorithm, breaking the $1-\frac{1}{e}$ barrier in the Non-IID model.

\paragraph{Summary of Techniques.}
Our results for Non-IID Stochastic Display Ads essentially follow by adopting the techniques we have developed in the other problems, and fitting them into the Display Ads problem with minor modifications.
For example, we prove an approximate Converse Jensen Inequality for the Non-IID matching model similar to the counterpart for AdWords.
It allows us to upper bound the contribution of one-way surrogate types.

Nonetheless, we need to handle the following subtlety rooted from the requirement of achieving the convergence rate at all weight-levels.
Recall that we consider an offline vertex $j$'s value by each weight-level $w > 0$, where $y_j(w)$ denotes the expected fractional allocation to offline vertex $j$ from edges whose weights are at least $w$.
Then, offline vertex $j$'s value for the fractional allocation is $\int_0^\infty y_j(w) \,\dif w$.
Correspondingly, the SOCS algorithm needs to allocate an edge with weight at least $w$ to $j$ with probability at least $1 - g\big(y_j(w)\big)$, with a convergence rate $g(\cdot)$ better than $e^{-y_j(w)}$.
However, the definition of the two-way SOCS marks time steps independent of the edge-weights therein.
As a result, the appearances of two-way types with edge-weight strictly smaller than $w$ for agent $j$ may stop the algorithm from making opposite selections in the time steps that matter, i.e., those with edge-weights at least $w$.

We resolve this issue based on the fact that the expected total allocation to any offline vertex $j$ is at most $1$, since Display Ads is a matching problem from the offline optimal solution's viewpoint.
On one hand, if $0 \le y_j(w) \le \theta$ for some threshold $\theta < 1$, then the baseline convergence rate $e^{-y_j(w)}$ is already good enough, because
the SOCS algorithm allocates an edge with weight at least $w$ to vertex $j$ with probability at least $1 - e^{-y_j(w)} \ge \frac{1 - e^{-\theta}}{\theta} y_j(w) > 0.644 \cdot y_j(w)$.
On the other hand, if $\theta < y_j(w) \le 1$, the expected allocation to agent $j$ from lower weight-levels is less than $1-\theta$.
Hence, such time steps do not appear very often.
In fact, with a constant probability, all time steps involving agent $j$ have edge-weights at least $w$.
See Section~\ref{sec:display-ads} for the detailed argument.

\subsection{Future Directions}

We initiate the study of SOCS, a family of online rounding algorithms for various stochastic online matching problems.
As the first paper on the topic, we often opt for simpler algorithms as proofs of concept, rather than pushing for the optimal algorithms. 
Hence, we believe there is plenty of room for further improvements.
The obvious future direction is to design better algorithms to get improved convergence rates and competitive ratios in these problems.
Beyond that, we further outline several future directions below that are conceptually interesting.

\paragraph{Genuine Multi-Way SOCS.}
This paper reduces the design of multi-way SOCS algorithms to that of two-way SOCS via type decomposition.
On one hand,  this simple method is surprisingly powerful and already capable of making progress on many problems including two open questions related to AdWords.
On the other hand, this is intrinsically wasteful.
This is most apparent in unweighted and vertex-weighted matching.
Even if an offline vertex $j$ has already been matched and thus could make no further contribution to the objective, the type decomposition would still sample a one-way or two-way surrogate type involving $j$ with a positive probability.
Ultimately, we would like to have genuine multi-way SOCS algorithms that do not rely on this intrinsically wasteful type decomposition.
For example, can one obtain such a multi-way SOCS for unweighted and vertex-weighted matching, by combining the ideas behind the Poisson OCS~\cite{HuangSY:STOC:2022}, which may be viewed as a SOCS for the IID case, and the new argument based on AM-GM in this paper?

\paragraph{OCS for Display Ads and Convergence Rate.}
Combining the results from this paper and those by \citet{GaoHHNYZ:FOCS:2021}, we can now measure the performance of OCS and SOCS in almost all settings under the unified metric of convergence rate.
The only exception is the OCS for Display Ads, for which we still need the more complicated metrics based on the concept of consecutive steps.
We conjecture the existence of multi-way OCS algorithms for Display Ads with convergence rates strictly better than the baseline $e^{-y}$.
Here, we can either follow the definition of convergence rate for SOCS and require an OCS algorithm to allocate an edge with weight at least $w$ to offline vertex $j$ with probability at least $1 - g\big(y_j(w)\big)$, or allow amortization across different weight-levels and only require the expected maximum edge-weight allocated to $j$ to be at least:
\[
	\int_0^\infty \big( 1 - g(y_j(w)) \big) \,\dif w
	~.
\]

In either case, we believe that developing OCS algorithms under this cleaner and more unified metric will advance our understanding of the Display Ads problem and lead to online algorithms with better performance.

\paragraph{SOCS/OCS for Online Submodular Welfare Maximization.}
Last but not least, it would be very interesting to explore SOCS or OCS algorithms for the more general Online Submodular Welfare Maximization problem.
The impossibility results by \citet*{KapralovPV:SODA:2013} are computational hardness (of Maximum Coverage) rather than information theoretic. 
In particular, the impossibility result for the stochastic model may be interpreted as the hardness of computing the fractional allocation $\matchrate^t$'s considered in this paper, and does not rule out the possibility of SOCS algorithms.
Another possible approach to circumvent the computational hardness and focus on the online decision-making aspect of this general problem is to assume access to an oracle that would enable us to solve the Maximum Coverage problem, e.g., one that can evaluate the concave closure of the agents' value functions.
Finally, we remark that the special case of coverage function is a particularly interesting frontier, because (1) the agents' value function in AdWords and Display Ads can both be viewed as a special form of coverage function (e.g., \cite{FahrbachHTZ:FOCS:2020, HuangZZ:FOCS:2020}), and (2) the hardness results by \citet{KapralovPV:SODA:2013} already hold for general coverage functions.

\subsection{Related Works}

\paragraph{Online Bipartite Matching.}
Following the seminal work by \citet*{KarpVV:STOC:1990}, there is a vast literature on online bipartite matching problems with an adversarially chosen graph and arrival order of online vertices.
\citet*{KarpVV:STOC:1990} and \citet*{AggarwalGKM:SODA:2011} gave optimal $1-\frac{1}{e}$ competitive algorithms for unweighted and vertex-weighted matching respectively.
\citet{FeldmanKMMP:WINE:2009} proposed the Display Ads problem, a.k.a., edge-weighted online matching with free disposal, and gave a $1-\frac{1}{e}$ competitive ratio under a large-market assumption.
\citet*{FahrbachHTZ:FOCS:2020} introduced the OCS technique and broke the $\frac{1}{2}$ barrier without the large-market assumption.
The OCS technique was then improved in a series of follow-up papers~\cite{ShinA:ISAAC:2021, GaoHHNYZ:FOCS:2021, BlancC:FOCS:2021}, leading to the state-of-the-art ratio $0.536$ by \citet{BlancC:FOCS:2021}.
The \emph{AdWords} problem was proposed by \citet{MehtaSVV:JACM:2007}, who also gave a $1-\frac{1}{e}$ competitive ratio under a large-market assumption.
\citet*{HuangZZ:FOCS:2020} used OCS to break the $\frac{1}{2}$ barrier without assumption.
This paper mainly studies the stochastic versions of these problems and their OCS, but also contributes to the AdWords problem in the adversarial model.

For the problem with submodular functions, known as \emph{Online Submodular Welfare Maximization}, \citet*{KapralovPV:SODA:2013} proved that no polynomial-time online algorithm can be better than $\frac{1}{2}$ competitive in the adversarial model, or better than $1-\frac{1}{e}$ in the stochastic model.

\paragraph{Online Stochastic Matching.}
\citet*{FeldmanMMM:FOCS:2009} introduced the problem and showed an unweighted matching algorithm with a competitive ratio better than $1-\frac{1}{e}$ based on the power of two choices.
This techniques was refined in a series of subsequent works \cite{ManshadiOS:MOR:2012, BahmaniK:ESA:2010, JailletL:MOR:2014, HuangS:STOC:2021, Yan:SODA:2024, FengQWZ:WINE:2023}, leading to $0.711$ and $0.7$ competitive two-choice algorithms for unweighted and vertex-weighted matching~\cite{HuangS:STOC:2021}, and a $0.65$ competitive two-choice algorithm for edge-weighted matching (without free disposal)~\cite{FengQWZ:WINE:2023}.

The recent literature further explored the power of multiple choices.
\citet*{HuangSY:STOC:2022} gave multi-choice algorithms with the state-of-the-art $0.716$ competitive ratio for unweighted and vertex-weighted matching, and $0.706$ competitive ratio for Stochastic Display Ads.
\citet*{TangWW:STOC:2022} introduced the non-IID model, and broke the $1-\frac{1}{e}$ barrier using an OCS algorithm by \citet{GaoHHNYZ:FOCS:2021}.
While in principle multi-choice algorithms are more general, and hence, more powerful than two-choice algorithms, they are also much harder to design effectively.
There may be a lasting competition between the two techniques.
This paper reduces multi-choice algorithms to two-choice algorithms via type decomposition. 
The resulting algorithms give a better competitive ratio than the existing multi-choice algorithm~\cite{TangWW:STOC:2022} in non-IID unweighted and vertex-weighted matching, and break the $1-\frac{1}{e}$ barrier in AdWords and Display Ads.

\paragraph{Query-Commit Model.}
The original model~\cite{ChenIKMR:ICALP:2009} was motivated by kidney exchange and online dating, and had an additional constraint that we can only probe a limited number of edges incident to each vertex, known as its patience.
The best competitive ratios so far are $0.5$ for unweighted matching~\cite{Adamczyk:IPL:2011} and $0.395$ for edge-weighted matching~\cite{PollnerRSW:EC:2022}.
When the graph is bipartite and only one side has patience constraints, \citet{BorodinM:arXiv:2023} gave a $1-\frac{1}{e}$ competitive algorithm.

The model without patience constraints, which we focus on in this paper, was first studied by \citet{MolinaroR:arXiv:2011}.
\citet*{CostelloTT:ICALP:2012} gave a $0.573$ competitive algorithm for unweighted matching on general graphs, and showed that no algorithm could be better than $0.898$ competitive.
\citet*{GamlathKS:SODA:2019} considered edge-weighted bipartite matching, and proposed a $1-\frac{1}{e}$ algorithm.
The ideas behind the rounding algorithms in these two papers allow us to losslessly simulate a Random-Order Non-IID algorithm in the Query-Commit model.
\citet*{DerakhshanF:SODA:2023} recently proposed a different rounding algorithm to improve the ratio for edge-weighted bipartite matching to $1-\frac{1}{e}+0.0014$.
For edge-weighted matching in general graphs, \citet{FuTWWZ:ICALP:2021} obtained a $8/15 \approx 0.533$-competitive algorithm by studying Random-Order Contention Resolution Schemes.
For unweighted and vertex-weighted bipartite matching, the best algorithms are from the random-order model of online bipartite matching~\cite{MahdianY:STOC:2011, KarandeMT:STOC:2011, HuangTWZ:TALG:2019, JinW:WINE:2021}.
The best competitive ratios are $0.696$~\cite{MahdianY:STOC:2011} and $0.662$~\cite{JinW:WINE:2021} respectively.

\paragraph{Edge-Weighted Online Stochastic Matching.}
Last but not least, the edge-weighted version of Online Stochastic Matching (without free disposal) has also been extensively studied. 
It may be viewed as a generalization of the classical Prophet Inequality, whose optimal competitive ratio is $0.5$~\cite{Samuel:AnnaProb:1984}.
\citet*{FeldmanGL:SODA:2014} studied combinatorial auctions via posted prices and their results imply a $0.5$ competitive algorithm for the edge-weighted matching problem.
\citet*{PapadimitriouPSW:EC:2021} considered comparing a polynomial-time online algorithm's performance to the optimal (exponential-time) online algorithm, and gave a $0.51$-competitive polynomial-time algorithm.
This was later improved to $0.527$ by \citet{SaberiW:ICALP:2021}, to $1-\frac{1}{e}$ by \citet*{BravermanDM:EC:2022}, and to $0.652$ by \citet*{NaorSW:arxiv:2023}.

\section{Preliminaries}
\label{sec:prelim}

\paragraph{Notations.}
We write $z^+$ for function $\max \{ z, 0 \}$.
Let $[n] = \{1, 2, \dots, n\}$ for any positive integer $n$.
For any set $S$ and any element $e$, we write $S + e$ for $S \cup \{e\}$ and $S - e$ for $S \setminus \{e\}$.
For any variable $x^t$ indexed in the superscript by a time step $t \in [T]$, and any subset of time steps  $S \subseteq [T]$, we write $x^S$ for $\sum_{t \in S} x^t$.
Further, we write $t_1:t_2$ for $\{t_1, t_1+1, \dots, t_2\}$ and thus $x^{t_1:t_2} = \sum_{t=t_1}^{t_2} x^t$.

\subsection{Stochastic Online Submodular Welfare Maximization}

Consider a set of \emph{online item types} $I$ and a set of \emph{offline agents} $J$.
Each agent $j \in J$ has a value function $v_j : 2^I \to [0, \infty)$ over subsets of online item types satisfying:
\begin{itemize}
\item $v_j(S) \ge 0$ for any $S \subseteq I$;
	\hfill (\emph{Non-negativity}) 
\item $v_j(S) \ge v_j(S')$ for any $S' \subset S \subseteq I$; and 
	\hfill (\emph{Monotonicity})
\item $v_j(S) + v_j(S') \ge v_j(S \cup S') + v_j(S \cap S')$ for any $S, S' \subseteq I$.
	\hfill (\emph{Submodularity})
\end{itemize}

Further, consider $T$ discrete time steps.
In each time step $t \in [T]$, an online item arrives with its type drawn from a distribution $\Distribution^t$ over the online types $I$, independent to the realization of previous items' types.
We will write $\mass_i^t$ for the probability that an online item of type $i$ arrives at step $t$ according to distribution $\Distribution^t$.
When each online item arrives, the online algorithm must allocate it to an offline agent immediately.
The goal is to maximize the \emph{social welfare}, i.e., the sum of the agents' values for the subsets of items allocated to them.

We make two remarks about this model.
First, it is general enough to capture the possibility that with a positive probability no online vertex arrives at time step $t$, e.g., by having a dummy online item type that contributes zero to the agents' value functions.
Second, multiple items of the same type $i \in I$ may arrive at different time steps in this model, and may further be allocated to the same agent $j$, even though agent $j$'s value function is defined on subsets of $I$ rather than multi-subsets.
We follow the treatment that the second item of an online type contributes zero to an agent's value. 
Nonetheless, the Non-IID model can also capture scenarios in which the agents have positive values for additional items of the same online type, e.g., by renaming online type $i$ at time step $t$ as type $(i,t)$ and thus, making the supports of distributions $\Distribution^t$'s disjoint.

Given any online algorithm, we let $\ALG$ denote the algorithm's expected objective value, over the random realization of the online items' types, and the algorithm's intrinsic randomness. 

Following the standard competitive analysis, we will compare an online algorithm's objective to the expectation of the offline optimal social welfare in hindsight, denoted as $\OPT$.
In other words, $\OPT$ is the expectation of the best achievable welfare if we had full information of the realized item types, and computed the best allocation accordingly with unlimited computational power.

An online algorithm is $\Gamma$-competitive if it guarantees $\ALG \ge \Gamma \cdot \OPT$ for all instances of the problem.
In other words, the algorithm would achieve at least a $\Gamma$ fraction of the expected optimal social welfare, even if the problem instance, including the offline agents, the online item types, the agents' value functions, and the distributions of online item types at different time steps (but not their realization), were chosen by an adversary who knows the algorithm (but not the realization of its internal random bits).

\clearpage
\subsection{Special Cases}

This paper focuses on the following four special cases that are widely studied in the literature.

\paragraph{Unweighted Matching.}
Consider a bipartite \emph{type graph} $G = (I, J, E)$.
The online item types are vertices on the left.
The offline agents are vertices on the right.
In this problem and the next, we will use online items and online vertices interchangeably, and similarly use offline agents and offline vertices interchangeably.
Further, $E$ denotes the set of edges between online vertex types $I$ and offline vertices $J$.
This problem considers maximizing the cardinality of the matching.
Hence, it is the special case when:
\[
	v_j(S) =
	\begin{cases}
		1 & \mbox{if there is $i \in S$ such that $(i, j) \in E$;} \\
		0 & \mbox{otherwise.}	
	\end{cases}
\]

Here, we interpret the allocation of a subset of  (types of) online vertices $S$ to $j$ as matching to $j$ the first online vertex therein whose type $i$ is $j$'s neighbor in $G$.
We will refer to this special case as \emph{Online Stochastic Matching}.

\paragraph{Vertex-Weighted Matching.}
This generalizes unweighted matching by associating each offline vertex $j \in J$ with a positive vertex-weight $w_j > 0$.
Instead of maximizing the cardinality, we now want to maximize the sum of the \emph{matched} offline vertices' weights. 
Hence, we have:
\[
	v_j(S) =
	\begin{cases}
		w_j & \mbox{if there is $i \in S$ such that $(i, j) \in E$;} \\
		0 & \mbox{otherwise.}	
	\end{cases}
\]

We will refer to this special case as \emph{Vertex-Weighted Online Stochastic Matching}.

\paragraph{AdWords.}
This problem considers an online advertising platform, where the online items are the impressions and the offline agents are the advertisers. 
For each online item type $i \in I$ and each agent $j \in J$, consider a non-negative bid $b_{ij} \ge 0$ that represents agent $j$'s willingness-to-pay for an item of type $i$.
Further, each agent $j$ has a positive budget $B_j > 0$ that upper bounds its total payment.
Hence, agent $j$'s value function is 
\[
	v_j(S) = \min \bigg\{\, \sum_{i \in S} b_{ij} \,,\, B_j \,\bigg\}
	~,
\]
i.e., either the sum of the allocated bids or its budget $B_j$, whichever is smaller.
We will refer to this special case as \emph{Stochastic AdWords}.

\paragraph{Display Ads.}
This problem also considers an online advertising platform.
Instead of setting a budget for its payment, each agent $j$ will only pay for the most valuable item allocated to it.
Following the terminology in previous works on this problem, we refer to agent $j$'s willingness-to-pay for an item of type $i$ as its edge-weight $w_{ij} \ge 0$.
This is the special case when:
\[
	v_j(S) = \max_{i \in S} \, w_{ij}
	~.
\]

This problem is also known as Edge-Weighted Online Bipartite Matching with Free Disposal, because the algorithm computes an edge-weighted matching where each agent is matched to the most valuable item allocated to it.
Allocating multiple items to an agent while keeping the maximum edge-weight is equivalent to allowing the agent to dispose previously allocated but lighter edges for free.
We will refer to this special case as \emph{Stochastic Display Ads}.

\subsection{Random-Order and Query-Commit Models}

For unweighted and vertex-weighted matching, we will further consider two relaxed models.

\paragraph{Random-Order Non-IID Model.}

In this model, the type graph and the distributions $\Distribution(t)$'s are still adversarially chosen, but the time steps are shuffled uniformly at random.
That is, consider a random permutation $\pi$ of $[T]$.
At each step $t \in [T]$, the algorithm observes $\pi(t)$ and an online item with its type drawn from $D^{\pi(t)}$, and needs to immediately allocate it to an agent.
We consider the algorithm's expected objective value over the random realization of the permutation and the online vertex types, and the algorithm's internal random bits as well. 
This kind of model has been studied for its own merit (see, e.g.,  the Prophet Secretary problem~\cite{EsfandiariHLM:SIDMA:2017}). 
In this paper, it will be a stepping stone toward designing online algorithms for the next model.

\paragraph{Query-Commit Model.}
Consider the following model of the unweighted and vertex-weighted matching problem.
Consider a bipartite graph $G = (I, J, E)$.
Each edge $(i,j) \in E$ exists independently with probability $0 \le p_{ij}\le 1$.
In the beginning, the algorithm knows the probabilities but not the realization of edges.
In each time step, the algorithm may query an edge $(i,j)$, and if $(i,j)$ exists then the algorithm must commit to including it in the matching.
We remark that the two sides of the bipartite graph have equal roles in this problem.
Nevertheless, we can artificially treat the two sides as offline and online vertices respectively in order to leverage algorithms from the previous Random-Order model.

This paper will consider a seemingly different but equivalent model in which the algorithm just needs to immediately decide whether to include the queried edge into the matching if the edge exists.
In other words, the algorithm can choose not to include it.
To simulate the option of not including an edge $(i,j)$ after querying it in the original model, the algorithm can flip a coin itself (instead of probing nature's coin flip), and proceed as if the edge existed with probability $p_{ij}$;
otherwise, the algorithm can proceed as if the edge did not exist.

Appendix~\ref{app:lossless-simulation} explains how to losslessly simulate an online algorithm designed for the Random-Order model in the Query-Commit model, using a known reduction~\cite{CostelloTT:ICALP:2012, GamlathKS:SODA:2019}.

\subsection{Existing Linear Program Relaxations}
\label{sec:existing-lp}

\subsubsection*{Stochastic Matching Linear Program}

Recall that $\mass_i^t$ is the probability that an online vertex of type $i$ arrives at time step $t$ according to distribution $\Distribution^t$.
We have $\sum_{i \in I} \mass_i^t \le 1$ for any time step $t \in [T]$.
Let $x_{ij}^t = \mass_i^t \cdot \matchrate^t_{ij}$ denote the probability that the offline optimal solution matches an online vertex of type $i$ to offline vertex $j$ at time step $t$.
With these variables, consider the following linear program by \citet{GamlathKS:SODA:2019}, which we will refer to as the \emph{Stochastic Matching LP}:
\begin{align}
    \text{maximize} \quad & \sum_{i \in I} \sum_{j \in J} \sum_{t \in [T]} w_{ij} \cdot x_{ij}^t \notag \\
    \text{subject to} \quad & \sum_{j \in J} x_{ij}^t \le \mass_i^t && \forall i \in I, \forall t \in [T] \notag \\
    & \sum_{(t,i) \in S} x_{ij}^t \le 1 - \prod_t \Big( 1 - \sum_{i : (t,i) \in S} \mass_i^t \Big) && \forall j \in J, \forall S \subseteq T \times I \label{eqn:discrete-time-lp} \\
    & x_{ij}^t \ge 0 && \forall i \in I, \forall j \in J, \forall t \in [T] \notag 
\end{align}

We will apply this LP to the Non-IID stochastic models for unweighted matching, vertex-weighted matching, and Display Ads. 
Unweighted matching is the special case when $w_{ij} = 1$ for edges $(i,j) \in E$ and $w_{ij} = 0$ otherwise.
Vertex-weighted matching is the special case when there are offline-vertex weights $(w_j)_{j \in J}$ such that $w_{ij} = w_j$ for edges $(i,j) \in E$ and $w_{ij} = 0$ otherwise.

The first set of constraints states that the probability of matching an online vertex of type $i$ at time step $t$ cannot exceed the probability that such a vertex arrives.
The second set of constraints says that the probability of matching an offline vertex $j$ to an online vertex with arrival time $t$ and type $i$ such that $(t,i) \in S$ is upper bounded by the probability of having at least one such online vertex arrive in the first place, which equals right-hand-side.

\begin{lemma}[Optimality, e.g., \cite{GamlathKS:SODA:2019}]
    \label{lem:discrete-time-lp-optimality}
    The optimal objective value of the Stochastic Matching LP is greater than or equal to the expected objective of the optimal matching in hindsight.
\end{lemma}

\begin{lemma}[Computational Efficiency, e.g., \cite{GamlathKS:SODA:2019}]
    \label{lem:discrete-time-lp-polytime}
    The Stochastic Matching LP is solvable within polynomial time.
\end{lemma}

We include its short proof below because the argument is insightful and will be useful for proving a similar lemma for our LP for Stochastic AdWords.

\begin{proof}
    The second set of constraints for any fixed $j \in J$ forms a polymatroid because the right-hand-side is a submodular set function over $T \times I$.
    Hence, the polytope of the Stochastic Matching LP is the intersection of polynomially many linear constraints (the first set of constraints) and $|J|$ polymatroids (the second set of constraints).
    Therefore, we have a polynomial-time separation oracle for it and can solve it in polynomial time using the ellipsoid method.
\end{proof}

\subsubsection*{Fluid Stochastic AdWords Linear Program}

\citet{DevanurSA:EC:2012} proposed a linear program for the IID special case of the Stochastic AdWords problem.
Its natural generalization to the Non-IID case is as follows:
\begin{align}
    \text{maximize} \quad & \sum_{i \in I} \sum_{j \in J} \sum_{t \in [T]} b_{ij} \cdot x_{ij}^t \notag \\
    \text{subject to} \quad & \sum_{j \in J} x_{ij}^t \le \mass_i^t && \forall i \in I, \forall t \in [T] \notag \\
    & \sum_{i \in I} \sum_{t \in [T]} b_{ij} \cdot x_{ij}^t \le B_j && \forall j \in J \label{eqn:adwords-fluid-lp} \\
    & x_{ij}^t \ge 0 && \forall i \in I, \forall j \in J, \forall t \in [T] \notag 
\end{align}

This kind of LPs is sometimes referred to as the fluid LPs in the literature on various offline and online optimization problems.
Therefore, we will call it the \emph{Fluid Stochastic AdWords LP}.
The first set of constraints is the same as the counterpart in the Stochastic Matching LP.
However, Constraint \eqref{eqn:adwords-fluid-lp} simply states that the expected amount of agent $j$'s budget spent on the allocation cannot exceed $B_j$.
This is substantially weaker than its counterpart in the Stochastic Matching LP, i.e., Constraint~\eqref{eqn:discrete-time-lp} therein, and is insufficient for our analysis.
We will introduce a new LP by strengthening this constraint in Subsection~\ref{sec:adwords-lp}.

\section{Stochastic Online Correlated Selection}
\label{sec:socs}

\subsection{Model}

Stochastic Online Correlated Selection (SOCS) is an online rounding algorithm for the Non-IID Stochastic Online Submodular Welfare Maximization problem and its special cases.
We will first recall the non-IID stochastic model.
Consider a set of online item types $I$ and a set of offline agents $J$.
Each offline agent $j \in J$ has a non-negative, monotone (non-decreasing), and submodular value function $v_j$ over subsets of $I$.
Consider $T$ discrete time steps.
In each time step $t \in [T]$, an online item arrives with its type drawn from distribution $\Distribution^t$.

In the model of SOCS, the online rounding algorithm is further given a fractional allocation $\matchrate_i^t = \big( \matchrate_{ij}^t \big)_{j \in J}$ at each time step $t \in [T]$ and for any online item type $i \in I$.
We will consider $\matchrate_{ij}^t$ as the probability that an online item of type $i$ arriving at time $t$ would be allocated to offline agent $j$ by some feasible offline allocation rule, over the random realization of the other online vertices' types.
To obtain a competitive online algorithm using SOCS, ideally we would like to use the optimal offline allocation to define $\mu^t_i$.
For computational efficiency, however, we will relax it to any fractional allocation such that $x_{ij}^t = \mass_i^t \cdot \matchrate_{ij}^t$ satisfies the linear constraints in either the LP relaxation in Subsection~\ref{sec:existing-lp}, and those in a new LP that we will develop in Subsection~\ref{sec:adwords-lp}.

Upon observing the online item's type $i$ at time $t$ and the fractional allocation $\matchrate^t_i$, the SOCS algorithm needs to immediately select an offline agent $j \in J$ and allocate the item to it, based on the fractional allocation $\matchrate_i^t$ and the selections in the previous time steps.

\subsubsection*{Special Case: Two-Way SOCS}

We will study the two-way SOCS as an important special case, where the fractional allocation for any time step $t \in [T]$ and any online item type $i \in I$ is half-integer, i.e., $\matchrate_{ij}^t \in \{0, \frac{1}{2}\}$ for all $j \in J$.
Effectively, the fractional allocation shortlists a pair of offline agents $J^t = \{j, k\}$ with no preference between the two, and the two-way SOCS needs to select one of them and allocate the online item to the selected agent.
We remark that this model is general enough to allow having no pair arrive in some time step $t \in [T]$ with a positive probability, e.g., by introducing dummy offline agents whose values are always zero, and interpreting this case as having a pair of dummies.

\subsubsection*{Convergence Rate}

Intuitively, we want to measure the quality of SOCS algorithms by the expectation of each offline agent $j$'s value function $v_j$ for the online items allocated to agent $j$, comparing against agent $j$'s contribution to the fractional allocation's objective.
Quantitatively, we will capture this by a non-increasing function $g : [0, 1] \to [0, 1]$ with $g(0) = 1$, which we will refer to as the \emph{convergence rate} of SOCS.
We will next elaborate on the guarantee of SOCS in the four special cases of Stochastic Online Submodular Welfare Maximization that we will study in this paper.

\paragraph{Unweighted and Vertex-Weighted Online Stochastic Matching.}
An offline agent/vertex $j$'s contribution to the fractional allocation's objective, normalized by its vertex-weight $w_j$, is:
\[
    y_j \: = \:  \sum_{t=1}^T \sum_{i \in I} \: \mass_i^t \cdot \matchrate_{ij}^t
    ~.
\]

A SOCS algorithm has convergence rate $g(\cdot)$ for unweighted and vertex-weighted matching if for any offline vertex $j$, the probability that vertex $j$ stays unmatched at the end is at most $g(y_j)$, over the random realization of online vertices' types and the internal randomness of the SOCS.

For notational convenience in the subsequent analysis, we further let $y^t_j = \sum_{i \in I} \mass_i^t \cdot \matchrate_{ij}^t$ denote the contribution from a time step $t \in [T]$.
We will also consider similar auxiliary notations in the other two problems.

\paragraph{Display Ads.}
As outlined in the introduction, we will consider an offline agent $j$'s objective by different weight-levels $w > 0$.
This has become the standard practice for designing and analyzing online algorithms for Display Ads~\cite{DevanurHKMY:TEAC:2016, FahrbachHTZ:FOCS:2020}.
For any weight-level $w > 0$, let:
\[
	y_j(w) \: = \: \sum_{t=1}^T \sum_{i \in I : w_{ij} \ge w } \mass_i^t \cdot \matchrate_{ij}^t
\]
denote the fractional allocation to offline agent $j$ from the online items types whose edge-weights are at least $w$.
Then, an offline agent $j$'s contribution to the fractional allocation's objective can be written as:
\[
    \int_0^\infty y_j(w) \:\dif w
    ~.
\]

Accordingly, we will measure the performance of SOCS for Display Ads not only for each offline agent $j \in J$ but also at each weight-level $w > 0$. 
Formally, a SOCS algorithm for Display Ads has convergence rate $g(\cdot)$ if for any offline agent $j$ and any weight-level $w > 0$, the probability that agent $j$ is \emph{not} allocated with any online item with edge-weight $w_{ij} \ge w$ is at most $g\big( y_j(w) \big)$, over the random realization of online items' types and the internal randomness of the SOCS algorithm.
This implies that the expected maximum edge-weight allocated to agent $j$ is at least:
\[
	\int_0^\infty \big( 1 - g(y_j(w)) \big) \,\dif w
	~.
\]

\paragraph{AdWords.}
An offline agent $j$'s contribution to the fractional allocation's objective, normalized by its budget $B_j$, equals:
\[
    y_j = \sum_{t=1}^T \sum_{i \in I} \: \mass_i^t \cdot \matchrate_{ij}^t \cdot \frac{b_{ij}}{B_j}
    ~,
\]
which can be viewed as the fraction of offline agent $j$'s budget used by the fractional allocation.

We say that a SOCS algorithm for AdWords has convergence rate $g(\cdot)$ if for any offline agent $j$, the allocation selected by SOCS leaves at most a $g(y_j)$ fraction of agent $j$'s budget \emph{unused} at the end in expectation, over the random realization of online items' types and the internal randomness of the SOCS algorithm.
That is, agent $j$'s expected value is at least:
\[
	\big( 1 - g(y_j) \big) \cdot B_j
	~.
\]

\paragraph{Baseline: Independent Rounding.}

For example, consider the baseline SOCS algorithm that allocates an item of type $i$ to agent $j$ with probability $\matchrate^t_{ij}$, independently at each time step $t \in [T]$.
For all aforementioned problems, this baseline algorithm's convergence rate is:
\[
    g(y_j) = e^{-y_j}
    ~,
\]

For unweighted and vertex-weighted matching, the baseline convergence rate follows by:
\[
	\prod_{t = 1}^T \Big( 1 - \sum_{i \in I} \mass_i^t \cdot \matchrate_{ij}^t \Big) 
	~\le~
	\prod_{t = 1}^T \exp \Big( - \sum_{i \in I} \mass_i^t \cdot \matchrate_{ij}^t \Big)
	~=~
	e^{-y_j}
	~.
\]

The proof for Display Ads is almost identical, except that we change the range of summation from all online types $i \in I$ to those with edge-weight at least $w$, i.e., $w_{ij} \ge w$.

Finally, the baseline convergence rate for AdWords follows by applying the above argument to auxiliary random variables that second-order stochastically dominate the actual spent budget. 
More precisely, for any offline agent $j$ and any time step $t \in [T]$, let random variable $X^t$ be the bid allocated to agent $j$ at time $t$, normalized by its budget: 
$X^t = \nicefrac{b_{ij}}{B_j}$ with probability $\mass_i^t \cdot \matchrate_{ij}^t$; and $X^t = 0$ if the agent receives no item.
Further, define auxiliary Bernoulli random variables $Y^t$ that equals $1$ with probability $X^t$, and $0$ otherwise.
By Jensen's inequality, the expected unused portion of agent $j$'s budget is:
\[
	\E\, \Big( 1 - \sum_t X^t \Big)^+ 
	~\le~
	\E\, \Big( 1 - \sum_t Y^t \Big)^+
	~=~
	\prod_{t=1}^T \Big( 1 - \sum_{i \in I} \: \mass_i^t \cdot \matchrate_{ij}^t \cdot \frac{b_{ij}}{B_j} \Big) 
	~\le~
	e^{-y_j}
	~.
\]

\subsection{General SOCS from Two-Way SOCS}
\label{sec:socs-reduction}

We will focus on a simple yet surprisingly powerful framework for designing SOCS algorithms, by reducing the problem to the two-way special case.

Upon observing the online item's type $i$ at time step $t \in [T]$, we will sample a \emph{surrogate type} based on the fractional allocation $\matchrate_i^t$.
We call this procedure \textbf{Type Decomposition} and defer its details to the end of the subsection.
Each surrogate type corresponds to either a single offline agent $j$ or a pair of offline agents $\{j, k\}$. 
We will write them as $i \sim j$ and $i \sim \{j, k\}$, and refer to them as \emph{one-way types} and \emph{two-way types} respectively.
We shall consider a one-way type $i \sim j$'s fractional allocation as fully allocating it to offline agent $j$, and a two-way online type $i \sim \{j, k\}$'s fractional allocation as allocating half of it to each of offline agents $j$ and $k$.
Then, instead of working with the original online types and the distribution $\Distribution^t$ over them at time $t$, we now have a distribution over the surrogate types.

We will now consider a SOCS algorithm that draws a surrogate type in each time step $t \in [T]$ and runs a two-way SOCS algorithm in the background as follows.
If we draw a one-way surrogate type $i \sim j$, then we have no choice but to allocate the online item to agent $j$.
Accordingly, let there be no arrival in the two-way instance for the two-way SOCS.
If we draw a two-way surrogate type $i \sim \{j, k\}$, let $\{j, k\}$ be the arriving pair for the two-way SOCS.
We will allocate the online item to the offline agent $j$ or $k$ selected by the two-way SOCS.

\bigskip

\begin{tcolorbox}
    \textbf{SOCS from Two-Way SOCS}\\[2ex]
    Upon observing the type $i$ of the online vertex at time $t$:
    \begin{enumerate}
        \item Draw a surrogate type using \textbf{Type Decomposition}.
        \item If it is a one-way type $i \sim j$, select $j$.
        \item If it is a two-way type $i \sim \{j, k\}$, select $j$ or $k$ using the two-way SOCS.
    \end{enumerate}
\end{tcolorbox}

\bigskip

We consider the following \textbf{Type Decomposition} algorithm that comes from an algorithm for (IID) Online Stochastic Matching~\cite{HuangS:STOC:2021, JailletL:MOR:2014}.
See also Figure~\ref{fig:correlated-sample} for an illustration.

\bigskip

\begin{tcolorbox}
    \textbf{Type Decomposition}\\[2ex]
    For each online type $i \in I$ and each time step $t \in [T]$:
    \begin{enumerate}
        \item Consider an interval $[0, 1)$.
        \item Align left-closed right-open subintervals $I_j$ of lengths $\matchrate_{ij}^t$ from left to right by lexicographical order of $j \in J$.\\[1ex]
            Let the unassigned interval be $I_\perp$, treating $\perp$ as a dummy agent with zero valuation.
        \item Sample $\eta \in [0, \frac{1}{2})$ uniformly at random, and let $\eta' = \eta + \frac{1}{2} \in [\frac{1}{2}, 1)$.
        \begin{enumerate}
            \item Find offline agents $j, k \in J$ such that $\eta \in I_j$ and $\eta' \in I_k$.
            \item If $j = k$, then return surrogate type $i \sim j$.
            \item Otherwise, return surrogate type $i \sim \{j,k\}$.
        \end{enumerate}
    \end{enumerate}
\end{tcolorbox}

\medskip

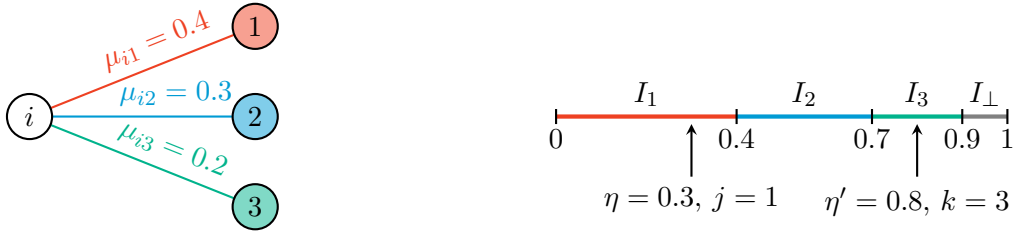
\begin{figure*}[ht]
	\centering
	\begin{tikzpicture}
		\draw(0,0)--(2.4,0)[hkured,ultra thick];
		\draw(2.4,0)--(4.2,0)[hkublue,ultra thick];
		\draw(4.2,0)--(5.4,0)[hkugreen,ultra thick];
		\draw(5.4,0)--(6,0)[gray,ultra thick];
		\draw[thick](0,-0.1)--(0,0.1);
		\draw[thick](2.4,-0.1)--(2.4,0.1);
		\draw[thick](4.2,-0.1)--(4.2,0.1);
		\draw[thick](5.4,-0.1)--(5.4,0.1);
		\draw[thick](6,-0.1)--(6,0.1);
		\draw(1.2,0.3)node{$I_1$};
		\draw(3.3,0.3)node{$I_2$};
        \draw(4.8,0.3)node{$I_3$};
        \draw(5.7,0.3)node{$I_\perp$};
        \draw(0,-0.3)node{$0$};
        \draw(2.4,-0.3)node{$0.4$};
        \draw(4.2,-0.3)node{$0.7$};		
		\draw(5.4,-0.3)node{$0.9$};
        \draw(6.0,-0.3)node{$1$};
		\draw(-7,0)--(-4,1.2)[hkured,thick] node[pos=0.6,above,sloped] {$\matchrate_{i1}=0.4$};;
		\draw(-7,0)--(-4,0)[hkublue,thick] node[pos=0.65,above,sloped] {$\matchrate_{i2}=0.3$};
		\draw(-7,0)--(-4,-1.2)[hkugreen,thick] node[pos=0.6,above,sloped] {$\matchrate_{i3}=0.2$};
		\filldraw[fill=white,thick](-7,0)circle(0.3)node{$i$};
		\filldraw[fill=hkured!50,thick](-4,1.2)circle(0.3)node{$1$};
		\filldraw[fill=hkublue!50,thick](-4,0)circle(0.3)node{$2$};
		\filldraw[fill=hkugreen!50,thick](-4,-1.2)circle(0.3)node{$3$};
        \draw[-stealth,thick] (1.8,-0.8)--(1.8,-0.1) node[at start,below] {$\eta=0.3$, $j=1$};
        \draw[-stealth,thick] (4.8,-0.8)--(4.8,-0.1) node[at start,below] {$\eta'=0.8$, $k=3$};
	\end{tikzpicture}
	\caption{Illustration of decomposing an online type $i\in I$ with neighbors $1,2,3$.}
	\label{fig:correlated-sample}
\end{figure*}

We design the algorithm based on two factors.
First, we prefer two-way surrogate types over one-way counterparts, because the former allows us to exploit the power of the two-way SOCS.
The next lemma indicates that a one-way type $i \sim j$ would be realized only when the fractional allocation at time step $t$ allocates more than half of the original online type $i$ to offline agent $j$, i.e., if $\matchrate_{ij}^t > \frac{1}{2}$.
This also means that at most one one-way surrogate type could be realized with a positive probability in each time step.
We omit this lemma's proof because it follows directly from the definition of the algorithm.

\begin{lemma}
	\label{lem:type-decomposition-probability}
	For any online item type $i \in I$, any offline agent $j \in J$, and any time step $t \in [T]$, \textbf{Type Decomposition} draws one-way surrogate type $i \sim j$ with probability $(2\matchrate_{ij}^t - 1)^+$;
	it draws a two-way surrogate type with agent $j$ as one of the two choices with probability $2\cdot \min\{ \matchrate_{ij}^t, 1-\matchrate^t_{ij} \}$.
\end{lemma}

Further, we maintain an invariant that the expected allocation to an offline agent $j$ of the decomposed surrogate types is the same as with the fractional allocation $\matchrate_{ij}^t$ of the original online type $i$.
To formally state this property, we introduce some notations that will also be useful later in our analyses.
Let $\mass_{i \sim j}^t$ and $\mass_{i \sim \{j,k\}}^t$ denote the probability of realizing surrogate types $i \sim j$ and $i \sim \{j, k\}$ respectively at time step $t \in [T]$, over the random realization of the online item's type $i$ at time $t$ and the randomness in the \textbf{Type Decomposition} algorithm.
By the definition of the algorithm, we have:

\begin{lemma}[Allocation Conservation]
	\label{lem:match-rate-conservation}
    For any online item type $i \in I$, any offline agent $j \in J$, and any time step $t \in [T]$, we have:
    \[
        \mass_i^t \cdot \matchrate_{ij}^t \: = \: \mass_{i \sim j}^t + \frac{1}{2} \sum_{k \ne j} \mass_{i \sim \{j,k\}}^t
        ~.
    \]
\end{lemma}

For unweighted and vertex-weighted matching, the original online item type $i$ does not affect how the objective would change if we match the online vertex to $j$ or $k$.
Therefore, we will omit $i$ and write the surrogate type as $j$ and $\{j, k\}$.
The probability of realizing a one-way type $j$ (respectively, two-way type $\{j, k\}$) will be the sum of the probability of realizing $i \sim j$ (respectively, $i \sim \{j, k\}$) over all online item types $i \in I$.
In other words:
\[
	\mass_j^t = \sum_{i \in I} \mass_{i \sim j}^t
	\quad
	,
	\quad
\mass_{\{j,k\}}^t = \sum_{i \in I} \mass_{i \sim \{j,k\}}^t
\]
are the probabilities of drawing surrogate types $j$ and $\{j, k\}$ respectively.

In this case, the Allocation Conservation property can be written as:
\begin{equation}
	\label{eqn:match-rate-conservation}
    y_j^t = \sum_{i \in I} \mass_i^t \cdot \matchrate_{ij}^t \: = \: \mass_j^t + \frac{1}{2} \sum_{k \ne j} \mass_{\{j,k\}}^t
    ~.	
\end{equation}

\section{Unweighted and Vertex-Weighted Matching}
\label{sec:vertex-weighted}

\subsection{Optimal Two-Way SOCS}
\label{sec:optimal-two-way-socs}

Recall the setting of two-way OCS.
Consider $T$ discrete time steps. 
At each time step $t \in [T]$, an online vertex arrives with its two-way type $\{j, k\}$ drawn independently from a distribution, which we abuse notation and still refer to as $\Distribution^t$ in this subsection.
The fractional allocation for a realized two-way type $\{j, k\}$ at time step $t$ is $\matchrate^t_j = \matchrate^t_k = \frac{1}{2}$.

We will monitor the expectation of the total fractional allocation to each offline vertex $j$, denoted as $y_j^{1:t}$.
Here, by definition we have:
\[
	y_j^t = \frac{1}{2} \sum_{k \ne j} \mass^t_{\{j,k\}}
	\quad,\quad
	y_j^{1:t} = \sum_{t'=1}^t y_j^{t'}
	~.
\]

In particular, recall that the fractional allocation in a two-way instance allocates the online vertex equally between the two offline vertices in the realized two-way type.
Therefore, $2 y_j^{1:t}$ equals the expected number of online vertices from time $1$ to $t$ whose two-way types involve offline vertex $j$ as one of the two choices, over the random realization of online types.

Following the terminology in the OCS literature, an offline vertex is \emph{unselected} if it has not yet been selected by the SOCS thus far.
We will keep track of the subset of unselected offline vertices, i.e., the set of unmatched offline vertices in the matching problem.
Since we consider unweighted and vertex-weighted matching, we only need to consider the unselected offline vertices in each pair, whenever such a vertex exists.
Hence, the only non-trivial decision is how to select an offline vertex from a pair when both offline vertices therein are unselected.

\bigskip

\begin{tcolorbox}
    \textbf{Two-Way SOCS for Unweighted and Vertex-Weighted Matching}\\[2ex]
    For each pair $\{j, k\}$ that arrives at time step $t \in [T]$:
    \begin{itemize}
        \item Select an unselected vertex $j$ or $k$ with probability proportional to $e^{2 y_j^{1:t}}$ and $e^{2 y_k^{1:t}}$.
    \end{itemize}
\end{tcolorbox}

\bigskip

We remark again that this algorithm is different from a seemingly similar (multi-way) Poisson OCS algorithm by \citet{HuangSY:STOC:2022} for the IID model.
Specializing the Poisson OCS to the two-way special case, it selects offline vertices $j$ and $k$ with probability proportional to $e^{y_j^{1:t}}$ and $e^{y_k^{1:t}}$ respectively.
By comparison, our algorithm favors the offline vertex with a larger expected number of appearances in the past more aggressively than the Poisson OCS algorithm, when both offline vertices are still unselected.

This subtle difference fundamentally changes the underlying analysis.
The new analysis in this paper relies on the structure of the two-way special case and the AM-GM inequality.

\begin{theorem}
    \label{thm:stochastic-ocs}
    \textrm{\bf Two-Way SOCS for Unweighted and Vertex-Weighted Matching} achieves convergence rate:
    \[
        g\big( y_j \big) = \big(1+y_j\big) e^{-2y_j}
        ~,
    \]
    i.e., for any offline vertex $j \in J$, the probability that $j \in J$ is unselected at the end is at most $g(y_j)$.
\end{theorem}

For any subset of offline vertices $S \subseteq J$, let $u_S^t$ be the probability that all vertices in $S$ remain unselected after the first $t$ time steps.
For singletons $S = \{j\}$, we write $u_j^t$ for $u_S^t$ for notational simplicity.
We have $u_S^0 = 1$ at the beginning for all subsets $S \subseteq J$.
Further, the above theorem is equivalent to $u_j^T \le g(y_j)$.
We will next characterize these unselected probabilities by a recurrence.

\begin{lemma}
    \label{lem:ocs-recurrence}
    For any time step $t \in [T]$ and any subset of offline vertices $S \subseteq J$, we have:
    \[
    	u_S^t ~= \sum_{\{ j, k \} : j, k \notin S} \mass_{\{j,k\}}^t \cdot u_S^{t-1} + \sum_{\{j, k\} : j \in S, k \notin S} \mass_{\{j,k\}}^t \cdot u_{S+k}^{t-1} \cdot \frac{e^{2 y_k^{1:(t-1)}}}{e^{2 y_j^{1:(t-1)}} + e^{2 y_k^{1:(t-1)}}}
    	~.
    \]
\end{lemma}

\begin{proof}
    Note that the arrival of a pair of offline vertices at time step $t$ is independent to the arrivals in the earlier time steps.
    Consider any time $t$ and any subset of offline vertices $S$. 
    There are three cases depending on the arrival at time $t$.
    
    First, if the realized pair of offline vertices at time step $t$ does not intersect with $S$, then the vertices in $S$ are unselected after time $t$ if and only if they were unselected before time $t$.
    In other words, the vertices in $S$ are unselected with probability $u_S^{t-1}$ in this case.
    This corresponds to the first term on the right-hand-side.

    Next, if the realized pair of offline vertices at time step $t$ are both in $S$, then one of them would be selected at time $t$.
    In other words, the vertices $S$ cannot all be unselected after time $t$.
    Thus, this case does not contribute to the right-hand-side.

    Finally, suppose that a pair $\{j, k\}$ arrives at time step $t$ with exactly one vertex in $S$.
    Without loss of generality, we may assume that $j \in S$ but $k \notin S$.
    Then, for all vertices in $S$ to remain unselected after time $t$, we need two conditions: (1) all elements in $S + k$ are unselected before time $t$, and (2) the algorithm selects $k$ instead of $j$ at time $t$.
    That is, this case contributes:
    \[
        u_{S+k}^{t-1} \cdot \frac{e^{2 y_k^{1:(t-1)}}}{e^{2 y_j^{1:(t-1)}} + e^{2 y_k^{1:(t-1)}}}
        ~,
    \]
    and corresponds to the second term on the right-hand-side.
\end{proof}

The next lemma is a corollary of Lemma~\ref{lem:ocs-recurrence} and the application of the AM-GM inequality on the denominator of the last term on the right-hand-side.

\begin{lemma}
	\label{lem:ocs-recurrence-am-gm}
    For any time $t \in [T]$ and any subset of offline vertices $S \subseteq J$, we have:
    \[
    	u_S^t ~\le \sum_{\{ j, k \} : j, k \notin S} \mass_{\{j,k\}}^t \cdot u_S^{t-1} + \frac{1}{2} \sum_{\{j, k\} : j \in S, k \notin S} \mass_{\{j,k\}}^t \cdot u_{S+k}^{t-1} \cdot e^{y_k^{1:(t-1)} - y_j^{1:(t-1)}}
    	~.
    \]
\end{lemma}

We will next prove a baseline upper bound of $u_S^t$ for all subsets $S$, which is the same bound that the baseline Independent Rounding would guarantee.
We defer the proof to Appendix~\ref{app:ocs-basic}, since it follows the same approach that we will demonstrate in the proof of Theorem~\ref{thm:stochastic-ocs}.

\begin{lemma}
	\label{lem:ocs-basic}
    For any time $t \in [T]$ and any subset of offline vertices $S \subseteq J$, we have:
    \[
        u_S^t \le e^{- \sum_{j \in S} y_j^{1:t}}
        ~.
    \]
\end{lemma}

\begin{proof}[Proof of Theorem~\ref{thm:stochastic-ocs}]
    Consider any offline vertex $j \in J$.
    By Lemma~\ref{lem:ocs-recurrence-am-gm} with $S = \{j\}$, we have:
	\[
    	u_j^t ~\le~ \sum_{k, \ell \ne j} \mass_{\{k, \ell\}}^t \cdot u_j^{t-1} ~+~ \frac{1}{2} \: \sum_{k \ne j} \mass_{\{j,k\}}^t \cdot u_{\{j, k\}}^{t-1} \cdot e^{y_k^{1:(t-1)} - y_j^{1:(t-1)}}
    \]
    
    Further, apply the bound from Lemma~\ref{lem:ocs-basic} to subsets $\{j, k\}$, we get that:
    \begin{equation}
        \label{eqn:ocs-singleton-recurrence}
    	u_j^t \:\le \sum_{k, \ell \ne j} \mass_{\{k,\ell\}}^t \cdot u_j^{t-1} + \frac{1}{2} \sum_{k \ne j} \mass_{\{j,k\}}^t \cdot e^{- 2 y_j^{1:(t-1)}}
		~.
    \end{equation}
    
    Using the above inequality, we will now prove the following inequality by an induction on the time step $t$ from $0$ to $T$:
    \[
    	u_j^t \le \Big(1 + y_j^{1:t} \Big) \cdot e^{-2 y_j^{1:t}}
    	~.
    \]
    
    The theorem would then follow as the case of $t = T$.

    The base case when $t = 0$ is trivial because both sides are equal to $1$.
    Next, suppose that the bound holds for $t-1$.
    Applying the induction hypothesis to the above Inequality~\eqref{eqn:ocs-singleton-recurrence} gives:
    \[
    	u_j^t \:\le\:\bigg( \sum_{k, \ell \ne j} \mass_{\{k, \ell\}}^t \cdot \Big(1 + y_j^{1:(t-1)} \Big) + \frac{1}{2} \sum_{k \ne j} \mass_{\{j,k\}}^t \bigg) \:\cdot \: e^{- 2 y_j^{1:(t-1)}}
    	~.
    \]
    
    Note that:
    \[
    	\sum_{k \ne j} f^t_{\{j,k\}} \:=\: 2 y_j^t
    	\quad\qquad
    	\sum_{k, \ell \ne j} \mass_{\{k, \ell\}}^t = 1 - \sum_{k \ne j} f^t_{\{j,k\}} = 1 - 2 y_j^t
    	~.
    \]
    
    The above inequality becomes:
    \[
    	u_j^t \:\le\:\bigg( \big(1 - 2 y_j^t\big) \cdot \Big(1 + y_j^{1:(t-1)} \Big) + y_j^t \bigg) \:\cdot \: e^{- 2 y_j^{1:(t-1)}}
    	~.
	\]
	
	Comparing this with the final bound $\big( 1 + y_j^{1:t} \big) e^{-2 y_j^{1:t}}$, it suffices to prove that:
	\begin{align*}
		\bigg( \big(1 - 2 y_j^t\big) \cdot \Big(1 + y_j^{1:(t-1)} \Big) + y_j^t \bigg)
		&
		\le \big( 1 + y_j^{1:t} \big) \cdot e^{- 2 y_j^t} \\
		&
		= \big( 1 + y_j^{1:(t-1)} + y_j^t \big) \cdot e^{- 2 y_j^t} 
		~.
	\end{align*}
	
	This inequality is linear in $y_j^{1:(t-1)}$.
	Hence, we will group terms by those with $y_j^{1:(t-1)}$ and those without, and prove the inequality separately for these two types of terms.
	The inequality for the coefficients of term $y_j^{1:(t-1)}$ is:
	\[
		1 - 2 y_j^t \le  e^{- 2 y_j^t}
		~.
	\]
	
	The inequality for the other terms is:
	\begin{equation}
		\label{eqn:e-2x}
		1 - y_j^t \le \big( 1 + y_j^t \big) \cdot e^{- 2 y_j^t}
		~,
	\end{equation}
	which holds for all $y_j^t \ge 0$ (Appendix~\ref{app:e-2x}).
\end{proof}

We also provide a theorem below formalizing the optimality of our algorithm with respect to the two-way SOCS problem for unweighted and vertex-weighted matching.

\begin{theorem}
For any $y \ge 0$, there is an instance such that for any algorithm, there exists an offline vertex $j$ with $y_j = y$ and it is unselected with probability at least $(1+y) e^{-2y}$.
\end{theorem}
\begin{proof}
    For any $y \ge 0$, consider an instance that only contains a pair of offline vertices $\{j,k\}$.
    Let the only online type $\{j, k\}$ arrive by a homogeneous Poisson process with arrival rate $2y$ in a time interval $[0,1]$;
    the Poisson process is achieved at the limit in the discrete-time model by considering arrival rate $\nicefrac{2y}{T}$ in each step and letting the number of time steps $T$ tend to infinity.
    
    The probability of $j$ is selected plus the probability of $k$ is selected equals the expectation of the number of selections, which is at most:
    \begin{align*}
    	\E_{n \sim \text{Poisson$(2y)$}} \min\{2, n\} 
    	&
    	~=~ 1 \cdot \Pr[n \sim \text{Poisson$(2y)$}]{n = 1} + 2 \cdot \Pr[n \sim \text{Poisson$(2y)$}]{n \ge 2} \\[1ex]
    	&
    	~=~ 1 \cdot 2y e^{-2y} + 2 \cdot \left(1 - e^{-2y} - 2y e^{-2y} \right) \\[1ex]
    	&
    	~=~ 2 - (2+2y)e^{-2y}
    	~.
    \end{align*}

    Hence, one (or both) of $j$ and $k$ is selected with probability no more than half of the above quantity, i.e., $1 - (1+y) e^{2y}$.
\end{proof}

\subsection{General SOCS from Two-Way SOCS}
\label{sec:jl-demonstration}

We can obtain a general SOCS for unweighted and vertex-weighted matching by reducing it to two-way SOCS via the \textbf{Type Decomposition} algorithm in Subsection~\ref{sec:socs-reduction}.
This simple approach is good enough to obtain a $0.688$-competitive algorithm for unweighted and vertex-weighted Non-IID Online Stochastic Matching, improving the state-of-the-art ratio $0.666$ by \citet{TangWW:STOC:2022}.
Since this is subsumed by the improved algorithm in the next subsection, we only sketch the analysis under an additional assumption that for any offline vertex $j \in J$ we have:
\begin{equation}
	\label{eqn:jl-constraint}	
	\sum_{t=1}^T \:\underbrace{\sum_{i \in I} \big( 2 x_{ij}^t - \mass_i^t \big)^+}_{\substack{\text{$\mass_j^t$: probability of }\\ \text{one-way surrogate type $j$}}} \le~ 1 - \ln 2
	~. 
\end{equation}

This is a constraint of the Jaillet-Lu LP for IID Online Stochastic Matching~\cite{JailletL:MOR:2014}; 
the inequality was first proved by \citet*{ManshadiOS:MOR:2012}.
It would also hold in the non-IID model if the probability of realizing any online type $i \in I$ in any step $t \in [T]$ is sufficiently small such that the arrival process is approximately equivalent to a non-homogeneous Poisson arrival process.
The stated competitive ratio $0.688$ holds unconditionally even when the above inequality fails.
Nonetheless, it is easier to sketch the argument with the assumption.

Consider any offline vertex $j \in J$.
Let $\mass_j = \mass_j^{1:T}$ denote the expected number of online vertices with one-way surrogate type $j$.
Similarly, let $\mass_{\{j,k\}} = \mass_{\{j,k\}}^{1:T}$ be the expected number of online vertices with two-way surrogate types $\{j, k\}$.
Denote the total fractional allocation to an offline vertex $j$ from two-way surrogate types, for which we will select using the Two-Way SOCS, as:
\[
    z_j = \frac{1}{2} \sum_{k \ne j} \mass_{\{j,k\}}
    ~.
\]

Finally, by the Allocation Conservation property (Lemma~\ref{lem:match-rate-conservation}), we have:
\[
    z_j = y_j - \mass_j
    ~.
\]

Offline vertex $j$ stays unmatched if and only if (1) the one-way surrogate type $j$ never arrives and (2) the Two-Way SOCS never selects $j$.
The probability of the first event is:
\[
	\prod_{t=1}^T \big(1 - \mass^t_j\big) \le e^{- \mass^{1:T}_j} = e^{-\mass_j}
\]

Further, Theorem~\ref{thm:stochastic-ocs} upper bounds the probability of the second event.
Since we handle one-way and two-way types independently, the probability that $j$ stays unmatched at the end is at most:
\[
    e^{-\mass_j} \cdot \big(1 + z_j\big) e^{-2 z_j}
    =
    \big( 1 + y_j - \mass_j \big) e^{-2y_j+\mass_j}
    ~.
\]

This is increasing in $\mass_j$ because $(1+x) e^{-x}$ is decreasing in $x \ge 0$ and we have $\mass_j \le y_j$.
By Equation~\eqref{eqn:jl-constraint}, we also have $\mass_j \le 1 - \ln 2$.
Hence, for any fixed $0 \le y_j \le 1$, the above unmatched probability is maximized when $\mass_j = \min \big\{ y_j, 1 - \ln 2\big\}$, in which case $y_j - \mass_j = \big(y_j - 1 + \ln 2\big)^+$.
In other words, vertex $j$ is matched with probability at least:
\begin{equation}
    \label{eqn:socs-basic-matching-bound}
    1 - \big( 1 + (y_j - 1 + \ln 2)^+ \big) e^{-y_j-(y_j-1+\ln 2)^+}
    ~.
\end{equation}

This is zero when $y_j = 0$ and is a concave function in $y_j$ (Appendix~\ref{app:socs-basic-matching-bound}).
Hence, it is at least:
\[
    \big( 1 - (1+\ln2) e^{-1-\ln2} \big) \cdot y_j > 0.688 \cdot y_j
    ~.
\]

\subsection{Improved General SOCS}
\label{sec:unweighted-socs}

The algorithm is as follows:

\bigskip

\begin{tcolorbox}
    \textbf{SOCS for Unweighted and Vertex-Weighted Matching}\\[2ex]
    When an online vertex with one-way surrogate type $j$ arrives:
    \begin{itemize}
        \item If $j$ is still unmatched, match to $j$.
    \end{itemize}
    When an online vertex with two-way surrogate type $\{j, k\}$ arrives at time $t$:
    \begin{itemize}
        \item Match to an unmatched vertex $\ell \in \{j,k\}$ with probability proportional to $e^{2y_\ell^{1:(t-1)}}$.
    \end{itemize}
\end{tcolorbox}

\bigskip

Compared to the original reduction to Two-Way SOCS, this algorithm has a subtle yet important difference, as we no longer use the Two-Way SOCS as a blackbox, independent to the decisions made regarding the one-way surrogate types.
In the presence of a two-way type $\{j, k\}$, the probability of sampling an unmatched offline vertex $\ell \in \{j, k\}$ is proportional to $e^{2 y_\ell^{1:(t-1)}}$, where the exponent is twice the fractional allocation to offline vertex $\ell$ from time $1$ to $t-1$, \emph{including the contribution from one-way surrogate types}.
By contrast, the Two-Way SOCS as a blackbox would let the exponent be only twice the fractional allocation to $\ell$ due to two-way surrogate types.

\begin{theorem}
    \label{thm:socs-vertex-weighted}
    \textbf{SOCS for Unweighted and Vertex-Weighted Matching} has convergence rate:
    \[
    	g(y_j) = 
    	\begin{cases}
			\frac{1}{4} \left( e^{-2y_j} + 3 - 2y_j \right)
			&
			0 \le y_j \le \frac{1}{2} \:; \\[1.5ex]
			e^{-2y_j} \left( \frac{1+e}{4} + \frac{e}{2} y_j \right)
			&
			\frac{1}{2} < y_j \le 1 \:; \\
		\end{cases}
    \]
    That is, for any offline vertex $j \in J$, it is matched with probability at least:
    \begin{equation}
        \label{eqn:socs-matching}	
        1 - g(y_j)
        ~.
    \end{equation}
\end{theorem}

The matched probability in Equation~\eqref{eqn:socs-matching} is a concave function in $y_j$ (Appendix~\ref{app:ocs-matching-bound}), and equals zero when $y_j = 0$.
Hence, it is at least:
\[
    \big( 1-g(1) \big) \cdot y_j
    =
    \bigg(1-\frac{3}{4e} - \frac{1}{4e^2}\bigg)
    \cdot y_j
    ~.
\]

As a corollary, we get the following competitive ratio for Unweighted and Vertex-Weighted Online Stochastic Matching, improving the best existing $0.666$ competitive algorithm by \citet{TangWW:STOC:2022}.

\begin{corollary}
	\label{cor:vertex-weighted-matching}
	Rounding the solution of \textbf{Stochastic Matching LP} with \textbf{SOCS for Unweighted and Vertex-Weighted Matching} is $1 - \frac{3}{4e} - \frac{1}{4e^2} > 0.69$ competitive for the Unweighted and Vertex-Weighted Online Stochastic Matching (in the general non-IID model).
\end{corollary}

Recall that $u_S^t$ denotes the probability that every offline vertex $j \in S$ is still unmatched by the end of time step $t$.
Further, for any offline vertex $j \in J$, we write $u_j^t$ for $u_{\{j\}}^t$ for notational simplicity.
We will first show the following recurrence for $u_S^t$, which is a counterpart to Lemma~\ref{lem:ocs-recurrence} for Two-Way SOCS.
The proof, which we include in Appendix~\ref{app:general-ocs-recurrence} for completeness, is essentially the same as that of Lemma~\ref{lem:ocs-recurrence}.

\begin{lemma}
    \label{lem:general-ocs-recurrence}
    At any time step $t \in [T]$ and any subset of offline vertices $S \subseteq J$, we have:
    \begin{align*}
        u_S^t
        &
        ~ = ~
        \bigg(1 - \sum_{j\in S}\mass_j^t - \sum_{\{j,k\} : j, k\in S}\mass_{\{j,k\}}^t - \sum_{\{j,k\} : j\in S,k\notin S}\mass_{\{j,k\}}^t \bigg) \cdot u_S^{t-1} \\
        & \qquad\qquad
        + \sum_{\{j,k\} : j\in S,k\notin S}\mass_{\{j,k\}}^t \cdot u_{S+k}^{t-1} \cdot \frac{e^{2y_k^{1:(t-1)}}}{e^{2y_k^{1:(t-1)}} + e^{2y_j^{1:(t-1)}}}
        ~.
    \end{align*}
\end{lemma}

Applying the AM-GM inequality to the denominator of the last term, we get the next lemma as a corollary of Lemma~\ref{lem:general-ocs-recurrence}.

\begin{lemma}
    \label{lem:general-ocs-recurrence-am-gm}
    At any time step $t \in [T]$ and any subset of offline vertices $S \subseteq J$, we have:
    \begin{align*}
        u_S^t
        &
        ~ \le ~
        \bigg(1 - \sum_{j\in S}\mass_j^t - \sum_{\{j,k\} : j, k\in S}\mass_{\{j,k\}}^t - \sum_{\{j,k\} : j\in S,k\notin S}\mass_{\{j,k\}}^t \bigg) \cdot u_S^{t-1} \\
        & \qquad\qquad
        + ~ \frac{1}{2} \sum_{\{j,k\} : j\in S,k\notin S}\mass_{\{j,k\}}^t \cdot u_{S+k}^{t-1} \cdot e^{y_k^{1:(t-1)} - y_j^{1:(t-1)}}
        ~.
    \end{align*}
\end{lemma}

We next show a baseline upper bound of $u_S^t$ for all subsets $S$, whose proof we include in Appendix~\ref{app:matching-ocs-basic} for completeness.
This is a counterpart to Lemma~\ref{lem:ocs-basic} for Two-Way SOCS.

\begin{lemma}
    \label{lem:matching-ocs-basic}
    At the end of any time step $t \in [T]$ and for any subset of offline vertices $S \subseteq J$:
    \[
        u_S^t \le e^{-\sum_{j \in S} y_j^{1:t}}
        ~.
    \]
\end{lemma}

\begin{proof}[Proof of Theorem~\ref{thm:socs-vertex-weighted}]

    Consider any offline vertex $j \in J$.
    By Lemma~\ref{lem:general-ocs-recurrence-am-gm} with $S = \{j\}$, and applying Lemma~\ref{lem:matching-ocs-basic} to subsets $S = \{j, k\}$, we have:
    \[
        u_j^t \le \bigg( 1 - \mass_j^t - \sum_{k \neq j}\mass_{\{j,k\}}^t \bigg) \cdot u_j^{t-1} + \sum_{k \neq j}\mass_{\{j,k\}}^t \cdot e^{-2y_j^{1:(t-1)}}
        ~.
    \]

    Up to this point, the analysis is almost verbatim to the counterpart for Two-Way SOCS.
    
    The rest of the argument will be different.
    By opening the blackbox of the Two-Way SOCS, we can use the relation between the distribution over the surrogate types and the distribution over the original online types.
    Concretely, for any offline vertex $j \in J$ and any time step $t \in [T]$, let:
    \[
        I_j^t = \Big\{ i : \matchrate_{ij}^t > \frac{1}{2} \Big\}
    \]
    be the subset of original online types more than half of which is allocated to $j$ at time step $t$.
    
    By the definition of \textbf{Type Decomposition} and specifically by Lemma~\ref{lem:type-decomposition-probability}, if an online vertex at time step $t$ has type $i \in I_j^t$, then it would draw a two-way type involving $j$ with probability $2 (1 - \matchrate^t_{ij})$ (and one-way type $j$ with the remaining $2 \matchrate^t_{ij} - 1$ probability).
    Otherwise, i.e., if the online vertex has type $i \notin I_j^t$, then it would draw a two-way type involving $j$ with probability $2 \matchrate^t_{ij}$ (and one-way type $j$ with zero probability).
    Putting together, the total probability of realizing a two-way type involving $j$ at time step $t$ can be written as:
    \begin{align*}
        \sum_{k\neq j} \mass_{\{j,k\}}^t
        &
        = \sum_{i \in I_j^t} \mass_i^t \cdot 2 \big(1 - \matchrate^t_{ij} \big) + \sum_{i\notin I_j^t} \mass^t_j \cdot 2 \matchrate_{ij}^t \\
        &
        = 2 \cdot \sum_{i\in I_j^t} \mass_i^t + 4 \cdot \sum_{i\notin I_j^t} x_{ij}^t - 2 \cdot y_j^t
        ~,
    \end{align*}
    where recall that $x_{ij}^t = \mass^t_i \matchrate^t_{ij}$ and $y^t_j = \sum_{i \in I} \mass^t_i \matchrate^t_{ij}$.

    Next, we combine it with the Allocation Conservation property (Equation~\eqref{eqn:match-rate-conservation}):
    \[
    	y_j^t = \mass_j^t + \frac{1}{2} \sum_{k \neq j} \mass_{\{j,k\}}^t
    	~,
    \]
    
    We get that:
    \[
    	\sum_{k\neq j} \mass_{\{j,k\}}^t = \sum_{i\in I_j^t} \mass_i^t + 2 \cdot \sum_{i\notin I_j^t} x_{ij}^t - \mass_j^t
    	~.
    \]
    
    Hence, we can rewrite the above upper bound of $u_j^t$ as:
    \begin{align*}
        u_j^t
        &
        ~\le~
        \bigg(1 - \sum_{i\in I_j^t} \mass_i^t - 2 \sum_{i\notin I_j^t} x_{ij}^t \bigg) u_j^{t-1} + \bigg(\sum_{i\in I_j^t} \mass_i^t + \sum_{i\notin I_j^t} 2x_{ij}^t - y_j^t\bigg) e^{-2y_j^{1:(t-1)}} \\
        & %
        ~=~ 
        \bigg(1 - \sum_{i\in I_j^t} \mass_i^t - 2 \sum_{i\notin I_j^t} x_{ij}^t \bigg) \bigg(u_j^{t-1} - e^{-2y_j^{1:(t-1)}}\bigg) + \Big(1 - y_j^t\Big) e^{-2y_j^{1:(t-1)}}
        ~.
    \end{align*}

    Consider an auxiliary array $Q^t = u_j^t - e^{-2 y_j^{1:t}}$.
    Then, it is sufficient to upper bound this auxiliary array.
    The above upper bound of $u_j^t$ is equivalent to:
    \[
        Q^t \le \bigg( 1-\sum_{i\in I_j^t} \mass_i^t - 2 \sum_{i\notin I_j^t} x_{ij}^t \bigg) \cdot Q^{t-1} + \Big(\big(1-y_j^t\big) \cdot e^{2y_j^t} - 1 \Big) e^{-2y_j^{1:t}}
        ~.
    \]

    Further by $(1-y) \cdot e^{2y} \le 1+y$ (Appendix~\ref{app:e-2x}) for any $0 \le y \le 1$, we have:
    \[
        Q^t \le \bigg(1-\sum_{i\in I_j^t} \mass_i^t - 2 \sum_{i \notin I_j^t} x_{ij}^t\bigg) \cdot Q^{t-1} + y_j^t \cdot e^{-2y_j^{1:t}}
        ~.
    \]
    
    Combining the above inequalities for $1 \le t \le T$ with the base case $Q^0 = 0$, we get that:
    \begin{equation}
        \label{eqn:Q-first-bound}
        Q^T ~ \le ~ \sum_{t \in [T]} y_j^t \cdot e^{-2y_j^{1:t}}\cdot \prod_{t' > t} \bigg(1-\sum_{i\in I_j^{t'}}\mass_i^{t'} - 2 \sum_{i\notin I_j^{t'}}x_{ij}^{t'}\bigg)
        ~.
    \end{equation}

    We next upper bound the product above, i.e., the second term.
    First, we have:
    \begin{align*}
        1-\sum_{i\in I_j^{t'}} \mass_i^{t'} - 2 \sum_{i\notin I_j^{t'}} x_{ij}^{t'}
        &
        ~\le~ 
        \bigg( 1-\sum_{i \in I_j^{t'}} \mass_i^{t'} \bigg) \bigg( 1 - 2 \sum_{i \notin I_j^{t'}} x_{ij}^{t'} \bigg) \\
        &
        ~\le~ 
        \bigg( 1-\sum_{i \in I_j^{t'}} \mass_i^{t'} \bigg) \exp\bigg( - 2 \sum_{i \notin I_j^{t'}} x_{ij}^{t'} \bigg)
        ~.
    \end{align*}

    Putting this into Eqn.~\eqref{eqn:Q-first-bound}, we have:
    \[
        Q^T
        ~\le~ 
        \sum_{t \in [T]} y_j^t \cdot \exp \bigg( - 2y_j^{1:t} - 2 \sum_{t' > t} \sum_{i \notin I_j^{t'}} x_{ij}^{t'} \bigg) \cdot \prod_{t' > t} \bigg( 1 - \sum_{i \in I_j^{t'}} \mass_i^{t'} \bigg)
        ~.
    \]
    
    Further, note that:
    \[
    	y_j = y_j^{1:t} + y_j^{(t+1):T} = y_j^{1:t} + \sum_{t' > t} \sum_{i \in I_j^{t'}} x_{ij}^{t'} + \sum_{t' > t} \sum_{i \notin I_j^{t'}} x_{ij}^{t'} 
    	~.
    \]
    
    We get that:
    \begin{equation}
        \label{eqn:for-random-order}    
    	Q^T
        ~\le~ 
        e^{-2y_j} \sum_{t \in [T]} y_j^t \cdot \exp \bigg( 2 \sum_{t' > t} \sum_{i\in I_j^{t'}} x_{ij}^{t'} \bigg) \cdot \prod_{t' > t} \bigg( 1 - \sum_{i \in I_j^{t'}} \mass_i^{t'} \bigg)
        ~.
    \end{equation}

    We further relax the product above, i.e., the last term, by Constraint~\eqref{eqn:discrete-time-lp} of the \textbf{Stochastic Matching LP}.
    We get that:
   	\[
		Q^T 
		~\le~ 
		e^{-2y_j} \sum_{t \in [T]} y_j^t \cdot \exp \bigg(2 \sum_{t' > t} \sum_{i\in I_j^{t'}} x_{ij}^{t'} \bigg) \cdot \bigg( 1 - \sum_{t' > t} \sum_{i\in I_j^{t'}}x_{ij}^{t'} \bigg)
        ~.
    \]

    For notational simplicity, let $z^t = \sum_{t' > t} \sum_{i\in I_j^{t'}} x_{ij}^{t'}$.
    The above upper bound becomes:
    \[
        Q^T ~ \le ~ e^{-2y_j} \sum_{t \in [T]} y_j^t \cdot e^{2 z^t}\big(1-z^t\big) 
        ~.
    \]

    Function $e^{2z}(1-z)$ achieves it maximum $\frac{e}{2}$ at $z=\frac12$.
    We flatten it beyond $\frac{1}{2}$ and get a non-decreasing function  $h$ with $h(z) = e^{2z}(1-z)$ when $z\le \frac12$, and $h(z) = \frac{e}{2}$ otherwise.
    We have:
    \[
        Q^T 
        \le 
        e^{-2y_j} \sum_{t \in [T]} y_j^t \, h\big(z^t\big) 
        ~.
    \]
    
    By the definition of $z^t$, we have $z^t \le y_j^{(t+1):T}$. 
    Since $h$ is non-decreasing, we further get that:
    \[
    	Q^T
        \le e^{-2y_j} \sum_{t \in [T]} y_j^t \, h\big(y_j^{(t+1):T}\big) 
        \le e^{-2y_j} \int_0^{y_j} h(z)\:\dif z
        ~.
    \]

    Recall that $u_j^T = Q^T + e^{-2y_j}$ is the probability that offline vertex $j$ is unmatched in the end.
    The probability that vertex $j$ stays unmatched is at most:
	\[
        e^{-2y_j}\left(1 + \int_0^{y_j} h(z) \:\dif z\right)
        ~.
    \]
    
    By the definition of $h$ and basic calculus, this is:
    \[
    	\frac{1}{4} \left( e^{-2y_j} + 3 - 2y_j \right)
    \]
    for $0 \le y_j \le \frac{1}{2}$, and:
    \[
    	e^{-2y_j} \left( \frac{1+e}{4} + \frac{e}{2} y_j \right)
    \]
    for $\frac{1}{2} < y_j \le 1$. 
\end{proof}

\subsection{Random-Order and Query-Commit Models}
\label{sec:unweighted-ro}

Intuitively, an adversarially chosen instance for the SOCS algorithm would leave as many one-way types to the end as possible, forcing the algorithm to select an offline vertex from two-way surrogate types before knowing which agent will reappear later as a one-way type.
The adversary's ability to do so is limited in the Random-Order model.
As a result, we can prove an improved convergence rate and correspondingly a better competitive ratio for the same algorithm.

\begin{theorem}
	\label{thm:socs-random-order}
    In the Random-Order model, \textrm{\bf SOCS for Unweighted and Vertex-Weighted Matching} has convergence rate:
    \[
    	g(y_j)
        ~=~ 
        \begin{cases}
        	\big(1 + \frac{y_j}{2}\big) e^{-2y_j} + \frac{1}{2} y_j \big(1 - y_j \big)
        	&
        	0 \le y_j \le \frac{1}{2} ~; \\[1ex]
            e^{-2y_j} \big( 1 + \big(\frac{1}{2} +  \frac{e}{4}\big)  y_j \big)
            &
            \frac{1}{2} < y_j \le 1 ~.	
        \end{cases}
    \]
    That is, for any offline vertex $j \in J$, it is matched with probability at least:
    \begin{equation}
        \label{eqn:socs-matching-random-order}	
        1 - g(y_j)
        ~.
    \end{equation}
\end{theorem}

Equation~\eqref{eqn:socs-matching-random-order} is a concave function in $y_j$ (Appendix~\ref{app:socs-matching-random-order}). Therefore, it is at least:
\[
	\big(1 - g(1) \big) y_j = \left(1 - \frac{1}{4e} - \frac{3}{2e^2}\right) y_j
	~.
\]

As corollaries, we get the competitive ratios for unweighted and vertex-weighted matching in the Random-Order and Query-Commit models.

\begin{corollary}
	\label{cor:vertex-weighted-matching-random-order}
	Rounding the solution of \textbf{Stochastic Matching LP} with \textbf{SOCS for Unweighted and Vertex-Weighted Matching} is $1 - \frac{1}{4e} - \frac{3}{2e^2} > 0.705$-competitive for unweighted and vertex-weighted matching in the Random-Order model of (Non-IID) Online Stochastic Matching.
\end{corollary}

\begin{corollary}
	\label{cor:vertex-weighted-matching-query-commit}
	There is a $1-\frac{1}{4e}-\frac{3}{2e^2} > 0.705$-competitive algorithm for unweighted and vertex-weighted matching in the Query-Commit model.
\end{corollary}

\begin{proof}[Proof of Theorem~\ref{thm:socs-random-order}]
    It is more convenient to work with an alternative way of shuffling the time steps as follows.
    For each time step $t \in [T]$, independently draw $\theta^t\in [0, 1]$ uniformly at random.
    Then, let the steps arrive in ascending order of $\theta^t$.
    We follow the analysis of Theorem~\ref{thm:socs-vertex-weighted} up to Equation~\eqref{eqn:for-random-order}, replacing the condition $t' > t$ with $\theta^{t'} > \theta^t$ due to the above shuffling:
    \begin{align*}
        Q^t
        &
        \le e^{-2y_j} \sum_{t \in [T]} y_j^t \cdot  \exp \bigg( \sum_{t' :\: \theta^{t'} > \theta^t} \sum_{i \in I_j^{t'}} x_{ij}^{t'} \bigg) \prod_{t' :\: \theta^{t'} > \theta^t} \bigg(1 - \sum_{i \in I_j^{t'}} \mass_i^{t'} \bigg) \\
        &
        = e^{-2y_j} \sum_{t \in [T]} y_j^t \prod_{t' :\: \theta^{t'} > \theta^t} \exp \bigg( \sum_{i \in I_j^{t'}} x_{ij}^{t'} \bigg) \bigg(1 - \sum_{i \in I_j^{t'}} \mass_i^{t'} \bigg)
    \end{align*}

    For notation simplicity, for any $t \in [T]$ we let:
    \[
        h^t = \exp\bigg(2\sum_{i\in I_j^t}x_{ij}^t\bigg)\bigg(1 - \sum_{i\in I_j^t}\mass_i^t\bigg)
        ~.
    \]

    The above bound can be rewritten as:
    \[
        Q^T \le e^{-2y_j} \sum_{t \in [T]} y_j^t \prod_{t' : \theta^{t'} > \theta^t} h^{t'}
        ~.
    \]

    Next, we will upper bound the expectation of the above product, i.e., the last term, over the random shuffling of time steps.
    First, we fix $\theta^t$ and take expectation over the randomness of $\theta^{t'}$ for $t' \ne t$.
    We get that:
    \[
        \E \, \Big[ \prod_{t' : \theta^{t'} > \theta^t} h^{t'} \,\mid\, \theta^t \, \Big] = \prod_{t' \ne t} \Big(\theta^t + \big(1-\theta^t\big) h^{t'} \Big) 
        ~.
    \]

    For those $t'$ for which $h^{t'} < 1$, we relax the corresponding term to $1$.
    For those $t'$ for which $h^{t'} \ge 1$, we repeatedly apply inequality:
    \[
        \Big( \theta^t + \big(1-\theta^t\big) a \Big) \Big( \theta^t + \big(1-\theta^t\big) b\Big) \le \theta^t + \big(1-\theta^t\big) ab
    \]
    for any $a, b \ge 1$, which is equivalent to $\big(\theta^t-(\theta^t)^2\big)(a-1)(b-1) \ge 0$.
    We have:
    \[
        \E \, \Big[ \prod_{t' : \theta^{t'} > \theta^t} h^{t'} \,\mid\, \theta^t \, \Big] ~\le~ \theta^t + \big(1-\theta^t\big) \prod_{t' \ne t : h^{t'} \ge 1} h^{t'}
        ~.
    \]

    Expanding the definition of $h^{t'}$, we get:
    \[
    	\prod_{t' \ne t : h^{t'} \ge 1} h^{t'}
        =
        \prod_{t' \ne t : h^{t'} \ge 1} \exp\bigg(2\sum_{i\in I_j^{t'}}x_{ij}^{t'}\bigg)\bigg(1 - \sum_{i\in I_j^{t'}}\mass_i^{t'}\bigg)
        ~.
    \]
    
    We next relax the last term using Constraint~\eqref{eqn:discrete-time-lp} of the \textbf{Stochastic Matching LP}:
    \[
    	\prod_{t' \ne t : h^{t'} \ge 1} h^{t'}
    	~\le~
    	\exp\bigg(2 \sum_{t' \ne t : h^{t'} \ge 1} \sum_{i\in I_j^{t'}}x_{ij}^{t'} \bigg) \bigg(1 - \sum_{t' \ne t : h^{t'} \ge 1} \sum_{i\in I_j^{t'}} x_{ij}^{t'} \bigg)
    	~.
    \]
    
	We view right-hand-side as a function $e^{2z}(1-z)$, which is increasing in $0 \le z \le \frac{1}{2}$ and decreasing in $\frac{1}{2} \le z \le 1$.
    If $y_j \le \frac{1}{2}$, we use $\sum_{i\in I_j^{t'}}x_{ij}^{t'} \le y_j$ to conclude that:
    \[
    	\prod_{t' \ne t : h^{t'} \ge 1} h^{t'} ~\le~ e^{2y_j} (1-y_j)
    	~.
    \]
    
    Otherwise, i.e., if $y_j > \frac{1}{2}$, we relax the right-hand-side to be the maximum value $\frac{e}{2}$ of function $e^{2z}(1-z)$, which is achieved at $z = \frac{1}{2}$.
    That is, we have:
    \[
    	\prod_{t' \ne t : h^{t'} \ge 1} h^{t'} ~\le~ \frac{e}{2}
    	~.
    \]
    
    Putting together, we conclude that:
    \[
        \E \, \Big[ \prod_{t' : \theta^{t'} > \theta^t} h^{t'} \,\mid\, \theta^t \, \Big] \le 
        \begin{cases}
			\theta^t + \big(1-\theta^t\big) \cdot e^{2y_j} (1-y_j) 
			&
			0 \le y_j \le \frac{1}{2} ~; \\[1ex]
			\theta^t + \big(1-\theta^t\big) \cdot \frac{e}{2}
			&
			\frac{1}{2} < y_j \le 1 ~.	
        \end{cases}
    \]

    Finally, we take expectation over the randomness of $\theta^t$ and get that:
    \begin{equation}
        \label{eqn:random-order-main-inequality}
        \E \prod_{t' : \theta^{t'} > \theta^t} h^{t'}
        ~ \le ~ 
        \begin{cases}
        	\frac{1}{2} \big( 1 + e^{2y_j}(1-y_j) \big)
        	&
        	0 \le y_j \le \frac{1}{2} ~; \\[1ex]
        	\frac{1}{2} + \frac{e}{4}
        	&
        	\frac{1}{2} < y_j \le 1 ~.
        \end{cases}
    \end{equation}

    Now, by linearity of expectation and Equation~\eqref{eqn:random-order-main-inequality}, we have:
    \[
        \E \, Q^T
        ~=~
        \begin{cases}
        	\frac{y_j}{2} \big( e^{-2y_j}  + 1 - y_j \big)
        	&
        	0 \le y_j \le \frac{1}{2} ~; \\[1ex]
            \big(\frac{1}{2} +  \frac{e}{4}\big) e^{-2y_j}  y_j
            &
            \frac{1}{2} < y_j \le 1 ~.	
        \end{cases}
    \]
        
    Recall that offline vertex $j$ stays unmatched in the end with probability $\E\,u_j^T = \E\,Q^T + e^{-2y_j}$, which is at most: 
    \[
        \E\,u_j^T
        ~\le~
        g(y_j)
        ~=~ 
        \begin{cases}
        	\big(1 + \frac{y_j}{2}\big) e^{-2y_j} + \frac{1}{2} y_j \big(1 - y_j \big)
        	&
        	0 \le y_j \le \frac{1}{2} ~; \\[1ex]
            e^{-2y_j} \big( 1 + \big(\frac{1}{2} +  \frac{e}{4}\big)  y_j \big)
            &
            \frac{1}{2} < y_j \le 1 ~.	
        \end{cases}
    \]
\end{proof}

We conclude the section by presenting in Figure~\ref{fig:socs-matching} a comparison of the convergence rates by the baseline \textbf{Independent Rounding} algorithm, and by our \textbf{SOCS for Unweighted and Vertex-Weighted Matching} in the original and the random-order models.

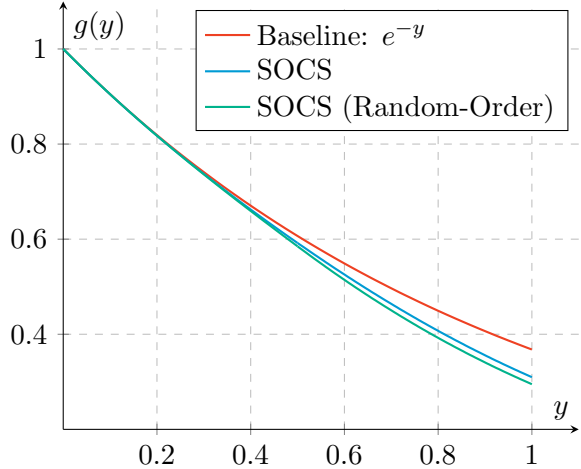
\begin{figure}
	\centering
	\begin{tikzpicture}
	\begin{axis}[
	  axis lines=middle,
	  xlabel={$y$},
	  ylabel={$g(y)$},
	  xmin=0, xmax=1.1,
	  ymin=0.2, ymax=1.1,
	  domain=0:1,
	  samples=100,
	  grid=major,
	  grid style={dashed, gray!50},
	  no markers,
	  clip=false,
	  legend cell align={left},
	]
	\addplot[thick, hkured] {exp(-x)};
	\addlegendentry{Baseline: $e^{-y}$}
	
	\addplot[thick, hkublue, domain=0:0.5] {0.25*(exp(-2*x) + 3 - 2*x)};
	\addlegendentry{SOCS}

	\addplot[thick, hkugreen, domain=0:0.5] {(1+x/2) * exp(-2*x) + (x/2) * (1-x)};
	\addlegendentry{SOCS (Random-Order)}	
	
	\addplot[thick, hkublue, domain=0.5:1] {exp(-2*x)*((1+exp(1))/4 + (exp(1)/2)*x)};
	
	\addplot[thick, hkugreen, domain=0.5:1] {exp(-2*x) * (1 + (1/2+e/4)*x)};	

	\end{axis}
	\end{tikzpicture}
	\caption{A comparison of the convergence rates of (1) the baseline algorithm that independently samples according to the fractional allocation $\matchrate^t$, (2) the SOCS for unweighted/vertex-weighted matching from this paper, and (3) the same SOCS algorithm in the random-order model.}
	\label{fig:socs-matching}
\end{figure}

\section{AdWords}
\label{sec:adwords}

\subsection{Preliminaries on Negative Association}

The results in this section rely on the concept of negatively associated random variables and their properties. 
We outline the necessary background for our analyses below and refer readers to the expository article by \citet{Wajc:NA:2017} for further information.

Given a set of random variables $X_1, X_2, \dots, X_n$ and a subset of indices $S \subseteq [n]$, we will write $X_S = (X_i)_{i \in S}$ for the subset of random variables with indices in $S$.

\begin{definition}[Negative Association]
	\label{def:na}
	Random variables $X_1, X_2, \dots, X_n$ are \emph{negatively associated} if for any disjoint subsets of indices $S, T \subset [n]$, and any functions $f$ and $g$ defined on $X_S$ and $X_T$ respectively that are both non-decreasing or both non-increasing, we have:
	\[
		\E \big[ f(X_S) \cdot g(X_T) \big] ~\le~ \E \big[ f(X_S) \big] \cdot \E \big[ g(X_T) \big]
		~. 
	\]
\end{definition}

The above inequality captures a notion of negative dependence that has been proven useful in many applications.
The next lemma considers two forms of intuitively negatively dependent random variables that satisfy negative association.

\begin{lemma}[Basic Forms of Negative Association]
	\label{lem:na-basic}
	A set of random variables $X_1, X_2, \dots, X_n$ is negatively associated if:
	\begin{enumerate}
	\item They are independent; or
	\item They are binary and satisfy $\sum_{i=1}^n X_i \le 1$.
	\end{enumerate}
\end{lemma}

Another appealing feature of negative association is its closure property under several natural operations.
Here, we only list one of them that will be relevant in our analysis. 

\begin{lemma}[Closure Property]
	\label{lem:na-closure}
	If (1) $X_1, X_2, \dots, X_n$ are negatively associated, (2) $Y_1, Y_2, \dots, Y_m$ are negatively associated, and (3) $(X_i)_{i \in [n]}$ is independent to $(Y_j)_{j \in [m]}$, then the joint distribution of $X_1, \dots, X_n, Y_1, \dots, Y_m$ satisfies negative association.
\end{lemma}

\subsection{Two-Way SOCS}
\label{sec:two-way-socs-adwords}

Recall that the allocation of larger bids is the main challenge of the AdWords problem.
We say that an offline agent $j$ makes a \emph{large bid} for an online item of type $i$, or that the item gets a large bid from agent $j$, if the bid is at least two-thirds of the agent's budget, i.e., $b_{ij} \ge \frac{2}{3} B_j$.
Otherwise, we say agent $j$ makes a \emph{small bid} for the item.
We will use an idea similar to the first OCS by \citet{FahrbachHTZ:FOCS:2020} to put a mild negative correlation into the allocation of large bids, and select independently in the other steps.
The algorithm is as follows:

\bigskip

\begin{tcolorbox}
    \textbf{Two-Way SOCS for AdWords}\\[2ex]
    When an online vertex with two-way surrogate type $i \sim \{j, k\}$ arrives at time step $t$:
    \begin{enumerate}
        \item Pick $m \in \{j, k\}$ uniformly at random.
        \item If type $i$ gets a large bid from $m$, i.e., $b_{i m} \ge \frac{2}{3} B_m$, mark this online vertex with $m$.
        \item If this is the second online vertex marked with $m$, then make the opposite selection to the first one (w.r.t.\ $m$).
        \item Otherwise, select $j$ or $k$ uniformly at random.
    \end{enumerate}
\end{tcolorbox}

\bigskip

Conditioned on the realization of online types, we may interpret the algorithm by considering a graph whose nodes are the time steps.
For each offline agent $j$, the graph has a clique of arcs among the time steps that receive large bids from agent $j$, to indicate that every pair of these steps could be the first two steps marked with $j$, in which case the algorithm would make opposite selections w.r.t.\ agent $j$.
This graph corresponds to the \emph{ex-ante dependence graph} in the OCS literature for Display Ads~\cite{FahrbachHTZ:FOCS:2020, GaoHHNYZ:FOCS:2021}, although the definition is different from the counterparts therein because of the different structures in these two problems.

Every time step that receives a large bid from offline agent $j$ is marked with $j$ with probability half.
The algorithm introduces a negative correlation between the first two time steps marked with $j$ by making the opposite selections (w.r.t.\ agent $j$) in these two time steps.
Describing it in the language of dependence graphs, the algorithm constructs a randomized matching in the \emph{ex-ante} dependence graph by matching the first two time steps marked with $j$ for every offline agent $j$.
Then, the algorithm makes random but opposite selections in each pair of matched time steps, and independent random selections in the unmatched time steps.
This matching corresponds to the \emph{ex-post dependence graph} in the OCS literature.

\begin{theorem}
	\label{thm:two-way-socs-adwords}
	\textbf{Two-Way SOCS for AdWords} has convergence rate:
	\[
		g(y_j)
		\le
		e^{-y_j}
		\cdot
		\bigg( \frac{3}{4} + \frac{1}{4} \cdot \frac{1+\nicefrac{y_j}{4}}{e^{\nicefrac{y_j}{4}}}\bigg)^2
		~.
	\]
\end{theorem}

This is better than the baseline convergence rate $e^{-y_j}$ because $1 + x \le e^x$ for any $x \ge 0$, with strict inequality for any $x \ne 0$.

The following analysis will directly consider an instance with both one-way and two-way surrogate types, so that the resulting bound will be more useful in the analysis of general SOCS and OCS for AdWords.
For an online item with one-way type $i \sim j$, the algorithm has no choice but to allocate it to agent $j$.
By contrast, the algorithm allocates items with two-way types according to the selections of the two-way SOCS.

Recall the definition of $y_j$ for the AdWords problem and that $\mass^t_{i \sim j}$ and $\mass^t_{i \sim \{j, k\}}$ are the probabilities of realizing one-way type $i \sim j$ and two-way type $i \sim \{j, k\}$ respectively at time step $t$.
For an instance with one-way and two-way surrogate types, the expression of $y_j$ simplifies to:
\[
	y_j = \sum_{t=1}^T \sum_{i \in I} \bigg( \mass^t_{i \sim j} + \frac{1}{2} \sum_{k \ne j} \mass^t_{i \sim \{j, k\}} \bigg) \cdot \frac{b_{ij}}{B_j} 
	~.
\]

The rest of the subsection will always focus on a fixed offline agent $j$.
Hence, we will suppress subscript $j$ in the following notations for simplicity.
We will further normalize its budget to be $B_j = 1$ without loss of generality.

Let $L$ and $S$ be the subsets of online types that receive large and small bids respectively from agent $j$.
We will refer to them as the large and small surrogate types respectively from now on.
\[
	L = \Big\{ i \in I : b_{ij} \ge \frac{2}{3} \Big\}
	\quad,\qquad
	S = \Big\{ i \in I : b_{ij} < \frac{2}{3} \Big\}
	~.
\]

Consider the following three kinds of contributions to $y_j$:
\begin{align*}
	y_{L1}^t
	&
	= \sum_{i \in L} \mass^t_{i \sim j} \cdot b_{ij}
	~, \\
	y_{L2}^t
	&
	= \frac{1}{2} \sum_{i \in L} \sum_{k \ne j} \, \mass^t_{i \sim \{j,k\}} \cdot b_{ij}
	~, \\
	y_{S}^t
	&
	= \sum_{i \in S} \bigg( \mass^t_{i \sim j} + \frac{1}{2} \sum_{k \ne j} \, \mass^t_{i \sim \{j,k\}} \bigg) \cdot b_{ij}
	~.
\end{align*}

That is, the large one-way surrogate types $i \sim j$ contribute $y^t_{L1}$ at time $t$.
Similarly, the large two-way surrogate types $i \sim \{j,k\}$ contribute $y^t_{L2}$.
Finally, the small surrogate types, including both one-way and two-way types, contribute $y^t_{S}$.
For the latter two, we further define:
\[
	y_{L2} = \sum_{t=1}^T y^t_{L2}
	\quad,\quad
	y_{S} = \sum_{t=1}^T y^t_{S}
	~.
\]

By definition, we have:
\[
	y_j = \sum_{t=1}^T y_{L1}^t + y_{L2} + y_{S}
	~.
\]

\medskip

The main result of this subsection is the next lemma, which will be useful in the proof of Theorem~\ref{thm:two-way-socs-adwords} and analysis of general SOCS for AdWords.

\begin{lemma}
	\label{lem:two-way-socs-adwords}
	The expected unspent fraction of agent $j$'s budget for the allocation selected by the \textbf{Two-Way SOCS for AdWords} is upper bounded by both:
	\[
		\prod_{t \in [T]} \big( 1 - y^t_{L1} - y^t_{L2} - y^t_S \big)
		~,
	\]
	and:
	\[
		e^{-y_j}
		\cdot
		\bigg( \frac{3}{4} + \frac{1}{4} \cdot \frac{1+\nicefrac{y_{L2}}{2}}{e^{\nicefrac{y_{L2}}{2}}}\bigg)
		\cdot
		\frac{1 + \nicefrac{y_S}{2}}{e^{\nicefrac{y_S}{2}}}
		\cdot
		\prod_{t=1}^T e^{y_{L1}^t} \big( 1 - y_{L1}^t \big)
		~.
	\]
\end{lemma}

The first bound is weakly better than the baseline convergence rate $e^{-y_j}$ due to $1-y \le e^{-y}$.
The second bound is also weakly better than the baseline because the three terms following $e^{y_j}$, related to $y_{L2}$, $y_S$, and $y_{L1}^t$'s respectively, are all at most $1$.
Further, if either $y_{L2}$ or $y_S$ was bounded away from $0$, then the above bound would be strictly better than the baseline convergence rate.
For example, this would be the case if we consider an instance with only two-way types and $y_j > 0$.

However, the above bounds would degenerate to the baseline bound $e^{-y_j}$ if $y_{L2} = y_S = 0$ and $y_{L1}^t$ was infinitesimally small at all time steps $t \in [T]$.
Ruling out such scenarios will be the main challenge in the subsequent applications of the theorem.

\paragraph{Roadmap.}
We devote the rest of this subsection to proving Lemma~\ref{lem:two-way-socs-adwords} and Theorem~\ref{thm:two-way-socs-adwords}.
We decompose this long proof into several parts and present them in separate subsubsections as follows.

The first part defines how we represent different sources of randomness, including the realization of online types and the internal randomness of the SOCS algorithm.
It further characterizes the worst-case scenario from an offline agent $j$'s point of view, simplifying the representation of randomness.
In essence, this part prepares the mathematical notations that we will use to prove Lemma~\ref{lem:two-way-socs-adwords} and Theorem~\ref{thm:two-way-socs-adwords}.

The second part is a decomposition lemma that separates the contributions from $y_S$, $y_{L2}$, and $y_{L1}^t$'s and expresses them as three expectations.
The proof of the decomposition crucially uses how we represent the sources of randomness in the first part.

The third part derives the upper bounds of these expectations.
These bounds correspond to the three terms in Lemma~\ref{lem:two-way-socs-adwords}.

Finally, we explain how to combine the first three parts to prove Theorem~\ref{thm:two-way-socs-adwords} and Lemma~\ref{lem:two-way-socs-adwords}.

\subsubsection{Representation of Randomness}

By definition, \textbf{Two-Way SOCS for AdWords} relies on two sets of randomness in its first and fourth steps.
Further, the realization of online types is also stochastic. 
We represent these sources of randomness by three sets of random variables.

\begin{enumerate}
	\item \emph{Realization of Online Types:~}
		Let $R^t_{i \sim j}$ and $R^t_{i \sim \{j, k\}}$ be the indicators for realizing an online vertex of surrogate types $i \sim j$ and $i \sim \{j, k\}$ respectively at time step $t$.
	\item \emph{Marks of Time Steps:~} 
		Let $M^t_{i \sim \{j,k\}}$ be independent and uniform over $\{0, 1\}$.
		It determines the choice of $m$ in the first step, \emph{at time $t$ and when the online type is $i \sim \{j,k\}$}.
	\item \emph{Choices of Agents:~}
		Let $C^t_{i \sim \{j,k\}}$ be independent and uniform over $\{0, 1\}$.
		It determines the choice of $j$ or $k$ in the fourth step, \emph{at time $t$ and when the online type is $i \sim \{j,k\}$}.
\end{enumerate}

Importantly, we use different random bits for different surrogate types to decide the marks in the first step, and the choices in the fourth step.
This treatment will be useful in the proof of the decomposition lemma in the next subsubsection.

\bigskip

Next, we establish the relevant properties of this representation of randomness.

\begin{lemma}
	\label{lem:adwords-negative-association}
	Random variables $R^t_{i \sim j}$ and $R^t_{i \sim \{j, k\}}$ are negatively associated.
\end{lemma}

\begin{proof}
	For any time step $t$, the random variables are binary and sum to $1$.
	Hence, they are of the second basic form of negative association (Lemma~\ref{lem:na-basic}).
	Further, the random variables for different time steps are independent by definition.
	Hence, the joint distribution also satisfies negative association (Lemma~\ref{lem:na-closure}).
\end{proof}

Let random variables $Y^t_S$, $Y^t_{L2}$, and $Y^t_{L1}$ denote the contributions of small surrogate types, large two-way surrogate types, and large one-way surrogate types respectively, at time step $t$.
In other words, $Y^t_S = b_{ij}$ if (1) $i \in S$, (2) time step $t$ realizes a surrogate type $i \sim j$ or $i \sim \{j,k\}$, and in the latter case, if (3) the two-way SOCS selects $j$;
$Y^t_S = 0$ otherwise.
Similarly, $Y^t_{L2} = b_{ij}$ if (1) $i \in L$, (2) time step $t$ realizes a surrogate type $i \sim \{j, k\}$, and (3) the two-way SOCS selects $j$;
$Y^t_{L2} = 0$ otherwise.
Finally, $Y^t_{L1} = b_{ij}$ if (1) $i \in L$ and (2) time step $t$ realizes a surrogate type $i \sim j$; $Y^t_{L1} = 0$. otherwise.
We do not need the third condition for large one-way types because the SOCS algorithm must select $j$.

Then, offline agent $j$'s expected unused budget can be written as:
\begin{equation}
	\label{eqn:adwords-unused-budget}
	\E \bigg[\: \bigg( 1 - \sum_{t = 1}^T \big( Y^t_S + Y^t_{L2} + Y^t_{L1} \big) \bigg)^+ \:\bigg]
	~.
\end{equation}

The next lemma characterizes the worst-case scenario for offline agent $j$'s convergence rate.

\begin{lemma}
	\label{lem:adwords-removing-arc}
	Conditioned on any realization of online types, Equation~\eqref{eqn:adwords-unused-budget} would weakly increase if we remove an arc in the \emph{ex-post} dependence graph.
\end{lemma}

\begin{proof}
	Consider removing an arc, say, between time steps $t$ and $t'$.
	Without loss of generality, we may consider the case when the online types at time steps $t$ and $t'$ are both two-way surrogate types that involve agent $j$; 
	otherwise, removing the arc does not affect the conditional expectation.
	Let $i \sim \{j, k\}$ and $i' \sim \{j, k'\}$ be the surrogate types at time steps $t$ and $t'$.
	With the arc, the contribution of these two time steps to the summation in \eqref{eqn:adwords-unused-budget} is uniform over support $\{ b_{ij}, b_{i'j} \}$.
	With the arc removed, the contribution becomes uniform over support $\{ 0, b_{ij}, b_{i'j}, b_{ij} + b_{i'j} \}$.
	The claim now follows by the Jensen's inequality and the convexity of $(1-x)^+$.	
\end{proof}

\begin{lemma}
	\label{lem:adwords-two-way-worst-case}
	Given any instance and any offline agent $j$, there is another instance with the same value of $y_j$ for agent $j$, but \textbf{Two-Way SOCS for AdWords} spends a weakly smaller portion of agent $j$'s budget in expectation.
	Further, in the new instance, each offline agent $k \ne j$ can be in the realized two-way surrogate type in at most one time step.
\end{lemma}

\begin{proof}
	We will modify the instance by making $T$ copies of each offline agent $k \ne j$, one for each time step, and changing the realization of the online types $i \sim \{j, k\}$ at each time step $t \in [T]$ to a new online type with the same $i$ and $j$ but changing $k$ to be the corresponding copy for that time step.
	By doing so, we remove all arcs in the dependence graphs except those due to agent $j$.
	By Lemma~\ref{lem:adwords-removing-arc}, Equation~\eqref{eqn:adwords-unused-budget} weakly increases after we remove these arcs.
\end{proof}

The rest of the argument will assume the conclusion of Lemma~\ref{lem:adwords-two-way-worst-case}.
In that case, the offline vertices other than $j$ do not affect the argument at all.
Hence, we will abuse notation and only refer to the relevant two-way surrogate types as $i \sim \{j, *\}$.
As a result, we can merge the two-way surrogate types with the same original type $i$, and the corresponding random variables will also be simplified as $R^t_{i \sim \{j, *\}}$, $M^t_{i \sim \{j, *\}}$, and $C^t_{i \sim \{j, *\}}$.

\smallskip
We will write:
\[
	R_S = \Big\{ R^t_{i \sim j}, R^t_{i \sim \{j, *\}} : i \in S \Big\}
	\quad,\quad 
	R_{L2} = \Big\{ R^t_{i \sim \{j, *\}} : i \in L \Big\}
	\quad,\quad 	
	R_{L1} = \Big\{ R^t_{i \sim j} : i \in L \Big\}	
\]
for the random variables for the realization of small types, large two-way types, and large one-way types respectively.
Define $M_{L2}, C_S, C_{L2}$ similarly.

We remark that defining $M_{L1}$ and $M_S$ would be redundant because the algorithm would not mark a time step that realizes a small or one-way surrogate type.
Similarly, defining $C_{L1}$ would be redundant because the algorithm has no choice but to allocate to the only agent for one-way surrogate types.

With this representation of randomness, $Y^t_S$, $Y^t_{L2}$, and $Y^t_{L1}$ depend on disjoint subsets of randomness, which will be useful in the proof of the decoupling lemma in the next part.
This is why we let $M^t$ and $C^t$ depend on the realization of surrogate types.
The next lemma summarizes the subsets of random variables that they depend on. 

\begin{lemma}
	\label{lem:adwords-disjoint-randomness}
	We have the following relation between the random variables $Y^t_S$, $Y^t_{L2}$, $Y^t_{L1}$ and the sources of randomness:
	\begin{itemize}
	\item $Y^t_S$ only depends on $R_S$ and $C_S$;
	\item $Y^t_{L2}$ only depends on $R_{L2}$, $M_{L2}$, and $C_{L2}$; and
	\item $Y^t_{L1}$ only depends on $R_{L1}$.
	\end{itemize}	
\end{lemma}

\subsubsection{Decoupling Three Types of Contributions}

Let us first introduce some notations to simplify the exposition of the decoupling lemma and our subsequent analysis. 
We will abuse notation and write $(t,i) \in R_{L1}$ (respectively, $(t,i) \in R_{L2}$) if $i \in L$ and $R^t_{i \sim j} = 1$ (respectively, $R^t_{i \sim \{j, *\}} = 1$), i.e., if the online vertex at time step $t$ has a large one-way type $i \sim j$ (respectively, a large two-way type $i \sim \{j, *\}$).
We further write the number of time steps that realize large two-way surrogate types as:
\[
	\big| R_{L2} \big| = \sum_{t=1}^T \sum_{i \in L} R^t_{i \sim \{j,*\}}
	~.
\]

\begin{lemma}[Decoupling Lemma]
	\label{lem:adwords-decomposition}
	We have:
	\begin{align*}
		\eqref{eqn:adwords-unused-budget}
		&
		~\le~
		\E \bigg[ \bigg( 1 - \sum_{t=1}^T Y^t_S \bigg)^+ \bigg]
		\cdot
		\E \bigg[ \Big( \frac{3}{4} + \frac{1}{4} \frac{|R_{L2}|+1}{2^{|R_{L2}|}} \Big) \prod_{(t,i) \in R_{L2}} \Big( 1 - \frac{1}{2} b_{ij} \Big) \bigg]
		\cdot
		\E \bigg[ \bigg( 1 - \sum_{t=1}^T Y^t_{L1} \bigg)^+ \bigg]
		~.
	\end{align*}
\end{lemma}

\begin{proof}
	First, we relax Equation~\eqref{eqn:adwords-unused-budget} to be:
	\begin{equation}
		\label{eqn:adwords-decompose-large-small-bids}	
		\E \bigg[ \bigg( 1 - \sum_{t = 1}^T \big( Y^t_S + Y^t_{L1} \big) \bigg)^+ \cdot \bigg( 1 - \sum_{t = 1}^T Y^t_{L2} \bigg)^+  \:\bigg]
	\end{equation}
	
	By Lemma~\ref{lem:adwords-disjoint-randomness}, the first part is independent to $M_{L2}$ and $C_{L2}$.
	
	We will next bound the expectation of the second part over the realization of $M_{L2}$ and $C_{L2}$, presented below as a standalone lemma, and deferring its proof to the end.

    \begin{lemma}
        \label{lem:decoupling-L2}
        Fix any realization of the sources of randomness other than $M_{L2}, C_{L2}$.
        We have:
        \[
            \E_{M_{L2}, C_{L2}} \bigg[ \bigg( 1 - \sum_{t = 1}^T Y^t_{L2} \bigg)^+ \:\bigg]
            ~\le~
            \Big( \frac{3}{4} + \frac{1}{4} \frac{|R_{L2}|+1}{2^{|R_{L2}|}} \Big) \prod_{(t,i) \in R_{L2}} \Big( 1 - \frac{1}{2} b_{ij} \Big)
            ~.
        \]
    \end{lemma}

	Putting it back to Equation~\eqref{eqn:adwords-decompose-large-small-bids}, it is at most:
	\[
		\E \bigg[ \bigg( 1 - \sum_{t = 1}^T \big( Y^t_S + Y^t_{L1} \big) \bigg)^+  \cdot \Big( \frac{3}{4} + \frac{1}{4} \frac{|R_{L2}|+1}{2^{|R_{L2}|}} \Big) \prod_{(t,i) \in R_{L2}} \Big( 1 - \frac{1}{2} b_{ij} \Big) \:\bigg]
		~,
	\]
	with the expectation taken over the random realization of $R_{L1}$, $R_{L2}$, $R_S$, and $C_S$.
	
	Next, we consider the following two functions:
	\begin{align*}
		f(R_S, R_{L1})
		&
		~=~ \E_{C_S} \bigg[ \bigg( 1 - \sum_{t = 1}^T \big( Y^t_S + Y^t_{L1} \big) \bigg)^+ \:\bigg] 
		~,
		\\
		g(R_L)
		&
		~=~ \Big( \frac{3}{4} + \frac{1}{4} \frac{|R_{L2}|+1}{2^{|R_{L2}|}} \Big) \prod_{(t,i) \in R_{L2}} \Big( 1 - \frac{1}{2} b_{ij} \Big)
		~.
	\end{align*}

	These two functions depend on disjoint subsets of variables, and are both non-increasing.
	By the definition of negative association (Definition~\ref{def:na}) and that $R_S, R_{L1}, R_{L2}$ are negatively associated (Lemma~\ref{lem:adwords-negative-association}), we have:
	\[
		\E \big[ f(R_S, R_{L1}) \cdot g(R_{L2}) \big] ~\le~ \E \big[ f(R_S, R_{L1}) \big] \cdot \E \big[ g(R_{L2}) \big]
		~.
	\]
	
	In other words, we have successfully decoupled the contribution from the large two-way surrogate types from the rest, bounding the probability by:
	\[
		\E \bigg[ \bigg( 1 - \sum_{t = 1}^T \big( Y^t_S + Y^t_{L1} \big) \bigg)^+ \:\bigg] 
		~\cdot~
		\E \bigg[ \Big( \frac{3}{4} + \frac{1}{4} \frac{|R_{L2}|+1}{2^{|R_{L2}|}} \Big) \prod_{(t,i) \in R_{L2}} \Big( 1 - \frac{1}{2} b_{ij} \Big) \bigg]
		~.
	\]
	
	Finally, we use the same method to decouple the contributions from the small surrogate types and large one-way surrogate types.
	We bound the first part above by:
	\[
		\E \bigg[ \bigg( 1 - \sum_{t = 1}^T Y^t_S \bigg)^+ \cdot \bigg( 1 - \sum_{t=1}^T Y^t_{L1} \bigg)^+ \:\bigg]
		~.
	\]
	
	Consider two functions:
	\begin{align*}
		f(R_S) & ~=~ \E_{C_S} \bigg[\: \bigg( 1 - \sum_{t = 1}^T Y^t_S \bigg)^+ \:\bigg] ~, \\
		g(R_{L1}) & ~=~ 1 - \sum_{t=1}^T Y^t_{L1} = 1 - \sum_{(t,i) \in R_{L1}} b_{ij} ~.
	\end{align*}
	
	These two functions depend on disjoint subsets of variables, and are both non-increasing.
	By the definition of negative association and that $R_S, R_{L1}$ are negatively associated, we have:
	\begin{align*}
		\E \big[ f(R_S) \cdot g(R_{L1}) \big]
		&
		~\le~ \E \big[ f(R_S) \big] \cdot \E \big[ g(R_{L1}) \big] \\[1ex]
		&
		~=~ \E \bigg[ \bigg( 1 - \sum_{t = 1}^T Y^t_S \bigg)^+ \:\bigg] \cdot \E \bigg[ \bigg( 1 - \sum_{t = 1}^T Y^t_{L1} \bigg)^+ \:\bigg] 
		~.
	\end{align*}
	
	This finishes the proof of the lemma.
\end{proof}

\begin{proof}[Proof of Lemma~\ref{lem:decoupling-L2}]
	By definition, the SOCS algorithm marks each time step in $R_{L2}$ with agent $j$ independently with probability half.
	
    With probability $\frac{|R_{L2}|+1}{2^{|R_{L2}|}}$, only zero or one time step in $R_{L2}$ is marked with $j$.
	In this case, the realization of $Y^t_{L2}$ is as follows.
	If the surrogate type at time step $t$ is \emph{not} a large two-way surrogate type, $Y^t_{L2} = 0$ with certainty.
	Otherwise, i.e., if there exists an online type $i \in L$ such that $(t,i) \in R_{L2}$, $Y^t_{L2}$ distributes independently and uniformly over $\{0, b_{ij}\}$.
	Without loss of generality, we may assume that the SOCS algorithm selects agent $j$ in such a time step $t$ if $C^t_{i \sim \{j,*\}} = 1$, and selects the other agent if $C^t_{i \sim \{j,*\}} = 0$.
	Then, we can bound the expectation by:
	\begin{align*}
		\E_{C_{L2}} \bigg[ \bigg( 1 - \sum_{t = 1}^T Y^t_{L2} \bigg)^+ \:\bigg]
		&
		~\le~ 
		\E_{C_{L2}} \bigg[ \prod_{t=1}^T\big( 1 - Y^t_{L2} \big) \:\bigg] \\
		&
		~=~
		\E_{C_{L2}} \bigg[ \prod_{(t,i) \in R_{L2}} \big( 1 -C^t_{i \sim \{j,*\}} \cdot b_{ij} \big) \:\bigg] \\
		&
		~=~ 
		\prod_{(t,i) \in R_{L2}} \Big( 1 - \frac{1}{2} b_{ij} \Big)
		~.
	\end{align*}
	
	With probability $1 - \frac{|R_{L2}|+1}{2^{|R_{L2}|}}$, the algorithm marks at least two time steps in $R_{L2}$ with $j$.
	By definition, the SOCS algorithm selects oppositely w.r.t.\ agent $j$ in the first two of these steps.
	Let these two steps and the online types therein be $(t_1,i_1)$ and $(t_2,i_2)$.
	Then, either $Y^{t_1}_{L2} = b_{i_1j}$ and $Y^{t_2}_{L2} = 0$, or $Y^{t_1}_{L2} = 0$ and $Y^{t_2}_{L2} = b_{i_2j}$, each with probability a half.
	Hence, we will bound the expectation by:
	\begin{align*}
		\E_{C_{L2}} \bigg[ \bigg( 1 - \sum_{t = 1}^T Y^t_{L2} \bigg)^+ \:\bigg]
		&
		~\le~ \E_{C_{L2}} \bigg[ \big( 1-  Y^{t_1}_{L2} - Y^{t_2}_{L2} \big) \prod_{t \ne t_1, t_2}^T\big( 1 - Y^t_{L2} \big) \:\bigg]\\
		&
		~=~ \E_{C_{L2}} \Big[ 1-  Y^{t_1}_{L2} - Y^{t_2}_{L2} \Big] \cdot \E_{C_{L2}} \bigg[ \prod_{t \ne t_1, t_2}^T\big( 1 - Y^t_{L2} \big) \:\bigg]
		~.
	\end{align*}
	
	The first part is at most:
	\begin{align*}
		\E_{C_{L2}} \Big[ 1-  Y^{t_1}_{L2} - Y^{t_2}_{L2} \Big]
		&
		= 1 - \frac{1}{2} b_{i_1j} - \frac{1}{2} b_{i_2j} \\[1ex]
		&
		\le \frac{3}{4} \Big( 1 - \frac{1}{2} b_{i_1j} \Big) \Big( 1 - \frac{1}{2} b_{i_2j} \Big)
		~,
	\end{align*}
	where the inequality follows by the assumption that these online types get large bids from offline agent $j$, i.e., $b_{i_1j}, b_{i_2j} \ge \frac{2}{3}$.
	
	The second part is at most:
	\[
		\prod_{(t, i) \in R_{L2} : t \ne t_1, t_2} \Big( 1 - \frac{1}{2} b_{ij} \Big)
		~.
	\]
	
	Combining the two parts yields an upper bound that is smaller than the first bound by a $\frac{3}{4}$ factor, i.e.:
	\[
		\frac{3}{4} \prod_{(t,i) \in R_{L2}} \Big( 1 - \frac{1}{2} b_{ij} \Big)
		~.
	\]
	
	Combining the two cases, we get the inequality of the lemma:
	\[
		\E_{M_{L2}, C_{L2}} \bigg[ \bigg( 1 - \sum_{t=1}^T Y^t_{L2} \bigg)^+ \:\bigg]
		~\le~ 
		\E \bigg[ \Big( \frac{3}{4} + \frac{1}{4} \frac{|R_{L2}|+1}{2^{|R_{L2}|}} \Big) \prod_{(t,i) \in R_{L2}} \Big( 1 - \frac{1}{2} b_{ij} \Big) \bigg]
		~.
	\]   
\end{proof}

\subsubsection{Upper Bounds for Three Types of Contribution}

Given the decomposition in Lemma~\ref{lem:adwords-decomposition}, we will next bound these three kinds of contributions in the following three lemmas.

\begin{lemma}
	\label{lem:adwords-S}
	The contribution from the small surrogate types is upper bounded as follows:
	\[
		\E \bigg[ \bigg( 1 - \sum_{t=1}^T Y^t_S \bigg)^+ \: \bigg]
		~\le~ 
		e^{-\frac{3}{2} y_{S}} \Big( 1 + \frac{y_{S}}{2} \Big)
		~.
	\]
\end{lemma}

\begin{proof}
	By the definition of $Y^t_S$, we have:
	\[
		\E \, Y^t_S = y^t_S
		~.
	\]
	
	Next, consider auxiliary Bernoulli random variables $\bar{Y}^t_S$, $t \in [T]$, obtained by first drawing $Y^t_S$ and then letting $\bar{Y}^t_S = \frac{2}{3}$ with probability $\frac{3}{2} Y^t_S$ and $\bar{Y}^t_S = 0$ otherwise.
	Here we use the fact that $0 \le Y^t_S \le \frac{2}{3}$, as it is either $0$ of the contribution from a small surrogate type.
	By definition:
	\[
		\E\,\bar{Y}_S^t = \E\,Y_S^t = y_S^t
		~.
	\]
	
	Further, by the convexity of function $(1-x)^+$ and Jensen's inequality, we have:
	\[
		\mathbf{E}\,\bigg( 1 - \sum_{t=1}^T Y_S^t \bigg)^+ ~\le~ \mathbf{E}\,\bigg( 1 - \sum_{t=1}^T \bar{Y}_S^t \bigg)^+
		~.
	\]
	
	Finally, consider independent Poisson random variables $Z^t$ with parameter $\frac{3}{2} \cdot \mathbf{E}\,Y_S^t = \frac{3}{2} y^t_S$ for every time step $t \in [T]$, coupled with $\bar{Y}_S^t$ such that $Z^t = 0$ whenever $\bar{Y}_S^t = 0$.
	By definition:
	\[
		\frac{2}{3} \cdot \E\,Z^t = \E\,\bar{Y}_S^t = \E\, Y_S^t = y_S^t
		~.
	\]
	
	By the convexity of function $(1-x)^+$, we have:
	\begin{equation}
		\label{eqn:adwords-small-bids-poisson}
		\mathbf{E}\,\bigg( 1 - \sum_{t=1}^T \bar{Y}_S^t \bigg)^+ ~\le~ \mathbf{E}\,\bigg( 1 - \frac{2}{3} \sum_{t=1}^T Z^t \bigg)^+
		~.
	\end{equation}
	
	This expectation on the right is easy to compute because $\sum_{t=1}^T Z^t$ is a Poisson random variable with parameter:
	\[
		\frac{3}{2} \sum_{t=1}^T y_S^t ~=~ \frac{3}{2} y_S
		~.
	\]
	
	By the definition of Poisson random variables, the right-hand-side of Equation~\eqref{eqn:adwords-small-bids-poisson} equals $1$ with probability $e^{-\frac{3}{2} y_{S}}$, and $\frac{1}{3}$ with probability $\frac{3}{2} y_{S} \cdot e^{-\frac{3}{2}y_{S}}$.
	Thus, the expectation equals:
	\[
		e^{-\frac{3}{2} y_{S}} \cdot 1 + \frac{3}{2} y_{S} \cdot e^{-\frac{3}{2} y_{S}} \cdot \frac{1}{3} = e^{-\frac{3}{2} y_{S}} \big( 1 + \frac{y_{S}}{2} \big)
		~.
	\]
\end{proof}

\begin{lemma}
	\label{lem:adwords-L2}
	The contribution of large two-way types is upper bounded by:
	\[
		\E \bigg[ \Big( \frac{3}{4} + \frac{1}{4} \frac{|R_{L2}|+1}{2^{|R_{L2}|}} \Big) \prod_{(t,i) \in R_{L2}} \Big( 1 - \frac{1}{2} b_{ij} \Big) \bigg]
		\le e^{-y_{L2}} \bigg( \frac{3}{4} + \frac{1}{4} \frac{1+\nicefrac{y_{L2}}{2}}{e^{\nicefrac{y_{L2}}{2}}}\bigg)
		~.
	\]
\end{lemma}

\begin{proof}
	We introduce the following notations to denote the sum of probabilities for realizing a large two-way online type, at a time step $t$ and throughout the process:
	\[
		\mass^t_{L2} = \sum_{i \in L} \, \mass^t_{i \sim \{j, *\}}
		\quad,\quad
		\mass_{L2} = \sum_{t=1}^T \mass^t_{L2}
		~.
	\]
	
	By $\frac{2}{3} \le b_{ij} \le 1$, we have:
	\begin{equation}
	    \label{eqn:adwords-L2-mass-y-bound}
	    \begin{aligned}
			\frac{1}{2} \mass^t_{L2}
			&
			~\ge~ \sum_{i \in L} \, \mass^t_{i \sim \{j, *\}} \cdot \frac{1}{2} \cdot b_{ij} ~=~ y^t_{L2}
			~, \\
			\frac{1}{3} \mass^t_{L2}
			&
			~\le~ \sum_{i \in L} \, \mass^t_{i \sim \{j, *\}} \cdot \frac{1}{2} \cdot b_{ij} ~=~ y^t_{L2}
			~.
		\end{aligned}
	\end{equation}
	
	Summing over all time steps $t \in [T]$ further gives:
	\[
		\frac{1}{2} \mass_{L2} ~\ge~ y_{L2} ~\ge~ \frac{1}{3} \mass_{L2}
		~.
	\]
	
	We now explicitly write down the expectation in the lemma by summing over all possible $R_{L2}$, the probability of realizing it multiplied by the expression in the expectation.
	Recall that $(t, i) \in R_{L2}$ stands for the event that the online vertex at time $t$ has a large two-way type $i \sim \{j, *\}$.
	Further, let $t \notin R_{L2}$ denote the event that the online type realized at step $t$ is \emph{not} a large two-way type, i.e., if $R^t_{i \sim \{j, *\}} = 0$ for all $i \in L$.
	The expectation in the lemma equals:
	\begin{align}
		&
		\sum_{R_{L2}} \underbrace{\prod_{t \notin R_{L2}} \Big(1 - \mass^t_{L2} \Big) \prod_{(t,i) \in R_{L2}} \mass^t_{i \sim \{j, *\}}}_{\text{probability of realizing $R_{L2}$}} \,\cdot\, \Big( \frac{3}{4} + \frac{1}{4} \frac{|R_{L2}|+1}{2^{|R_{L2}|}} \Big) \prod_{(t,i) \in R_{L2}} \Big( 1 - \frac{1}{2} b_{ij} \Big)
		\notag \\[1ex]
		& \qquad
		= 
		\sum_{R_{L2}} \Big( \frac{3}{4} + \frac{1}{4} \frac{|R_{L2}|+1}{2^{|R_{L2}|}} \Big) \prod_{t \notin R_{L2}} \Big(1 - \mass^t_{L2} \Big) \prod_{(t,i) \in R_{L2}} \mass^t_{i \sim \{j, *\}} \Big( 1 - \frac{1}{2} b_{ij} \Big)
		~.
		\label{eqn:adwords-socs-L2}
	\end{align}

	To further simplify it, we consider the following generating function:
	\begin{align*}
		h(x)
		&
		= \prod_{t=1}^T \bigg( 1 - \mass^t_{L2} + x \cdot \sum_{i \in L} \mass^t_{i \sim \{j,*\}} \Big( 1 - \frac{1}{2} b_{ij} \Big) \bigg) \\
		&
		= \prod_{t=1}^T \Big( 1 - \mass^t_{L2} + x \cdot \big( \mass^t_{L2} - y^t_{L2} \big) \Big)
		~.
	\end{align*}
	
	Then, the above Equation~\eqref{eqn:adwords-socs-L2} can be written as:
	\begin{equation}
		\label{eqn:adwords-socs-L2-gen-func}
		\frac{3}{4} \cdot h(1) + \frac{1}{4} \cdot \bigg( \frac{1}{2} h' \Big(\frac{1}{2}\Big) + h \Big(\frac{1}{2}\Big) \bigg)
		~.
	\end{equation}

	In particular:
	\begin{align}
		h(1)
		&
		~=~ \prod_{t=1}^T \big( 1 - y^t_{L2} \big)
		\notag \\
		&
		~\le~ \prod_{t=1}^T e^{- y^t_{L2}} = e^{-y_{L2}}
		~.
		\label{eqn:adwords-socs-L2-g-one}
	\end{align}
	
	We also have:
	\begin{align*}
		h \Big( \frac{1}{2} \Big) 
		&
		~=~ \prod_{t=1}^T \Big( 1 - \frac{1}{2} \mass^t_{L2} - \frac{1}{2} y^t_{L2} \Big) 
		~, \\
		h' \Big( \frac{1}{2} \Big)
		&
		~=~ \sum_{t=1}^T \big( \mass^t_{L2} - y^t_{L2} \big) \prod_{t' \ne t} \Big( 1 - \frac{1}{2} \mass^t_{L2} - \frac{1}{2} y^t_{L2} \Big) 
		~.
	\end{align*}
	
	Hence, the sum of $h(\frac{1}{2})$ and $\frac{1}{2} h'(\frac{1}{2})$ can be written as:
	\begin{equation}
		\label{eqn:adwords-socs-L2-g-half}
		\prod_{t=1}^T \Big( 1 - \frac{1}{2} \mass^t_{L2} - \frac{1}{2} y^t_{L2} \Big) \bigg( 1 + \sum_{t=1}^T \frac{\mass^t_{L2} - y^t_{L2}}{2 - \mass^t_{L2} - y^t_{L2}} \bigg)
		~.
	\end{equation}
	
	If all $\mass^t_{L2}$ and $y^t_{L2}$ are infinitesimally small, the above would further simplify to:
	\[
		e^{-\frac{1}{2} \mass_{L2} - \frac{1}{2} y_{L2}} \cdot \Big( 1 + \frac{1}{2} \mass_{L2} - \frac{1}{2} y_{L2} \Big) 
		~.
	\]
	
	We next prove that it is an upper bound for Equation~\eqref{eqn:adwords-socs-L2-g-half} even when the variables are not infinitesimal, through a hybrid argument.
	For any $0 \le t \le T$, define $A^t$ to be:
	\[
		\exp \bigg( - \frac{1}{2} \sum_{t' \le t} \big( \mass^{t'}_{L2} + y^{t'}_{L2} \big) \bigg) \prod_{t < t' \le T} \Big( 1 - \frac{1}{2} \mass^{t'}_{L2} - \frac{1}{2} y^{t'}_{L2} \Big) \bigg( 1 + \frac{1}{2} \sum_{t' \le t} \big( \mass^{t'}_{L2} - y^{t'}_{L2} \big)  + \sum_{t < t' \le T} \frac{\mass^{t'}_{L2} - y^{t'}_{L2}}{2 - \mass^{t'}_{L2} - y^{t'}_{L2}} \bigg) 
		~.
	\]
	
	Observe that $A^0$ is the original Equation~\eqref{eqn:adwords-socs-L2-g-half} and $A^T$ is the claimed upper bound.
	Hence, we just need to show for any $t \in [T]$ that $A^{t-1} \le A^t$.
	Comparing the two sides of this inequality, it is sufficient to prove that:
	\begin{align*}
		&	
		\Big( 1 - \frac{1}{2} \mass^t_{L2} - \frac{1}{2} y^t_{L2} \Big) \bigg( 1 + \frac{1}{2} \sum_{t' < t} \big( \mass^{t'}_{L2} - y^{t'}_{L2} \big)  + \sum_{t \le t' \le T} \frac{\mass^{t'}_{L2} - y^{t'}_{L2}}{2 - \mass^{t'}_{L2} - y^{t'}_{L2}} \bigg) \\
		& \qquad\qquad
		\le~ 
		e^{- \frac{1}{2} \mass^t_{L2} - \frac{1}{2} y^t_{L2}} \bigg( 1 + \frac{1}{2}  \sum_{t' \le t} \big( \mass^{t'}_{L2} - y^{t'}_{L2} \big)  + \sum_{t < t' \le T} \frac{\mass^{t'}_{L2} - y^{t'}_{L2}}{2 - \mass^{t'}_{L2} - y^{t'}_{L2}} \bigg)
		~.
	\end{align*}	
	
	For notational simplicity, denote the common terms in the second parts of both sides as:
	\[
		\Delta ~=~ \frac{1}{2} \sum_{t' < t} \big( \mass^{t'}_{L2} - y^{t'}_{L2} \big)  + \sum_{t < t' \le T} \frac{\mass^{t'}_{L2} - y^{t'}_{L2}}{2 - \mass^{t'}_{L2} - y^{t'}_{L2}}
		~.
	\]
	
	The inequality simplifies to:
	\[
		\Big( 1 - \frac{1}{2} \mass^t_{L2} - \frac{1}{2} y^t_{L2} \Big) \bigg( 1 + \frac{\mass^t_{L2} - y^t_{L2}}{2 - \mass^t_{L2} - y^t_{L2}} + \Delta \bigg) 
		~\le~ 
		e^{- \frac{1}{2} \mass^t_{L2} - \frac{1}{2} y^t_{L2}} \cdot \Big( 1 + \frac{1}{2} \mass^t_{L2} - \frac{1}{2}  y^t_{L2} + \Delta \Big)
		~.
	\]
	
	Since $\Delta \ge 0$ and its coefficients satisfy $1 - \frac{1}{2} \mass^t_{L2} - \frac{1}{2} y^t_{L2} \le e^{- \frac{1}{2} \mass^t_{L2} - \frac{1}{2} y^t_{L2}}$, we only need to prove the inequality for the remaining terms.
	After merging terms on the left, this is:
	\[
		1 - y^t_{L2} \le e^{- \frac{1}{2} \mass^t_{L2} - \frac{1}{2} y^t_{L2}} \cdot \Big( 1 + \frac{1}{2} \mass^t_{L2} - \frac{1}{2} y^t_{L2} \Big)
		~.
	\]
	
	The right-hand-side is decreasing in $\mass^t_{L2}$ because $e^{-x} (1+x)$ is decreasing in $x \ge 0$.
	Further, recall that $\mass^t_{L2} \le 3 y^t_{L2}$ (Equation~\eqref{eqn:adwords-L2-mass-y-bound}).
	Hence, the above inequality reduces to:
	\[
		1 - y^t_{L2} \le e^{-2 y^t_{L2}} \big(1 + y^t_{L2} \big)
	\]
	which holds for all $y^t_{L2} \ge 0$ (Appendix~\ref{app:e-2x}).
		
	In sum, we have:
	\[
		\frac{1}{2} h' \Big( \frac{1}{2} \Big) + h \Big( \frac{1}{2} \Big) \le e^{-\frac{1}{2} \mass_{L2} - \frac{1}{2} y_{L2}} \Big( 1 + \frac{1}{2} \mass_{L2} - \frac{1}{2} y_{L2} \Big)
		~.
	\]
	
	Putting this and Equation~\eqref{eqn:adwords-socs-L2-g-one} back to Equation~\eqref{eqn:adwords-socs-L2-gen-func}, the expectation in the lemma is upper bounded by:
	\[
		\frac{3}{4} e^{-y_{L2}} + \frac{1}{4} e^{-\frac{1}{2} \mass_{L2} - \frac{1}{2} y_{L2}} \Big( 1 + \frac{1}{2} \mass_{L2} - \frac{1}{2} y_{L2} \Big)
	\]
	
	This is decreasing in $\mass_{L2}$ because $e^{-x} (1+x)$ is decreasing in $x \ge 0$.
	Further, recall that we have $\mass_{L2} \ge 2 y_{L2}$ (Equation~\eqref{eqn:adwords-L2-mass-y-bound}).
    Hence, this is at most:
	\[
		\frac{3}{4} e^{-y_{L2}} + \frac{1}{4} e^{-\frac{3}{2} y_{L2}} \Big( 1 + \frac{1}{2} y_{L2} \Big)
		~.
	\]
\end{proof}

\begin{lemma}
	\label{lem:adwords-L1}
	The contribution of large one-way types is bounded by:
	\[
		\E \bigg[ \bigg( 1 - \sum_{t=1}^T Y^t_{L1} \bigg)^+ \bigg]
		~\le~
		\prod_{t=1}^T \big(1 - y^t_{L1} \big)
		~.
	\]
\end{lemma}

\begin{proof}
	First, observe that:
	\[
		\bigg( 1 - \sum_{t=1}^T Y^t_{L1} \bigg)^+ \le~ \prod_{t=1}^T \big( 1 - Y^t_{L1} \big)
		~.
	\]
	
	The lemma follows by the independence of $Y^t_{L1}$ at different time $t \in [T]$, and $\E\,Y^t_{L1} = y^t_{L1}$.
\end{proof}

\subsubsection{Proof of Lemma~\ref{lem:two-way-socs-adwords}}

The first bound in Lemma~\ref{lem:two-way-socs-adwords} follows from relaxing Equation~\eqref{eqn:adwords-unused-budget} by removing all arcs in the \emph{ex-post graph}
(Lemma~\ref{lem:adwords-removing-arc}).
Then, we have:
\begin{align*}
	\eqref{eqn:adwords-unused-budget} 
	&
	~\le~ 
	\E \bigg[\, \prod_{t \in [T]} \big( 1 - Y^t_{L1} - Y^t_{L2} - Y^t_S \big) \,\bigg] \\
	&
	~=~ \prod_{t \in [T]} \E \,\big( 1 - Y^t_{L1} - Y^t_{L2} - Y^t_S \big) \\
	&
	~=~ \prod_{t \in [T]} \big( 1 - y^t_{L1} - y^t_{L2} - y^t_S \big)
	~.
\end{align*}

The second bound in Lemma~\ref{lem:two-way-socs-adwords} follows by applying the bounds in Lemmas~\ref{lem:adwords-S}, \ref{lem:adwords-L2}, and \ref{lem:adwords-L1} to the three terms on the right-hand-side of Lemma~\ref{lem:adwords-decomposition}.

\subsubsection{Proof of Theorem~\ref{thm:two-way-socs-adwords}}

We apply the second bound of Lemma~\ref{lem:two-way-socs-adwords} with $y_{L1} = 0$ because there is no one-way type in this setting.
The expected unspent fraction of agent $j$'s budget is at most:
\[
	e^{-y_j}
	\cdot
	\bigg( \frac{3}{4} + \frac{1}{4} \cdot \frac{1+\nicefrac{y_{L2}}{2}}{e^{\nicefrac{y_{L2}}{2}}}\bigg)
	\cdot
	\frac{1 + \nicefrac{y_S}{2}}{e^{\nicefrac{y_S}{2}}}
\]
where $y_{L2} + y_S = y_j$.
Hence, the convergence rate can be written as:
\[
	e^{-y_j}
	\cdot
	\max_{\substack{\text{$y_S, y_{L2} \ge 0 :$}\\ \text{$y_{L2} + y_S = y_j$}}} 
	\bigg( \frac{3}{4} + \frac{1}{4} \cdot \frac{1+\nicefrac{y_{L2}}{2}}{e^{\nicefrac{y_{L2}}{2}}}\bigg)
	\cdot
	\frac{1 + \nicefrac{y_S}{2}}{e^{\nicefrac{y_S}{2}}}
	~.
\]

Symmetrize the two terms by relaxing the second term related to $y_S$ in the above maximization, it is at most:
\[
	\bigg( \frac{3}{4} + \frac{1}{4} \cdot \frac{1+\nicefrac{y_{L2}}{2}}{e^{\nicefrac{y_{L2}}{2}}}\bigg)
	\cdot
	\bigg( \frac{3}{4} + \frac{1}{4} \cdot \frac{1+\nicefrac{y_S}{2}}{e^{\nicefrac{y_S}{2}}}\bigg)
	~.
\]

By the concavity of (Appendix~\ref{app:adwords-convergence-function}):
\begin{equation}
	\label{eqn:adwords-convergence-function}
	\log \left( \frac{3}{4} + \frac{1}{4} \cdot \frac{1+\nicefrac{y}{2}}{e^{\nicefrac{y}{2}}} \right)
	~,
\end{equation}
this is maximized when $y_S = y_{L2} = \nicefrac{y_j}{2}$.
Hence, we get the convergence rate in Theorem~\ref{thm:two-way-socs-adwords}:
\[
	g(y_j) = e^{-y_j} \cdot \bigg( \frac{3}{4} + \frac{1}{4} \cdot \frac{1+\nicefrac{y_j}{4}}{e^{\nicefrac{y_j}{4}}}\bigg)^2
	~.
\]

\subsection{Stochastic AdWords Linear Program}
\label{sec:adwords-lp}

For any subset $S \subseteq T \times I$ of pairs of time step and online type, we define an auxiliary function $\bar{v}_j(S)$ to denote the value we would get from allocating them to $j$, i.e., if we allocate to agent $j$ all online item whose type $i$ and arrival time $t$ satisfy $(t,i) \in S$.
Recall that $R^t_i$ is the indicator of whether the online item at time step $t$ has type $i \in I$.
By definition, this auxiliary function is:
\[
	\bar{v}_j(S) = \E \bigg[ \min \bigg\{\: \sum_{(t,i) \in S} R^t_i \cdot b_{ij} \,,\, B_j \:\bigg\} \:\bigg]
	~. 
\]

For $x^t_{ij} = \mass^t_i \cdot \matchrate^t_{ij}$, we will consider the following \emph{Stochastic AdWords LP}:
\begin{align}
    \text{maximize} \quad & \sum_{i \in I} \sum_{j \in J} \sum_{t \in [T]} b_{ij} \cdot x_{ij}^t \notag \\
    \text{subject to} \quad & \sum_{j \in J} x_{ij}^t \le \mass_i^t && \forall i \in I, \forall t \in [T] \notag \\
    & \sum_{(t,i) \in S} b_{ij} \cdot x_{ij}^t \le \bar{v}_j(S) && \forall j \in J, \forall S \subseteq T \times I \label{eqn:adwords-lp} \\
    & x_{ij}^t \ge 0 && \forall i \in I, \forall j \in J, \forall t \in [T] \notag 
\end{align}

We will next show that it is an LP relaxation of the Stochastic AdWords problem, and discuss whether it can be solved in polynomial-time.

\begin{lemma}[Optimality]
    \label{lem:adwords-lp-optimality}
    The optimal objective value of the Stochastic AdWords LP is greater than or equal to the expected objective of the optimal allocation in hindsight.
\end{lemma}

\begin{proof}
	Consider the offline optimal allocation.
	Let $\eta_{Sj}$	be the probability that agent $j$ gets a subset of items $S \subseteq T \times I$,
	where $(t,i) \in S$ means that agent $j$ gets the online item at time $t$ and the item's type is $i$.
	We will abuse notation and write $v_j(S)$ for agent $j$'s value for receiving this subset of items, effectively dropping the time step from each $(t, i) \in S$ before feeding the subset to function $v_j$.
	Further, let:
	\[
		b_j(S) = \sum_{(i,t) \in S} b_{ij}
	\]
	be the sum of agent $j$'s bids for these items.
	
	We now define a feasible LP solution as follows.
	For any $i \in I$, $j \in J$, and $t \in [T]$, let
	\begin{equation}
		\label{eqn:adwords-lp-solution}
		x^t_{ij} = \sum_{S \subseteq T \times I \::\: (t,i) \in S} \eta_{Sj} \cdot \frac{v_j(S)}{b_j(S)}
		~.
	\end{equation}
	
	We will next verify that this LP solution's objective value is equal to the optimal objective of the AdWords instance, and that it satisfies all LP constraints.
	
	\paragraph{Objective Value.}
	By definition, we have:
	\begin{align*}
		\sum_{i \in I} \sum_{j \in J} \sum_{t \in [T]} b_{ij} x^t_{ij}
		&
		~=~ 
		\sum_{i \in I} \sum_{j \in J} \sum_{t \in [T]} b_{ij} \sum_{S \subseteq T \times I \::\: (t,i) \in S} \eta_{Sj} \cdot \frac{v_j(S)}{b_j(S)} \\
		&
		~=~
		\sum_{j \in J} \sum_{S \subseteq T \times I} \eta_{Sj} \cdot \sum_{(t,i) \in S} b_{ij} \cdot\frac{v_j(S)}{b_j(S)} \\[1ex]
		&
		~=~
		\sum_{j \in J} \sum_{S \subseteq T \times I} \eta_{Sj} \cdot v_j(S) 
		~=~ \OPT
		~.
	\end{align*}
	
	\paragraph{First Set of Constraints.}
	For any $i \in I$ and any $t \in [T]$, we have:
	\[
		\sum_{j \in J} x^t_{ij}
		~\le~ 
		\sum_{j \in J} \sum_{S \subseteq T \times I \::\: (t,i) \in S} \eta_{Sj}
		~.
		\tag{definition of $x^t_{ij}$ and $\bar{v}_j(S) \le b_j(S)$}
	\] 
	
	The right-hand-side is the probability that the offline optimal solution allocates an online item of type $i$ to some agent $j$ at time $t$. 
	This is at most $\mass^t_j$, the probability that such an item arrives at time $t$.
	
	\paragraph{Second Set of Constraints.}
	The left-hand-side of Constraint~\eqref{eqn:adwords-lp} is equal to the actual value that we get from the subset of items in $S$ allocated to agent $j$ by the offline optimal solution.
	Note that when agent $j$'s sum of bids for its allocated items exceeds its budget $B_j$, the definition of $x^t_{ij}$ effectively scales that contribution of each item proportionally.
	The right-hand-side of Constraint~\eqref{eqn:adwords-lp} by definition is the value we could get by allocating all items in $S$ to $j$.
	Hence, the inequality holds.
\end{proof}

\begin{lemma}[Computational Efficiency]
	\label{lem:adwords-lp-poly-time}
	The Stochastic AdWords LP is solvable in polynomial time if we could compute $\bar{v}_j(S)$ in polynomial time.
\end{lemma}

\begin{proof}
    For any realization of $R^t_i$, the function:
    \[
    	\min \bigg\{\: \sum_{(t,i) \in S} R^t_i \cdot b_{ij} \,,\, B_j \:\bigg\}
    \]
    is submodular.
    Hence, $\bar{v}_j$ is also submodular as it is a linear combination of submodular functions. 
    This means that the second set of constraints for any fixed $j \in J$ forms a polymatroid because the right-hand-side is a submodular set function over $T \times I$.

    Therefore, the Stochastic AdWords LP's polytope is the intersection of polynomially many linear constraints (the first set of constraints) and $|J|$ polymatroids (the second set of constraints).
    If we could compute $\bar{v}_j$ in polynomial time, then we had a polynomial-time separation oracle for the LP and could solve it in polynomial time, e.g., using the ellipsoid method.
\end{proof}

Unlike solving the Stochastic Matching LP, whose computational efficiency is unconditional (Lemma~\ref{lem:discrete-time-lp-polytime}), solving the Stochastic AdWords LP with the same approach requires an oracle for evaluating an expectation over a potentially exponential-size support.
We discuss below three possible ways to circumvent this obstacle.

First, we can approximate the value of $\bar{v}_j(S)$ up to an inverse-polynomially small additive error via the Monte Carlo method.

Second, we may interpret the $x^t_{ij}$ as allocation statistics of the optimal allocation of past data, and the LP constraints as properties that the statistics shall satisfy in expectation.
In other words, we do not solve the LP but treat it as a characterization of such allocation statistics.
See e.g., \citet{TangWW:STOC:2022} and \citet{AouadM:EC:2023} for some previous works that follow this approach.

Last but not least, we can relax Constraint~\eqref{eqn:adwords-lp} so that for any offline agent $j$, we only consider subsets $S$ of $(t, i)$ for large surrogate types $i$.
When we only have large bids, we can evaluate the expectation in polynomial-time, because the value $v_j$ would simply be $B_j$ whenever we have two or more large bids.
In other words, there are at most $|S| + 2$ possible realizations of $v_j$ instead of exponentially many.
We remark that in this case, we need to bring back Constraint~\eqref{eqn:adwords-fluid-lp} from the fluid LP to bound the contribution of small surrogate types.

\subsection{General SOCS from Two-Way SOCS}
\label{sec:adwords-general-socs}

This subsection considers the \textbf{General SOCS for AdWords} obtained by combining the \textbf{Two-way SOCS for AdWords} in Subsection~\ref{sec:two-way-socs-adwords} and the \textbf{Type Decomposition} in Subsection~\ref{sec:socs-reduction}.
We will next analyze its convergence rate.

\begin{theorem}
	\label{thm:adwords-general-socs}
    \textrm{\bf General SOCS for Adwords} has convergence rate:
    \[
    	g(y_j) ~=~ e^{-y_j} \cdot \left( \frac{3}{4} + \frac{1}{4} \cdot \frac{1+(y_j-0.417)^+/4}{e^{(y_j-0.417)^+/4}} \,\right)^2
    	~.
    \]
    That is, for any offline vertex $j \in J$, it is matched with probability at least:
    \begin{equation}
        \label{eqn:adwords-general-socs-convergence}	
        1 - g(y_j)
        ~.
    \end{equation}
\end{theorem}

Since Equation~\eqref{eqn:adwords-general-socs-convergence} (Appendix~\ref{app:adwords-general-socs-convergence}) is concave and equals zero when $y_j = 0$, we have:
\[
	1 - g(y_j) \ge \big( 1 - g(1) \big) \cdot y_j > 0.6338 \cdot y_j
	~.
\]

As a corollary, we get the following competitive ratio for Stochastic AdWords.

\begin{corollary}
	\label{cor:stochastic-adwords}
	Rounding the solution of \textbf{Stochastic AdWords LP} using \textbf{General SOCS for AdWords} is $0.634$-competitive for Stochastic AdWords.
\end{corollary}

The rest of the subsection will be devoted to proving Theorem~\ref{thm:adwords-general-socs}.
Consider any fixed offline agent $j$.
By the definition of \textbf{Type Decomposition}, we have:
\begin{align*}
    & y_{S}^t = \sum_{i\in S} x_{ij}^t b_{ij} ~, \\
    & y_{L1}^t = \sum_{i\in L} \left(2x_{ij}^t - f_i^t\right)^+ \cdot b_{ij} ~, \\
    & y_{L2}^t = \sum_{i\in L} \left(x_{ij}^t - \left(2x_{ij}^t - f_i^t\right)^+\right) \cdot b_{ij} ~.
\end{align*}

Further define the set of \emph{critical} time-type pairs as:
\[
	C = \Big\{\, (t, i) : i \in L, x_{ij}^t \ge \frac{1}{2} f_i^t \,\Big\}
	~.
\]

Note that for any time step $t$ there can be at most one online type $i \in L$ such that $(t, i) \in C$.
Moreover, for any critical time-item pair $(t, i) \in C$, let $q^t$ denote the expected consumption agent $j$'s budget if we always allocate an item of type $i$ at time step $t$.
That is:
\[
	q^t = f_i^t \cdot b_{ij}
	~.
\]

Then, we have:
\begin{equation}
	\label{eqn:adwords-L1-decompose}
    y_{L1}^t = \sum_{i : (t,i) \in C} \left(2x_{ij}^t b_{ij} - f_i^t b_{ij}\right) = \sum_{i : (t,i) \in C} \big( 2x_{ij}^t b_{ij} - q^t \big)
    ~. 
\end{equation}

We first consider the case when $q^t\le 0.4$ for all time steps $t \in [T]$.
By Constraint \eqref{eqn:adwords-lp} of the Stochastic AdWords LP:
\begin{align*}
	\sum_{(t, i) \in C} x_{ij}^t b_{ij}
	&
	~\le~ 
	1 - \prod_{(t, i) \in C} \big(1 - \mass^t_i \big) - \sum_{(t, i) \in C} \mass^t_i \prod_{(t', i') \in C \,:\, (t',i') \ne (t,i)} \big(1 - \mass^t_i \big) \cdot \big( 1 - b_{ij} \big) \\
	&
	~=~
	1 - \prod_{(t, i) \in C} \big(1 - \mass^t_i \big) - \sum_{(t, i) \in C} \big( \mass^t_i - q^t \big) \prod_{(t', i') \in C \,:\, (t',i') \ne (t,i)} \big(1 - \mass^t_i \big)
	~.
\end{align*}

This is increasing in $f^t_i$ for any $(t, i) \in C$.
By $f^t_i = \nicefrac{q^t}{b_{ij}} \le \nicefrac{3 q^t}{2}$, we get that:
\begin{align*}
    \sum_{(t, i) \in C} x_{ij}^t b_{ij}
    & 
    ~\le~ 
    1 - \prod_{(t,i) \in C} \left(1 - \frac{3}{2} q^t \right)  - \frac{1}{2} \sum_{(t,i) \in C} q^t \cdot \prod_{(t', i') \in C \,:\, (t',i') \ne (t,i)} \left(1-\frac32 q^{t'} \right) \\
    & 
    ~=~ 
    1 - \bigg(1 + \frac{1}{2} \sum_{(t,i) \in C} \frac{q^t}{1-\nicefrac{3q^t}{2}} \bigg) \prod_{(t,i) \in C} \left(1 - \frac{3}{2} q^t \right)
    ~.
\end{align*}

By $1 - x \ge e^{-x-x^2}$ for $0 \le x \le 0.6$, we have:
\[
    1 - \frac32 q^t \ge \exp \left(-\frac32 q_t - \frac94 q_t^2 \right)
    ~.
\]

Further by $\nicefrac{1}{(1-x)} \ge 1 + x$ we have:
\[
	\frac{q^t}{1-\nicefrac{3q^t}{2}} \ge q^t + \frac{3}{2} \big(q^t\big)^2
	~.
\]

Let $S = \sum_{(t,i) \in C} (q_t)^2$, we get that:
\[ 
    \sum_{(t, i) \in C} x_{ij}^t b_{ij} \le 1 -  \bigg(1 + \frac12 \sum_{t\in [T]} q_t + \frac34 S \bigg) \cdot \exp\bigg(-\frac32 \sum_{t\in [T]} q_t - \frac94 S \bigg).
\]

Combining this bound with Equation~\eqref{eqn:adwords-L1-decompose}, we derive the following upper bound on the contribution of large one-way types:
\begin{align*}
    y_{L1}
    &
    ~=~ \sum_{t\in [T]} y_{L1}^t \\
    &
    ~=~ 2 \sum_{(t,i) \in C} x^t_{ij} b_{ij} - \sum_{(t,i) \in C} q^t \\
    &
    ~\le~ 2 - 2\bigg(1 + \frac12 \sum_{(t,i) \in C} q_t + \frac34 S \bigg)\cdot \exp\bigg(-\frac32 \sum_{(t,i) \in C} q_t - \frac94 S \bigg) - \sum_{(t,i) \in C} q_t
    ~.
\end{align*}

If we were in the non-homogeneous Poisson arrival model, all $q^t$ would be infinitesimally small, and thus $S = 0$.
In general, the above bound minus $\frac32 S$ is at most (with $x = \frac{1}{2} \sum_{(t,i) \in C} q^t + \frac32 S$):
\[
	\max_{x \ge 0} ~ 2 - 2 \big(1 + x \big) e^{-3 x} - 2x < 0.358
	~.
\]

Therefore, we have:
\[
	y_{L1} \le 0.358 + \frac32 S
	~.	
\]

Let $Q = \sum_{t\in [T]} (y_{L}^t)^2$. 
Observe that $q_t \le \sum_{i\in L}2x_{ij}^t b_{ij} = 2y_L^t$.
Hence, we have $S\le 4Q$ and:
\begin{equation}
	\label{eqn:adwords-L1-bound}	
	y_{L1} \le 0.358 + 6Q
	~.	
\end{equation}

By Lemma~\ref{lem:two-way-socs-adwords}, the expected fraction of agent $j$'s unspent budget is at most:
\[
	e^{-y_j}
	\cdot
	\bigg( \frac{3}{4} + \frac{1}{4} \cdot \frac{1+\nicefrac{y_{L2}}{2}}{e^{\nicefrac{y_{L2}}{2}}}\bigg)
	\cdot
	\frac{1 + \nicefrac{y_S}{2}}{e^{\nicefrac{y_S}{2}}}
	\cdot
	\prod_{t=1}^T e^{y_{L1}^t} \big( 1 - y_{L1}^t \big)
	~.
\]

Dropping the third term as it is at most $1$, this is at most:
\[
    e^{-y_j} \cdot \left(\frac34 + \frac14 \cdot \frac{1+y_{L2}/2}{e^{y_{L2}/2}} \right) \cdot \frac{1+y_S/2}{e^{y_S/2}}
    ~.
\]

Symmetrizing the last two terms by relaxing the term related to $y_S$, we can bound it by:
\[
	e^{-y_j} \cdot \left(\frac34 + \frac14 \cdot \frac{1+y_{L2}/2}{e^{y_{L2}/2}}\right) \cdot \left(\frac34 + \frac14 \cdot \frac{1+y_S/2}{e^{y_S/2}}\right) 
	~.
\]

By the concavity of $\log \left(\frac34 + \frac14 \cdot \frac{1+y_{L2}/2}{e^{y_{L2}/2}}\right)$ (Appendix~\eqref{app:adwords-convergence-function}), this is at most:
\[
	e^{-y_j} \cdot \left(\frac34 + \frac14 \cdot \frac{1+(y_{L_2}+y_S)/4}{e^{(y_{L_2}+y_S)/4}}\right)^2
\]

Finally, by $y_j = y_S + y_{L2} + y_{L1}$ and $y_{L1} \le 0.358 + 6Q$ (Equation~\eqref{eqn:adwords-L1-bound}), the expected unspent budget is at most:
\begin{align*}
    e^{-y_j} \cdot \left(\frac34 + \frac14 \cdot \frac{1+(y_j-0.358-6Q)^+/4}{e^{(y-0.358-6Q)^+/4}}\right)^2.
\end{align*}

This would imply the stated bound if $Q \le 0.0098 < \nicefrac{(0.417 - 0.358)}{6}$.

On the other hand, we have the following simple upper bound for the expected unspent budget as a corollary of the first bound of Lemma~\ref{lem:two-way-socs-adwords}:
\[
	\prod_{t \in T} \big( 1 - y^t_L - y^t_S \big) \le \prod_{t \in T} \big( 1 - y^t_L\big) \big(1 - y^t_S \big) \le e^{-y_S} \cdot \prod_{t=1}^T \big( 1 - y^t_L \big)
	~.
\]

We can relax it as follows:
\begin{align*}
    e^{-y_j} \cdot \prod_t e^{y_L^t} (1-y_L^t)
    & \le e^{-y_j} \cdot \prod_t e^{-(y_L^t)^2/2} \\
    & = e^{-y_j} \cdot e^{-Q/2} \le e^{-y_j} \cdot e^{-\nicefrac{0.0098}{2}}
    ~.
\end{align*}

This is at most the stated bound in the theorem because:
\[
	\bigg( \frac34 + \frac14 \cdot \frac{1+(1-0.417)/4}{e^{(1-0.417)/4}} \bigg)^2 > 0.99518 > e^{-\nicefrac{0.0098}{2}}
	~.
\]

Finally, we consider the case when $q_t > 0.4$ for some critical time-type pair $(t, i) \in C$.
Then, we have $y_L^t > 0.2$. 
By the above bound, the expected fraction of the unspent budget is at most:
\[
    e^{-y_j} \cdot e^{y_L^t} \big(1 - y_L^t\big) \le e^{-y_j} \cdot e^{0.2} \cdot 0.8 < 0.97 \cdot e^{-y_j}
    ~,
\]
which is at most the stated bound in the theorem.

\subsection{(Multi-Way) OCS from Uninformed Two-Way SOCS}
\label{sec:adwords-ocs}

\paragraph{OCS for AdWords.}
Consider $T$ discrete time steps. 
At each step $t \in [T]$, consider an online item with online type $i \in I$ and a fractional allocation $\matchrate^t_i = \big( \matchrate^t_{ij} \big)_{j \in J}$ such that $\sum_{j \in J} \matchrate^t_{ij} = 1$.
The type and fractional allocation are chosen by an adversary at the beginning, and revealed to the algorithm at time step $t$.
After observing the online item's type and the fractional allocation at each time step $t$, the OCS must allocate the item to an offline agent immediately.
Since each time step $t \in [T]$ is associated with only one online type, we will abuse notation and refer to this type as $t$.
Hence, the notation for the fractional allocation simplifies to $\matchrate^t = \big( \matchrate^t_j \big)_{j \in J}$, and the bids are denoted as $b^t_j$'s.

Similar to the definition of SOCS in Section~\ref{sec:socs}, we consider each offline agent $j$'s fraction of budget used by the fractional allocation, and denote it as:
\[
	y_j = \sum_{t=1}^T \matchrate^t_j \frac{b^t_j}{B_j}
	~.
\]

We measure the quality of an OCS for AdWords also by its convergence rate $g(y_j)$.
For any offline agent $j \in J$, the OCS needs to guarantee that the expected unused fraction of agent $j$'s budget is at most $g(y_j)$.
Equivalently, offline agent $j$'s expected contribution to the objective of the AdWords problem is at least:
\[
	\big( 1 - g(y_j) \big) B_j
	~.
\]

The next lemma explains the relation between the convergence rate of OCS and the resulting competitive ratio for the AdWords problem in the adversarial model.
Its proof is almost identical to a corresponding result for unweighted and vertex-weighted matching proved by \citet{GaoHHNYZ:FOCS:2021}.
Hence, we defer it to Appendix~\ref{app:adwords-competitive}.

\begin{lemma}
	\label{lem:adwords-competitive}
	Suppose that we have an OCS for AdWords with convergence rate $g(y_j)$.
	Then, rounding a Balance algorithm with this OCS achieves competitive ratio:
	\[
		1 - \int_0^\infty g(z) e^{-z} \,\dif z
		~.
	\]	
\end{lemma}

Recall that the baseline algorithm that independently samples an offline agent $j$ according to $\matchrate^t$ yields convergence rate $g(y_j) = e^{-y_j}$.
This convergence rate leads to a competitive ratio of $\frac{1}{2}$ according to Lemma~\ref{lem:adwords-competitive}, which is also the baseline competitive ratio given by the greedy algorithm.
Hence, our goal is to design an OCS for AdWords with a faster convergence rate, and to obtain an improved competitive ratio for the AdWords problem.

\paragraph{OCS for AdWords from Uninformed Two-Way SOCS.}
Given the fractional allocation $\mu^t$ of each time step $t \in [T]$, we can run the \textbf{Type Decomposition} algorithm from Section~\ref{sec:socs-reduction} to get a distribution over one-way and two-way surrogate types. 
By doing so, we get an instance of SOCS with one-way and two-way types. 
Can we then select using a two-way SOCS?
The caveat is that the two-way SOCS algorithm must be \emph{uninformed} in the sense that at any time step $t \in [T]$ it only knows the distributions of time steps $1$ to $t$ but not those in the future.
It is easy to see that our \textbf{Two-Way SOCS for AdWords} is uninformed. 
We will next analyze the convergence rate of the resulting OCS.

\begin{theorem}
	\label{thm:ocs-adwords}
	Combining \textbf{Type Decomposition} and \textbf{Two-Way SOCS for AdWords} gives an OCS for the AdWords problem with convergence rate:
	\[
		g(y_j) = e^{-y_j} 
		\cdot 
		\max_{\substack{\vphantom{|} y_S,\, y_{L1},\, y_{L2}  \, \ge 0 \,: \\ \vphantom{|} y_S + y_{L1} + y_{L2} = y_j}} 
		\frac{1+\nicefrac{y_S}{2}}{e^{\nicefrac{y_S}{2}}}
		\cdot
		\bigg( \frac{3}{4} + \frac{1}{4} \frac{1+\nicefrac{y_{L2}}{2}}{e^{\nicefrac{y_{L2}}{2}}} \bigg) 
		\cdot 
		\exp \Big( - \frac{y_{L1}^2}{3 y_{L1} + 6 y_{L2}} \Big)
		~.
	\]
\end{theorem}

This is strictly better than the baseline convergence rate $e^{-y_j}$ for any $y_j > 0$.
To see this, if $y_S \ge \frac{y_j}{3}$, the first term in the above maximization is strictly smaller than $1$.
If $y_{L2} \ge \frac{y_j}{3}$, the second term is the above maximization is strictly smaller than $1$.
Finally, if $y_{L1} \ge \frac{y_j}{3}$ but $y_{L2} < \frac{y_j}{3}$, the third term is the above maximization is strictly smaller than $1$.

To compute the convergence rate for each $y_j$, we solve the maximization problem numerically using a standard solver.
By solving the value of $g(y_j)$ and then the integration in Lemma~\ref{lem:adwords-competitive} numerically, we get the following corollary.

\begin{corollary}
	\label{cor:adwords}
	Rounding the \textbf{Balance-OCS for AdWords} using the \textbf{OCS for AdWords} is $0.504$-competitive for the AdWords problem in the adversarial model.
\end{corollary}

\begin{proof}[Proof of Theorem~\ref{thm:ocs-adwords}]
Fix any offline agent $j \in J$.
Let $\mass^t_j$, $\mass^t_{\{j, *\}}$ denote the probability of realizing one-way and two-way surrogate types involving offline agent $j$.
Recall that $S$ and $L$ denote the subsets of time steps (which uniquely determine the types) that receive small and large bids respectively from agent $j$.

We now adopt the notations of $y_S$, $y_{L1}$, and $y_{L2}$ into the OCS setting as follows, normalizing agent $j$'s budget to be $B_j = 1$ for ease of notations:
\begin{align*}
	y_S^t & = \mathbf{1}_{t \in S} \cdot \matchrate^t_j \cdot b^t_j 
    ~, \\[2ex]
	y_{L1}^t & = \mathbf{1}_{t \in L} \cdot \underbrace{(2\matchrate^t_j - 1)^+}_{\text{$\mass^t_j$, prob.\ of one-way type $j$}} \cdot~ b^t_j
    ~, \\
	y_{L2}^t & = \mathbf{1}_{t \in L} \cdot \underbrace{\min \big\{ \matchrate^t_j, 1 - \matchrate^t_j \big\}}_{\text{$\mass^t_{\{j,*\}}$, prob.\ of two-way type $\{j, *\}$}} \cdot~ \frac{1}{2} \cdot b^t_j
    ~.
\end{align*}

Further, let:
\[
	y_S = \sum_{t=1}^T y_S^t
	\quad,\qquad
	y_{L2} = \sum_{t=1}^T y_{L2}^t
    ~.
\]

The expected unused budget is bounded by (Theorem~\ref{thm:two-way-socs-adwords}):
\[
	e^{-y_j}
	\cdot
	\frac{1 + \nicefrac{y_S}{2}}{e^{\nicefrac{y_S}{2}}}
	\cdot
	\bigg( \frac{3}{4} + \frac{1}{4} \cdot \frac{1 + \nicefrac{y_{L2}}{2}}{e^{\nicefrac{y_{L2}}{2}}} \bigg)
	\cdot 
		\underbrace{\prod_{t=1}^T e^{y_{L1}^t} \big( 1- y_{L1}^t \big)}_{\text{$(\star)$}} 
    ~.
\]

For any time step whose $y_{L1}^t > 0$, the contributions from large one-way and two-way surrogate types satisfy:
\begin{equation}
    \label{eqn:adwords-OCS-L1-L2}
	2 y_{L2}^t + y_{L1}^t = b^t_j \ge \frac{2}{3}
	~.
\end{equation}
because we have:
\[
	y_{L2}^t = \big(1 - \matchrate^t_{j} \big) b_{ij}
	\quad,\qquad
	y_{L1}^t = \big( 2 \matchrate^t_{j} - 1 \big) b_{ij}
	~.
\]

Let $\Delta$ denote the average value of $y^t_{L1}$ over the time steps with a positive $y^t_{L1}$.
The number of such time steps is then $\nicefrac{y_{L1}}{\Delta}$.
By the concavity of $\log \big( e^y (1-y) \big) = y + \ln(1-y)$ for $0 \le y < 1$, the contribution from $(\star)$  above is bounded by:
\[
	(\star) \le \big( e^\Delta ( 1 - \Delta ) \big)^{\nicefrac{y_{L1}}{\Delta}} = \exp \Big( \frac{y_{L1} \big( \Delta + \ln (1-\Delta) \big)}{\Delta} \Big) \le \exp \Big( - \frac{\Delta}{2} y_{L1} \Big)
	~,
\]
where the inequality follows by $\ln(1-x) \le - x - \nicefrac{x^2}{2}$.

Further, by summing Equation~\eqref{eqn:adwords-OCS-L1-L2} over time steps with positive $y^t_{L1}$, we have:
\[
	2y_{L2} + y_{L1} ~\ge~ \sum_{t \,:\, y_{L1}^t > 0} \big( 2y_{L2}^t + y_{L1}^t \big) ~\ge~ \frac{y_{L1}}{\Delta} \cdot \frac{2}{3} 
	~.
\]

Rearranging terms, we get that:
\[
	\Delta \ge \frac{2y_{L1}}{3y_{L1} + 6y_{L 2}}
	~.
\]

Combining the above two inequalities proves the theorem.
\end{proof}

\section{Display Ads}
\label{sec:display-ads}

This section presents our results for Stochastic Display Ads.
Recall that we will account for each offline agent $j$'s contribution to the objective by each weight-level.
For any weight-level $w > 0$, we consider the total fractional allocation given to offline agent $j$ from online items with edge-weights at least $w$, denoted as $y_j(w)$.
A SOCS with convergence rate $g(\cdot)$ needs to ensure that it allocates an online item with edge-weight at least $w$ to offline agent $j$ with probability at least $1 - g\big(y_j(w)\big)$.

\subsection{Two-Way SOCS for Display Ads}

Consider the following algorithm similar to the two-way SOCS for AdWords.

\bigskip

\begin{tcolorbox}
    \textbf{Two-Way SOCS for Display Ads}\\[2ex]
    When an online vertex with two-way surrogate type $i \sim \{j, k\}$ arrives at time step $t$:
    \begin{enumerate}
        \item Pick $m \in \{j, k\}$ uniformly at random and mark this online vertex with $m$.
        \item If this is the second online vertex marked with $m$, then make the opposite selection to the first one (w.r.t.\ $m$).
        \item Otherwise, select $j$ or $k$ uniformly at random.
    \end{enumerate}
\end{tcolorbox}

\medskip

\begin{theorem}
	\label{thm:display-ads-two-way}
	\textbf{Two-Way SOCS for Display Ads} has convergence rate:
	\[
		g(y) =
		\min \Big\{ \Big(1+\frac{1}{2}y\Big) e^{-\frac{3}{2}y} + \frac{1-y}{15}, e^{-y} \Big\}
		~.
	\]
\end{theorem}

\medskip

We remark that as $y$ increases and even approaches $1$, which will be the most important regime for our competitive analysis, the above bound behaves like $e^{-y} \cdot \big(1+\frac{1}{2}y\big) e^{-\frac{1}{2}y}$ and is strictly better than the baseline $e^{-y}$.

We first present a basic property of this two-way algorithm about the probability of selecting an agent $j$ within a subset of time steps.

\begin{lemma}
	\label{lem:display-ads-two-way-basic}
	For any subset of time steps $S$ that realize two-way types involving agent $j$, the probability that agent $j$ is never selected in these time steps is at most:
	\[
		2^{-|S|}
		~.
	\]
	If $S$ is the subset of \emph{all} time steps with two-way types involving agent $j$, the bound improves to:
	\[
		\big(1+|S|\big) 2^{-2|S|}
		~.
	\]
\end{lemma}

\begin{proof}
	If any pair of these time steps are the first two steps marked with some offline vertex $m$ (which may or may not be $j$), then $j$ would be selected in one of the two steps because the algorithm would make opposite selections by definition.
	If $S$ is further the set of \emph{all} time steps involving $j$, then we can bound this probability by just considering the possibility of having at least two of these steps marked with $j$, which happens with probability at least $1 - (|S|+1) 2^{-|S|}$.
	Otherwise, we just use the trivial bound $0$.
	
	Next, suppose the above event does not happen.
	By definition, the algorithm selects independently and uniformly at random in these $|S|$ time steps.
	Hence, the probability of never selecting agent $j$ is at most $2^{-|S|}$.
	
	Combining the two parts proves the lemma.
\end{proof}

\paragraph{Additional Notations.}
The rest of the subsection will fix an offline agent $j$ and a weight-level $w > 0$.
Let $I_1$ and $I_2$ denote the sets of one-way types $i \sim j$ and two-way types $i \sim \{j, k\}$ involving agent $j$ with edge-weights $w_{ij} \ge w$.
Let $I_-$ denote the two-way types $i \sim \{j, k\}$ involving agent $j$ with edge-weights $w_{ij} < w$.
We remark that the one-way types with edge-weights less than $w$ do not affect our analysis.

Correspondingly, let $p^t_1$, $p^t_2$, and $p^t_-$ be the probabilities of realizing a type in $I_1$, $I_2$, and $I_-$ respectively at time step $t$.
Denote the expected fractional allocations to $j$ from these three kinds of online types as:
\[
	y_1 = \sum_{t \in [T]} p^t_1
	\quad,\quad
	y_2 = \frac{1}{2} \sum_{t \in [T]} p^t_2
	\quad,\quad
	y_- = \frac{1}{2} \sum_{t \in [T]} p^t_-
	~.
\]

\begin{lemma}
	\label{lem:display-ads-two-way-general}
	The probability that agent $j$ does not get any online item with edge-weight at least $w$ is upper bounded by both:
	\[
		\prod_{t \in T} \Big(1 - p^t_1 - \frac{1}{2} p^t_2 \Big) ~\le~ e^{-y_1-y_2}
		~,
	\]
	and:
	\[
		e^{-y_1-\frac{3}{2}y_2} \Big(1 + \frac{1}{2}y_2\Big) + \frac{1}{15} \big( 1-y_1-y_2 \big)
		~.
	\]
\end{lemma}

\medskip

Note that Theorem~\ref{thm:display-ads-two-way} follows as a corollary of Lemma~\ref{lem:display-ads-two-way-general} when $y_1 = 0$ and with $y = y_2$.
The rest of the subsection will be devoted to proving this main lemma.

\begin{lemma}
	\label{lem:display-ads-two-way-closed-form-basic}
	The probability that the maximum edge-weight allocated to offline agent $j$ is strictly smaller than $w$ is at most:
	\begin{equation}
		\label{eqn:display-ads-two-way-closed-form-basic}
		\begin{aligned}
			&
			\prod_{t \in [T]} \Big( 1 - p^t_1 - \frac{1}{2} p^t_2 \Big)
			- 
			\prod_{t \in [T]} \Big( 1 - p^t_1 - \frac{1}{2} p^t_2 - p^t_-\Big) \\
			& \qquad
			+ 
			\prod_{t \in [T]} \Big( 1 - p^t_1 - \frac{3}{4} p^t_2 - p^t_-\Big)
			+
			\frac{1}{4} \sum_{t \in [T]} p^t_2 \prod_{t' \ne t} \Big( 1 - p^{t'}_1 - \frac{3}{4} p^{t'}_2 - p^{t'}_-\Big)
			~.
		\end{aligned}
	\end{equation}
	Further, we remark that the last three terms combined are non-positive.
\end{lemma}

\begin{proof}
	Note that if we have an online type from $I_1$ in any time step, then agent $j$ will get the item whose edge-weight is at least $w$ by definition.
	Hence, we will focus on the realization of online types where none is from $T_1$.
	
	Consider any disjoint subsets $T_2, T_-$ of the time steps.
	The probability of realizing online types from $I_2$ in time steps $T_2$, types from $I_-$ in time steps $I_-$, and types not in $I_1, I_2, I_-$ in the other time steps, is equal to:
	\[
		\prod_{t \in T_2} p^t_2 \prod_{t \in T_-} p^t_- \prod_{t \notin T_2 \cup T_-} (1 - p^t_1 - p^t_2 - p^t_-)
		~.
	\]
	
	By Lemma~\ref{lem:display-ads-two-way-basic}, for any realization of $T_2$ and $T_-$, the probability of not allocating an item with edge-weight at least $w$ to agent $j$ is at most $2^{-|T_2|}$.	
	Further, if $T_- = \varnothing$, the probability of not allocating such an item to agent $j$ further decreases to $(1+|T_2|) 2^{-2|T_2|}$.
	The latter event happens with probability:
	\[
		\prod_{t \in T_2} p^t_2 \prod_{t \notin T_2} (1 - p^t_1 - p^t_2 - p^t_-)
		~.
	\]

	Putting together, the probability of concern is at most:
	\begin{align*}
		&
		\underbrace{\sum_{T_2, T_-} \prod_{t \in T_2} p^t_2 \prod_{t \in T_-} p^t_- \prod_{t \notin T_2 \cup T_-} (1 - p^t_1 - p^t_2 - p^t_-) \cdot 2^{-|T_2|}}_{(a)} ~-~ \underbrace{\sum_{T_2} \prod_{t \in T_2} p^t_2 \prod_{t \notin T_2} (1 - p^t_1 - p^t_2 - p^t_-) \cdot 2^{-|T_2|}}_{(b)} \\
		& \qquad
		+~ \underbrace{\sum_{T_2} \prod_{t \in T_2} p^t_2 \prod_{t \notin T_2} (1 - p^t_1 - p^t_2 - p^t_-) \cdot (1+|T_2|) \,2^{-2|T_2|}}_{(c)}
		~.
	\end{align*}
	
	We will next simplify the three terms.
	First, consider $(a)$:
	\begin{align*}
		(a)
		&
		= \sum_{T_2, T_-} \prod_{t \in T_2} \frac{1}{2} p^t_2 \prod_{t \in T_-} p^t_- \prod_{t \notin T_2 \cup T_-} (1 - p^t_1 - p^t_2-p^t_-) \\
		& 
		= \prod_{t \in [T]} \Big( \frac{1}{2} p^t_2 + p^t_- + 1 - p^t_1 - p^t_2-p^t_- \Big) \\
		&
		= \prod_{t \in [T]} \Big( 1- p^t_1 - \frac{1}{2} p^t_2 \Big) 
		~.
	\end{align*}
	
	The second part can be simplified similarly as:
	\begin{align*}	
		(b)
		&
	 	= \sum_{T_2} \prod_{t \in T_2} \frac{1}{2} p^t_2 \prod_{t \notin T_2} (1 - p^t_1 - p^t_2 - p^t_-) \\
	 	&
	 	= \prod_{t \in [T]} \Big( 1 - p^t_1- \frac{1}{2} p^t_2 - p^t_- \Big) 
	 	~.
	\end{align*}
	
	Next, we consider part $(c)$ by introducing a generating function:
	\[
		h(x) = \prod_{t \in [T]} \big( 1 - p^t_1 - p^t_2 - p^t_- + p^t_2 \cdot x \big)
		~.
	\]
	
	We have:
	\begin{align*}
		(c)
		&
		= h\Big(\frac{1}{4}\Big) + \frac{1}{4} h'\Big(\frac{1}{4}\Big) \\
		&
		= \prod_{t \in [T]} \Big( 1 - p^t_1 - \frac{3}{4} p^t_2 - p^t_-\Big)
		+
		\frac{1}{4} \sum_{t \in [T]} p^t_2 \prod_{t' \ne t} \Big( 1 - p^{t'}_1 - \frac{3}{4} p^{t'}_2 - p^{t'}_-\Big)
		~.
	\end{align*}

	Finally, the remark holds because the last three terms in the lemma come from parts $(b)$ and $(c)$, and we have $(b) \ge (c)$ by $(1+|T_2|) 2^{-|T_2|} \le 1$.	
\end{proof}

\begin{lemma}
	\label{lem:display-ads-merge-minus}
	Subject to $\sum_t p^t_- = 2 y_-$, Equation~\eqref{eqn:display-ads-two-way-closed-form-basic} is maximized when $p^t_-$ equals $2 y_-$ in one time step and $0$ in the other time steps.
\end{lemma}

\begin{proof}
	Suppose that there are two different time steps $t_1 \ne t_2$ such that $p^{t_1}_-$ and $p^{t_2}_-$ are both positive.
	We will next modify them such that one becomes $p^{t_1}_- + p^{t_2}_-$ and the other becomes $0$ and prove that Equation~\eqref{eqn:display-ads-two-way-closed-form-basic} weakly increases.
	Note that Equation~\eqref{eqn:display-ads-two-way-closed-form-basic} is multi-linear in $p^{t_1}_-$ and $p^{t_2}_-$.
	Hence, we can write it as:
	\[
		A + B \cdot p^{t_1}_- + C \cdot p^{t_2}_- + D \cdot p^{t_1}_- \, p^{t_2}_-
		~.
	\]
	
	The coefficient $D$ is equal to:
	\begin{align*}
		&
		- 
		\prod_{t \ne t_1, t_2} \Big( 1 - p^t_1 - \frac{1}{2} p^t_2 - p^t_-\Big) \\
		& \qquad
		+ 
		\prod_{t \ne t_1, t_2} \Big( 1 - p^t_1 - \frac{3}{4} p^t_2 - p^t_-\Big)
		+
		\frac{1}{4} \sum_{t \ne t_1, t_2} p^t_2 \prod_{t' \ne t, t_1, t_2} \Big( 1 - p^{t'}_1 - \frac{3}{4} p^{t'}_2 - p^{t'}_-\Big)
		~.
	\end{align*}
	
	Note that it has the same form as the last three terms of Equation~\eqref{eqn:display-ads-two-way-closed-form-basic}, and corresponds to the improvement upon the baseline probability with time steps $t_1, t_2$ removed.
	Hence, we have $D \le 0$ by the remark in Lemma~\ref{lem:display-ads-two-way-closed-form-basic}.

	By symmetry, we may assume without loss of generality that $B \ge C$.
	Then, changing $p^{t_1}_-$ to $p^{t_1}_- + p^{t_2}_-$ and $p^{t_2}_-$ to $0$ weakly increases this multi-linear function.

	The lemma follows by repeating this process until $p^t_-$ is nonzero for only one time step $t$.
\end{proof}

\begin{lemma}
\label{lem:display-ads-split}
	For any time step $\hat{t}$ with $p^{\hat{t}}_- = 0$, Equation~\eqref{eqn:display-ads-two-way-closed-form-basic} would weakly increase if we split it into two time steps $t_1, t_2$ such that $p^{t_1}_1 = p^{t_2}_1 = \frac{1}{2} p^{\hat{t}}_1$ and $p^{t_1}_2 = p^{t_2}_2 = \frac{1}{2} p^{\hat{t}}_2$.
\end{lemma}

\begin{proof}
	The change is:
	\begin{align*}
	    &
	    \frac{1}{4} \Big(p^{\hat{t}}_1 + \frac{1}{2} p^{\hat{t}}_2 \Big)^2 \prod_{t \neq \hat{t}} \Big( 1 - p^t_1 - \frac{1}{2} p^t_2 \Big)
	    - 
	    \frac{1}{4} \Big(p^{\hat{t}}_1 + \frac{1}{2} p^{\hat{t}}_2 \Big)^2 \prod_{t \neq \hat{t}} \Big( 1 - p^t_1 - \frac{1}{2} p^t_2 - p^t_-\Big)
	    \\
	    & \qquad
	    + 
	    \frac{1}{4} \Big(p^{\hat{t}}_1 + \frac{3}{4} p^{\hat{t}}_2 \Big)^2 \prod_{t \neq \hat{t}} \Big( 1 - p^t_1 - \frac{3}{4} p^t_2 - p^t_-\Big)
	    -
	    \frac{1}{8} p^{\hat{t}}_2 \Big(p^{\hat{t}}_1 + \frac{3}{4} p^{\hat{t}}_2 \Big) \prod_{t \neq \hat{t}} \Big( 1 - p^t_1 - \frac{3}{4} p^t_2 - p^t_-\Big) \\
	    & \qquad
	    + \frac{1}{16} \Big(p^{\hat{t}}_1 + \frac{3}{4} p^{\hat{t}}_2 \Big)^2 \sum_{t \neq \hat{t}} p^t_2 \prod_{t' \ne t, \hat{t}} \Big( 1 - p^{t'}_1 - \frac{3}{4} p^{t'}_2 - p^{t'}_-\Big)
		~.
	\end{align*}
	
	This is non-negative because all three lines are non-negative.
\end{proof}

\begin{proof}[Proof of Lemma~\ref{lem:display-ads-two-way-general}]
	Recall that we want to upper bound the probability that agent $j$ does not get any online item with edge-weight at least $w$ by:
	\[
		\prod_{t \in T} \Big(1 - p^t_1 - \frac{1}{2} p^t_2 \Big) ~\le~ e^{-y_1-y_2}
	\]
	and:
	\[
		e^{-y_1-y_2} \Big(1 + \frac{1}{4}y_2\Big)^2 e^{-\frac{1}{2}y_2} + \frac{1}{15} (1-y_1-y_2)
		~.
	\]
	
	The first bound follows by Lemma~\ref{lem:display-ads-two-way-closed-form-basic} and its remark.
	
	It remains to prove the second bound.
	We will next make a sequence of modifications to the instance using Lemmas~\ref{lem:display-ads-merge-minus} and \ref{lem:display-ads-split}.
	In this process, Equation~\eqref{eqn:display-ads-two-way-closed-form-basic} will weakly increase but the instance will become more structured. 
	
	First, we apply Lemma~\ref{lem:display-ads-merge-minus} to modify the instance such that there is a time step $t^*$ for which $p^{t^*}_- = 2 y_-$ and we have $p^t_- = 0$ in the other time steps $t \ne t^*$.
	As a result, Equation~\eqref{eqn:display-ads-two-way-closed-form-basic} weakly increases and becomes:
	\begin{align*}
		&
		\prod_{t \in [T]} \Big( 1 - p^t_1 - \frac{1}{2} p^t_2 \Big)
		~-~ 
		\Big(1 - p^{t^*}_1 - \frac{1}{2} p^{t^*}_2 - 2y_- \Big) \prod_{t \ne t^*} \Big( 1 - p^t_1 - \frac{1}{2} p^t_2\Big) \\
		& \qquad
		+~ 
		\Big( 1 - p^{t^*}_1 - \frac{1}{2} p^{t^*}_2 - 2y_- \Big) \prod_{t \ne t^*} \Big( 1 - p^t_1 - \frac{3}{4} p^t_2 \Big) \\
		& \qquad
		+~
		\frac{1}{4} \Big( 1 - p^{t^*}_1 - \frac{3}{4} p^{t^*}_2 - 2y_- \Big) \sum_{t \ne t^*} p^t_2 \prod_{t' \ne t, t^*} \Big( 1 - p^{t'}_1 - \frac{3}{4} p^{t'}_2 \Big)
		~.	
	\end{align*}

	Next, by Lemma~\ref{lem:display-ads-split}, we may now consider without loss of generality the case when $p^t_1$ and $p^t_2$ are infinitesimally small in any time step $t \ne t^*$.
	The equation then simplifies to:
	\begin{align*}
		&
		\Big( 1 - p^{t^*}_1 - \frac{1}{2} p^{t^*}_2 \Big)
		\exp\Big( - \sum_{t \ne t^*} \Big( p^t_1 + \frac{1}{2} p^t_2 \Big) \Big)
		- 
		\Big(1 - p^{t^*}_1 - \frac{1}{2} p^{t^*}_2 - 2y_- \Big) \exp\Big( - \sum_{t \ne t^*} \Big( p^t_1 + \frac{1}{2} p^t_2 \Big) \Big) \\
		& \qquad
		+ 
		\Big( 1 - p^{t^*}_1 - \frac{1}{2} p^{t^*}_2 - 2y_- \Big) \exp\Big( - \sum_{t \ne t^*} \Big( p^t_1 + \frac{3}{4} p^t_2 \Big) \Big) \\
		& \qquad
		+
		\frac{1}{4} \Big( 1 - p^{t^*}_1 - \frac{3}{4} p^{t^*}_2 - 2y_- \Big) \sum_{t \ne t^*} p^t_2 \cdot \exp\Big( - \sum_{t \ne t^*} \Big( p^t_1 + \frac{3}{4} p^t_2 \Big) \Big)
		~.	
	\end{align*}
	
	To simplify the notations, we introduce $z_1$ and $z_2$ to denote:
	\[
		z_1 = \sum_{t \ne t^*} p^t_1
		\quad,\quad
		z_2 = \frac{1}{2} \sum_{t \ne t^*} p^t_2
		~.
	\]
	
	Merging the first two terms and applying the above notations to the equation, it can then be rewritten as follows:
	\begin{align*}
		&
		2y_- \cdot e^{-z_1-z_2}
		+ 
		\Big( 1 - p^{t^*}_1 - \frac{1}{2} p^{t^*}_2 - 2y_- \Big) \cdot e^{-z_1-\frac{3}{2}z_2} 
		+
		\frac{1}{2} \Big( 1 - p^{t^*}_1 - \frac{3}{4} p^{t^*}_2 - 2y_- \Big) z_2 \cdot e^{-z_1-\frac{3}{2}z_2}
		~.	
	\end{align*}
	
	The terms unrelated to $y_-$ sum to:
	\[
		e^{-z_1-\frac{3}{2} z_2} \bigg( 
		\Big( 1 - p^{t^*}_1 - \frac{1}{2} p^{t^*}_2 \Big) +
		\frac{1}{2} \Big( 1 - p^{t^*}_1 - \frac{3}{4} p^{t^*}_2 \Big) z_2 \bigg)
		~.
	\]
	
	Since $1-a-b \le (1-a)(1-b)$ for any $a, b \ge 0$, this is upper bounded by:
	\[
		e^{-z_1-\frac{3}{2} z_2} \big( 1 - p^{t^*}_1 \big) \bigg( 
		\Big( 1 - \frac{1}{2} p^{t^*}_2 \Big) +
		\frac{1}{2} \Big( 1 - \frac{3}{4} p^{t^*}_2 \Big) z_2 \bigg)
		~.
	\]
	
	Since $z_1 + p^{t^*}_1 = y_1$ and $1-p^{t^*}_1 \le e^{-p^{t^*}_1}$, we further relax it to:
	\[
		e^{-y_1} \cdot \underbrace{e^{-\frac{3}{2} z_2} \bigg( 
		\Big( 1 - \frac{1}{2} p^{t^*}_2 \Big) +
		\frac{1}{2} \Big( 1 - \frac{3}{4} p^{t^*}_2 \Big) z_2 \bigg)}_{(*)}
		~.
	\]
	
	Note that $z_2 + \frac{1}{2} p^{t^*}_2 = y_2$.
	We will next prove that $(*)$ is maximized when $z_2 = y_2$ and $p^{t^*}_2 = 0$.
	For notational simplicity, let $x = \frac{1}{2} p^{t^*}_2$. 
	We just need to show that:
	\[
		1 - x +
		\frac{1}{2} \Big( 1 - \frac{3}{2} x \Big) z_2
		\le 
		e^{-\frac{3}{2} x} \bigg( 1 +
		\frac{1}{2} (x+z_2) \bigg)
		~.
	\]
	
	The coefficients of $z_2$ satisfy:
	\[
		\frac{1}{2} \Big(1 - \frac{3}{2} x\Big) \le \frac{1}{2} e^{-\frac{3}{2} x}
		~.
	\]
	
	The constant terms satisfy (Appendix~\ref{app:e-2x}):
	\begin{equation}
		\label{eqn:e-abx}
		1 - x \le e^{-\frac{3}{2}x} \Big(1 + \frac{1}{2} x \Big)
	\end{equation}

	Therefore, we can upper bound the sum of terms unrelated to $2y_-$ by:
	\[
		e^{-y_1-\frac{3}{2} y_2} \Big( 1 + \frac{1}{2} y_2 \Big)
		~.
	\]
 
	The coefficients of $2y_-$ sum to:
	\[
		e^{-z_1 -z_2} \Big( 1 - e^{-\frac{1}{2} z_2} \Big(1 + \frac{1}{2}z_2\Big) \Big)
		~.
	\]
	
	Since $e^{-z} (1+z)$ is decreasing in $z \ge 0$, this is at most:
	\begin{equation}
		\label{eqn:display-ads-y-minus-coef}
		e^{-z_1 -z_2} \Big( 1 - e^{-\frac{1}{2} (z_1 + z_2)} \Big(1 + \frac{1}{2}(z_1 + z_2)\Big) \Big)
		~.
	\end{equation}
	
	Further, this is non-decreasing in $0 \le z_1 + z_2 \le 1$ (Appendix~\ref{app:display-ads-y-minus-coef}), the sum of coefficients of $2y_-$ is at most:
	\[
		\frac{1}{e} - \frac{3}{2 e^{\frac{3}{2}}} < \frac{1}{30}
		~.
	\]
	
	Finally, the bound follows by $y_- \le 1 -  y_1 - y_2$.
\end{proof}

\subsection{General SOCS from Two-Way SOCS}

 We will next consider a general \textbf{SOCS for Display Ads} obtained by combining the \textbf{Type Decomposition} from Section~\ref{sec:socs-reduction} and the \textbf{Two-Way SOCS for Display Ads} and its convergence rate from the last subsection.

\begin{theorem}
\label{thm:general-socs-display-ads}
    \textrm{\bf SOCS for Display Ads} has convergence rate:
    \[
    	g(y)
        ~=~ 
        \min \Big\{
        	e^{-y}
            \,,\,
            e^{-y} \Big(1 + \frac12 (y-0.44)^+ \Big) e^{- \frac12 (y-0.44)^+} + \frac{1-y}{15}
        \Big\}
        ~.
    \]
    That is, for any offline agent $j \in J$ and any weight-level $w > 0$, the maximum edge-weight allocated to $j$ is at least $w$ with probability at least $1 - g\big(y_j(w)\big)$.
\end{theorem}

\medskip
We observe that for any $0 \le y_j(w) \le 1$ (see Appendix~\ref{app:display-ads-ratio}):
\begin{equation}
	\label{eqn:display-ads-ratio}
	1 - g\big(y_j(w)\big) \ge 0.644 \cdot y_j(w)
	~.
\end{equation}

Integrating over all weight-levels $w > 0$ and summing over all offline agents $j \in J$, we get the following competitive ratio for the Stochastic Display Ads problem as a corollary.

\begin{corollary}
	\label{cor:display-ads}
	Rounding the solution of \textbf{Stochastic Matching LP} with \textbf{SOCS for Display Ads} is $0.644$-competitive for Stochastic Display Ads.
\end{corollary}
\begin{proof}[Proof of Theorem~\ref{thm:general-socs-display-ads}]
	Fix any agent $j$ and any weight-level $w > 0$.
	Recall our notations from the previous subsection.
	Let $y_1$ and $y_2$ denote the total fractional allocation to agent $j$ from one-way and two-way surrogate types respectively whose edge-weights are at least $w_{ij} \ge w$.
	We have:
	\[
		y_j(w) = y_1 + y_2
		~.
	\]
	
	By the first bound of Lemma~\ref{lem:display-ads-two-way-general}, we have:
	\[
		g \big( y_j(w) \big) \le e^{-y_1-y_2} = e^{- y_j(w)}
		~,
	\]
	matching the first upper bound in the theorem.
	
	By the second bound of Lemma~\ref{lem:display-ads-two-way-general}, we have:
	\begin{align*}
		g \big( y_j(w) \big) 
		&
		\le e^{-y_1-\frac{3}{2} y_2} \Big(1 + \frac{1}{2} y_2 \Big) + \frac{1}{15} \big( 1 - y_1 - y_2 \big) \\
		&
		= e^{-y_j(w)} e^{-\frac{1}{2}y_2} \Big(1 + \frac{1}{2} y_2 \Big) + \frac{1}{15} \big( 1 - y_j(w) \big)
		~.
	\end{align*}
	
	Noting that $e^{-x}(1+x)$ is decreasing in $x \ge 0$, and comparing the above with the second upper bound in the theorem, ideally we would like to show that:
	\[
		y_2 ~\ge~ y_j(w) - 0.44
		~.
	\]
	
	Since $y_1 + y_2 = y_j(w)$, lower bounding $y_2$ is equivalent to upper bounding $y_1$.
	The definition \textbf{Type Decomposition} ensures that (Lemma~\ref{lem:type-decomposition-probability}):
	\[
		y_1 = \sum_{i \in I : w_{ij} \ge w} \sum_{t \in [T]} \big( 2 x_{ij}^t - \mass^t_i \big)^+
		~.
	\]
	
	If we were in the non-homogeneous Poisson arrival model, we can apply a Converse Jensen Inequality similar to the one shown for IID Online Stochastic Matching~\cite{HuangS:STOC:2021} to derive an upper bound of $1 - \ln 2$.
	Unfortunately, the Non-IID model is not asymptotically equivalent to the (non-homogeneous) Poisson arrival model, unlike their IID counterparts.
	Nevertheless, we prove an approximate version of the Converse Jensen Inequality below.
	Compared to the counterpart, we need to increase the right-hand-side by the second moment of the matched probability of different time steps.
	This lemma may be of independent interest.
	We present its proof in Appendix~\ref{app:converse-jensen}.

	\begin{lemma}[Converse Jensen Inequality]
	    \label{lem:discrete-converse-jensen}

	    For any feasible assignment to the Stochastic Matching LP, and subset $I' \subseteq I$ of online item types, and any offline agent $j \in J$, we have:

	    \[
	        \sum_{i \in I'} \sum_{t \in [T]} \big( 2 x_{ij}^t - \mass_i^t \big)^+ \le 1 - \ln 2 + 2 \sum_t \Big( \sum_{i \in I'} x_{ij}^t \Big)^2
	    \]
	\end{lemma}

    Define $S = \sum_t (\sum_{i : w_{ij} \ge w} x^t_{ij})^2$, we have:
    \[
    	y_2 ~\ge~ y_j(w) - 1 + \ln 2 - 2S
    	~.
    \]

	Therefore, we get the desired bound if $S \le 0.665 < \nicefrac{(0.44-1+\ln 2)}{2}$.
	
	Otherwise, i.e., if $S > 0.0665$, we resort back to the first bound of Lemma~\ref{lem:display-ads-two-way-general}:
	\begin{align*}
		\prod_{t \in [T]} \Big(1 - p^t_1 - \frac{1}{2} p^t_2 \Big)
		&
		=
		\prod_{t \in [T]} \bigg(1 - \sum_{i \in I: w_{ij} \ge w} x^t_{ij} \bigg)
		\tag{Allocation Rate Conservation} \\
		&
		= e^{-y_j(w)} \prod_{t \in [T]} \exp \bigg( \sum_{i \in I: w_{ij} \ge w} x^t_{ij} \bigg) \bigg(1 - \sum_{i \in I: w_{ij} \ge w} x^t_{ij} \bigg) \\
		&
		\le e^{-y_j(w)} \prod_{t \in [T]} \exp \bigg( - \frac{1}{2} \Big( \sum_{i \in I: w_{ij} \ge w} x^t_{ij} \Big)^2 \bigg) 
		\tag{$\ln(1-x) \le -x - \frac{1}{2} x^2$}\\
		&
		= e^{-y_j(w)} \cdot e^{-\frac{1}{2}S }
		~.
	\end{align*}

	This is always smaller than the bound in the theorem because:
	\[
		\Big( 1 + \frac{1}{2} (y-0.44) \Big) e^{-\frac{1}{2} (y-0.44)} + \frac{1}{15} (1-y) e^y
	\]
	is decreasing in $0.44 \le y \le 1$, and its minimum value at $y = 1$ is:
	\[
		\Big( 1 + \frac{1}{2} (1-0.44) \Big) e^{-\frac{1}{2} (1-0.44)} > 0.9674 > e^{-\nicefrac{0.0665}{2}}
		~.
	\]
\end{proof}

\bibliographystyle{plainnat}
\bibliography{matching}

\appendix

\addtocontents{toc}{\protect\setcounter{tocdepth}{1}}

\section{Lossless Simulation in the Query-Commit Model}
\label{app:lossless-simulation}

Consider an instance of a bipartite matching problem in the \emph{Query-Commit Model}.
For unweighted matching, we can arbitrarily let one side of the bipartite graph be the offline vertices and the other side be the online vertices.
For vertex-weighted matching, we let the weighted side be the offline vertices.
For edge-weighted matching with free-disposal, we let the side that could be matched more than once be the offline vertices.
Then, the instance in the \emph{Query-Commit Model} induces a distribution over the realization of the neighborhood of each online vertex, and thus, may be viewed as an instance in the \emph{Best/Random-Order Discrete Model}.

Next, consider any online algorithm for the matching problem in the \emph{Best/Random-Order Discrete Model}.
We will simulate it losslessly in the \emph{Query-Commit Model} as follows.
Suppose the algorithm will next inspect an online vertex $i$, and the corresponding matching decisions would yield matching probabilities $x_{ij}$ for each offline vertex $j$.
The vector $(x_{ij})_{j \in J}$ lies within a $|J|$-dimensional polymatroid defined by:
\begin{align}
    \forall S \subseteq J \quad & \sum_{j \in S} x_{ij} \le 1 -  \prod_{j \in J} (1-p_{ij})
    \label{eqn:lossless-simulation-polytope-1}
    \\
    \forall j \in J \quad & x_{ij} \ge 0
    \label{eqn:lossless-simulation-polytope-2}
\end{align}
where the first set of constraints holds because the probability of matching $i$ to an offline vertex $j \in S$ cannot exceed the probability that there is an edge between $i$ and one of these offline vertices.
Therefore, vector $(x_{ij})_{j \in J}$ can be written as the weighted average of at most $m$ vertices of the polymatroid.
The main insight from \citet{GamlathKS:SODA:2019} was the following lemma.

\begin{lemma}[Lemma 5 of \citet{GamlathKS:SODA:2019} rephrased]
    \label{lem:lossless-simulation-polytope}
    Every vertex of the polytope defined by Equations~\eqref{eqn:lossless-simulation-polytope-1} and \eqref{eqn:lossless-simulation-polytope-2} corresponds to a permutation of a subset of the offline vertices $j_1, j_2, \dots, j_\ell$ such that:
    \begin{align*}
    x_{ij_1} & = p_{ij_1} \\
    x_{ij_2} & = (1-p_{ij_1}) p_{ij_2}  \\
    \dots & \\
    x_{ij_\ell} & = (1-p_{ij_1})\dots(1-p_{ij_{\ell-1}})p_{ij_\ell} \\
    x_{ij} & = 0 \quad \text{if $j \ne j_1, j_2, \dots, j_\ell$}
    \end{align*}
\end{lemma}

We include a proof sketch below for completeness.
    
\begin{proof}
    Each vertex corresponds to $m$ tight constraints.
    Further, we may assume without loss of generality that for any subset $S$ for which Eqn.~\eqref{eqn:lossless-simulation-polytope-1} is one of these $m$ tight constraints, we have $p_{ij} > 0$ for any $j \in S$.
    The first constraint is tight for a sequence of subsets $S_1 \subset S_2 \subset \dots \subset S_\ell$ where $|S_{k+1}| = |S_k|+1$.
    We first prove that they satisfy pairwise proper containment.
    Otherwise, suppose that the subset is tight for $S, T$ such that $S \not\subset T$ and $T \not\subset S$. Consider subsets $S \cap T$ and $S \cup T$. We have:
    \begin{align*}
        1 - \prod_{j \in S} (1-p_{ij}) + 1 - \prod_{j \in T} (1-p_{ij})
        &
        = \sum_{j \in S} x_{ij} + \sum_{j \in T} x_{ij} \\
        &
        = \sum_{j \in S \cap T} x_{ij} + \sum_{j \in S \cup T} x_{ij} \\
        &
        \le 1 - \prod_{j \in S \cap T} (1-p_{ij}) + 1 - \prod_{j \in S \cup T} (1-p_{ij})
    \end{align*}
    which is a contradiction because:
    $$
    \prod_{j \in S \cap T} (1-p_{ij}) + \prod_{j \in S \cup T} (1-p_{ij}) > \prod_{j \in S} (1-p_{ij}) + \prod_{j \in T} (1-p_{ij})
    $$

    Next, suppose for contradiction that $|S_{k+1}| > |S_k| + 1$. The tightness of constraint~\eqref{eqn:lossless-simulation-polytope-1} for $S_k$ and $S_{k+1}$ pins down the values of $\sum_{j \in S_{k+1} \setminus  S_k} x_{ij}$.
    We further need $x_{ij} = 0$ for at least one $j \in S_{k+1} \setminus S_k$ to pin down vector $(x_{ij})_{j \in J}$.
    However, we now get that constraint~\eqref{eqn:lossless-simulation-polytope-1} is violated for $S_{k+1} - j$.
\end{proof}

Given the lemma, we can simulate the algorithm in the \emph{Query-Commit Model} by first sampling a vertex of the polytope according to the decomposition of $(x_{ij})_{j \in J}$.
Then, we can probe edges $(i,j_1), (i,j_2), \dots, (i,j_\ell)$ from Lemma~\ref{lem:lossless-simulation-polytope} one by one, and match the first edge that exists.
Finally, we probe all the remaining edges without including them.

\section{Omitted Analyses of Univariate Functions}
\label{app:univariate-function}

\subsection{Proof of Equations~\eqref{eqn:e-2x} and \eqref{eqn:e-abx}}
\label{app:e-2x}

We will prove that for any $a \ge b$, and $x \ge 0$:
\[
	1 - ax \le (1+bx) e^{-(a+b)x}
	~.
\]

Consider function $h(x) = (1+bx) e^{-(a+b)x} - 1 + ax$.
We only need to show $h(x) \ge 0$ for $x \ge 0$. 
Since $h(0) = 0$, it suffices to prove $h'(x) \ge 0$. 
This follows by:
\begin{align*}
	h'(x) 
	&
	= a - (a + b(a+b)x) e^{-(a+b)x} \\
	&
	\ge a \Big(1 - \big(1 + (a+b) x \big) e^{-(a+b)x} \Big) \tag{$a \ge b$} \\
	&
	\ge 0 \tag{$1+x \le e^x$}
	~.
\end{align*}

\subsection{Proof of Concavity of Equation~\eqref{eqn:socs-basic-matching-bound}}
\label{app:socs-basic-matching-bound}

We first restate the equation below:
\[
    1 - \big( 1 + (y_j - 1 + \ln 2)^+ \big) e^{-y_j-(y_j-1+\ln 2)^+}
    ~.
\]

For $0 \le y_j \le 1 - \ln 2$, it is $1 - e^{-y_j}$ and therefore is concave.
Further, the derivative at $y_j = 1 - \ln 2$ is equal to $e^{-1+\ln2} = \frac{2}{e}$.

For $1 - \ln 2 \le y_j \le 1$, the function is:
\[
    1 - \frac{e}{2} \big( y_j + \ln 2 \big) e^{-2y_j}
    ~.
\]

Its derivative is:
\[
    \frac{e}{2} \big( 2 y_j + 2\ln 2 - 1 \big) e^{-2y_j}
\]

On one hand, it equals $\frac{2}{e}$ at $y_j = 1-\ln 2$, matching the value from the other case.
On the other hand, it is decreasing for $y_j \ge 1-\ln 2$ because the second-order derivative is:
\[
    - 2e \big( y_j-1+\ln2 \big) e^{-2y_j} \le 0
    ~.
\]

\subsection{Proof of Concavity of Equation~\eqref{eqn:socs-matching}}
\label{app:ocs-matching-bound}

We will prove concavity for $0 \le y_j \le \frac{1}{2}$ and $\frac{1}{2} \le y_j \le 1$ separately, and verify that the left and right derivatives at $y_j = \frac{1}{2}$ are equal.

For $0 \le y_j\le \frac{1}{2}$, the function is:
\[
	\frac14 +\frac{y_j}{2} - \frac{1}{4}e^{-2y_j}
	~.
\]

Its derivative $\frac{1}{2} + \frac{1}{2} e^{-2y_j}$ is decreasing in $y_j$.
Further, the derivative at $y_j = \frac{1}{2}$ equals $\frac12 + \frac{1}{2e}$.

For $\frac{1}{2} \le y_j \le 1$, the function is:
\[
	1 - e^{-2y_j}\left(\frac{e+1}{4} + \frac{e}{2}y_j\right)
	~.
\]

Its derivative is $e^{-2y_j}\left(\frac12 + e\cdot y_j\right)$. On one hand, it equals $\frac12 + \frac{1}{2e}$ at $y_j = \frac12$, matching the value from the other case. On the other hand, it is decreasing for $y_j\ge \frac12$ because the second-order derivative is:
\[
	e^{-2y_j} \big(e - 1 - 2e\cdot y_j \big)
\]
which is negative when $y_j \ge \frac{1}{2}$.

\subsection{Proof of Concavity of Equation~\eqref{eqn:socs-matching-random-order}}
\label{app:socs-matching-random-order}

We will prove concavity for $0 \le y_j \le \frac{1}{2}$ and $\frac{1}{2} \le y_j \le 1$ separately, and verify that the left and right derivatives at $y_j = \frac{1}{2}$ are equal.

For $0 \le y_j\le \frac{1}{2}$, the function is:
\[
	1 - \big(1 + \frac{y_j}{2}\big) e^{-2y_j} - \frac{1}{2} y_j \big(1 - y_j \big)
	~.
\]

Its derivative is $(y_j+\frac32)e^{-2y_j} + y_j - \frac12$;
it equals $\frac{2}{3}$ at $y_j = \frac{1}{2}$.

Further, the second-order derivative is $-(2y_j + 2)e^{-2y_j} + 1$.
Rearranging terms, the non-positivity of this second-order derivative is equivalent to:
\[
	e^{2y_j} \le 2 + 2y_j
	~.
\]

The left-hand-side is convex and the right-hand-side is linear.
Hence, it suffices to verify the inequality at $y_j = 0$, where we have $1 < 2$, and at $y_j = \frac{1}{2}$, where we have $e < 3$.

For $\frac{1}{2} \le y_j \le 1$, the function is:
\[
	1 - e^{-2y_j} \big( 1 + \big(\frac{1}{2} +  \frac{e}{4}\big)  y_j \big)
	~.
\]

Its derivative is $e^{-2y_j}\left(\frac32 - \frac{e}{4} + (1+\frac{e}{2})\cdot y_j\right)$. 
In particular, it equals $\frac2e$ at $y_j = \frac12$, matching the value from the other case. 

Further, it is decreasing for $y_j\ge \frac12$ because the second-order derivative is:
\[
	e^{-2y_j} \big(e - 2 - (2 + e) y_j \big)
\]
which is negative when $y_j \ge \frac{1}{2}$ because $e - 2 - \nicefrac{(2+e)}{2} < 0$.

\subsection{Proof of Concavity of Equation~\eqref{eqn:adwords-convergence-function}}
\label{app:adwords-convergence-function}

We first restate the function, removing a constant term $-\log 4$ and considering the natural logarithm without loss of generality:

\[
    f(x) := \ln\left(3 + (1+x)e^{-x}\right). 
\]

The first-order and second-order derivatives are:
\begin{align*}
    & f'(x) = -\frac{x}{1 + x + 3e^x} ~,\\
    & f''(x) = -\frac{1 + 3(1-x) e^x}{(1 + x + 3e^x)^2} ~.
\end{align*}

For any $0 \le x \le 1$, we have $f''(x) < 0$.

\subsection{Proof of Concavity of Equation~\eqref{eqn:adwords-general-socs-convergence}}
\label{app:adwords-general-socs-convergence}

Let $c = 0.413$.
The equation is:
\[
    1 - g(y_j) = 1 - e^{-y_j}\left(\frac34 + \frac14 \cdot \frac{1 + (y_j-c)^+ / 4}{e^{(y_j - c)^+/4}}\right)^2.
\]

We show the concavity for $0 \le y_j \le c$ and $c \le y_j \le 1$ separately, and verify that the left and right derivatives at $y_j = c$ are equal.

For $0 \le y_j\le c$, the function is:
\[
	1 - e^{-y_j}
	~,
\]
which is concave, and the derivative at $y_j = c$ equals $e^{-c}$.

For $c \le y_j \le 1$, consider the change of variable $x = y_j - c$, the function is:
\[
	1 - e^{-x-c}\left(\frac34 + \frac14 \cdot \frac{1+x/4}{e^{x/4}}\right)^2
	~.
\]

Its derivative is 
\[
    \frac{1}{16} e^{-3x/2-c} \cdot \left(3e^{x/4} + \frac{x}{4} + 1\right)\left(3e^{x/4} + \frac{3x}{8} + 1\right)
\]

It equals $e^{-c}$ at $y_j = c ~ (x = 0)$, matching the value from the other case. Further, it is decreasing since it equals
\[
	\frac{e^{-c}}{16} \cdot e^{-7x/8} \cdot \left(3 + \frac{1+x/4}{e^{x/4}}\right) \left(3e^{-x/8} + \frac{1+3x/8}{e^{3x/8}}\right)
\]
where every term is non-increasing for $x\ge 0$. Hence, the equation is concave.

\subsection{Proof of Monotonicity of Equation~\eqref{eqn:display-ads-y-minus-coef}}
\label{app:display-ads-y-minus-coef}

Recall that we want to prove that:
\[
	f(x) = e^{-x} \big(1-e^{-\frac{x}{2}}(1+\frac{x}{2})\big)
\]
is non-decreasing in $0 \le x \le 1$.

The derivative is
\begin{align*}
    f'(x) = e^{-\frac{3}{2}x} \Big(1+\frac{3}{4}x-e^{\frac{x}{2}} \Big)
    ~.
\end{align*}

We next verify the non-negativity of $1+\frac{3}{4}x-e^{\frac{x}{2}}$.
Since this is concave, it achieves its minimum value at $x = 0$ or $x = 1$.
It is equal to $0$ at $x = 0$, and $1+\frac{3}{4}-\sqrt{e} > 0$ at $x = 1$.

\subsection{Proof of Inequality~\eqref{eqn:display-ads-ratio}}
\label{app:display-ads-ratio}

We first restate the inequality below.
For any $0 \le y \le 1$, we need to show that:
\[
	1 - g(y)
    ~=~ 
    1 - \min \Big\{
    	e^{-y}
        \,,\,
        e^{-y} \Big(1 + \frac12 (y-0.44)^+ \Big) e^{- \frac12 (y-0.44)^+} + \frac{1-y}{15}
    \Big\}
    \ge 0.644 \cdot y
    ~.
\]

We will consider two cases depending on the value of $y$. 
If $0 \leq y < 0.9$, we use $1- g(y) \geq 1 - e^{-y}$.
Since $e^{-y}$ is concave and $1 - g(0) = 0$, we have:
\[
	1- g(y) \geq 1 - e^{-y} \geq \frac{1 - e^{-0.9}}{0.9} y > 0.659 \cdot y
	~.
\]

It remains to consider the case when $0.9 \leq y \leq 1$.
We apply the second bound in this case and need to prove that:
$$
	1 -e^{-y} \Big(1 + \frac12 (y-0.44) \Big) e^{- \frac12 (y-0.44)} - \frac{1-y}{15} > 0.644 \cdot y
	~.
$$
            
We define $$f(y) = 1 -e^{-y} \Big(1 + \frac12 (y-0.44) \Big) e^{- \frac12 (y-0.44)} - \frac{1-y}{15} - 0.644 \cdot y ~.$$

The derivative is 
\begin{align*}
    f'(y) =& \frac{3}{2}e^{-y} \Big(1 + \frac12 (y-0.44) \Big) e^{- \frac12 (y-0.44)} - e^{-y} \Big(1 + \frac12 (y-0.44) \Big) e^{- \frac12 (y-0.44)} + \frac{1}{15} - 0.644\\
    =&e^{-y} \Big(\frac{1}{2} + \frac34 (y-0.44) \Big) e^{- \frac12 (y-0.44)} + \frac{1}{15} - 0.644\\
    \leq &e^{-0.9} \Big(\frac{1}{2} + \frac34 (1-0.44) \Big) e^{- \frac12 (0.9-0.44)} + \frac{1}{15} - 0.644\\
    = &e^{-1.13} \times 0.92 + \frac{1}{15} - 0.644
    <-0.28
    ~.
\end{align*}

Therefore, the function $f(y)$ is decreasing in $0.9 \leq y \leq 1$.
Hence, for any $0.9 \leq y \leq 1$, we have
\begin{align*}
    f(y) \geq& f(1)\\
    =& 1 -e^{-1} \Big(1 + \frac12 (1-0.44) \Big) e^{- \frac12 (1-0.44)} - 0.644\\
    >&0.0001
    ~.
\end{align*}

\section{Missing Proofs from Section~\ref{sec:vertex-weighted}}
\label{app:vertex-weighted}

\subsection{Proof of Lemma~\ref{lem:ocs-basic}}
\label{app:ocs-basic}

We will prove the lemma by an induction on the time step $t$ from $0$ to $T$.
The base case when $t = 0$ is trivial because both sides are equal to $1$.

Suppose that the inequality holds for time step $t-1$.
We next consider time step $t$.
By the inequality in Lemma~\ref{lem:ocs-recurrence-am-gm} and the induction hypothesis for subsets $S$ and $S+k$ at time $t-1$, we have:
\[
	u_S^t \:\le \sum_{\{ j, k \} : j, k \notin S} \mass_{\{j,k\}}^t \cdot e^{-\sum_{\ell \in S} y_\ell^{1:(t-1)}} + \frac{1}{2} \sum_{\{j, k\} : j \in S, k \notin S} \mass_{\{j,k\}}^t \cdot e^{- \sum_{\ell \in S} y_\ell^{1:(t-1)} - y_j^{1:(t-1)}}
\]

We further relax the right-hand-side by dropping $y_j^{1:(t-1)}$ from the exponent of the last term:
\begin{align*}
	u_S^t 
	&
	\: \le \bigg( \sum_{\{ j, k \} : j, k \notin S} \mass_{\{j,k\}}^t \:+\: \frac{1}{2} \sum_{\{j, k\} : j \in S, k \notin S} \mass_{\{j,k\}}^t \bigg) \cdot e^{- \sum_{\ell \in S} y_\ell^{1:(t-1)}} \\
	&
	\: = \bigg( 1 - \sum_{\{ j, k \} : j, k \in S} \mass_{\{j,k\}}^t \:-\: \frac{1}{2} \sum_{\{j, k\} : j \in S, k \notin S} \mass_{\{j,k\}}^t \bigg) \cdot e^{- \sum_{\ell \in S} y_\ell^{1:(t-1)}}
	~.
\end{align*}

Observe that:
\[
	\sum_{j \in S} \: y_j^t ~ = ~ \sum_{\{ j, k \} : j, k \in S} \mass_{\{j,k\}}^t \:+\: \frac{1}{2} \sum_{\{j, k\} : j \in S, k \notin S} \mass_{\{j,k\}}^t
	~.
\]

Hence, we have:
\begin{align*}
	u_S^t 
	&
	~ \le ~ 
	\bigg( 1 - \sum_{j \in S} \: y_j^t \bigg) e^{- \sum_{j\in S} y_j^{1:(t-1)}} \\
	&
	~ \le e^{- \sum_{j \in S} y_j^{1:t}} 
	~,
\end{align*}
where the second inequality follows by $1 - y \le e^{-y}$ and $y_j^{1:t} = y_j^{1:(t-1)} + y_j^t$.

\subsection{Proof of Lemma~\ref{lem:general-ocs-recurrence}}
\label{app:general-ocs-recurrence}

The argument relies on the fact that the arrival of an online vertex at time $t$ is independent to the arrivals of online vertices before time $t$.
There are three cases depending on the arrival at time $t$.

First, if no online vertex arrives at time $t$, or the arrived online vertex's surrogate type does not involve any offline vertex in $S$, then the vertices in $S$ are unmatched after time $t$ if and only if they are unmatched before time $t$, which happens with probability $u_S^{t-1}$.
This corresponds to the first term on the right-hand-side.

Second, if a single-way surrogate type $j \in S$ or a two-way surrogate $\{j,k\}$ with $j, k \in S$ arrives at time $t$, then at least one vertex in $S$ is matched at the end of time $t$.
Hence, this case contributes zero to the right-hand-side.

Finally, if a two-way surrogate type $\{j,k\}$ arrives with $j \in S$ and $k \notin S$, then the vertices in $S$ are unmatched after time $t$ if (1) the vertices in $S+k$ are unmatched before time $t$, and (2) the algorithm matches to $k$ at time $t$.
The former happens with probability $u_{S+k}^{t-1}$.
The latter happens with probability:
\[
    \frac{e^{2y_k^{1:(t-1)}}}{e^{2y_k^{1:(t-1)}} + e^{2y_j^{1:(t-1)}}}
\]
by the definition of the algorithm.
Note that by $\lambda_{j,k}(t) = \lambda_{k,j}(t)$ (Proportionality), the total arrival rate of these two types is $2 \lambda_{j,k}(t)$.
This corresponds to the second term on the right-hand-side.

\subsection{Proof of Lemma~\ref{lem:matching-ocs-basic}}
\label{app:matching-ocs-basic}

We will prove this lemma by induction on $t$.
The base case $t=0$ is trivial since both sides are equal to $1$.
Suppose that the lemma holds for time step $t-1$.
We next consider time step $t$.
By Lemma~\ref{lem:general-ocs-recurrence-am-gm} and applying the induction hypothesis to subsets $S$ and $S+k$, we have:
\begin{align*}
    u_S^t
    &
    ~\le~ \bigg(1 - \sum_{j \in S}\mass_j^t - \sum_{\{j,k\} : j, k \in S} \mass_{\{j,k\}}^t - \sum_{\{j,k\} : j\in S, k\notin S} \mass_{\{j,k\}}^t \bigg) \cdot e^{- \sum_{j \in S} x_j^{1:(t-1)}} \\
    & \qquad\qquad
    + ~ \frac{1}{2} \sum_{\{j,k\} : j\in S, k\notin S}\mass_{\{j,k\}}^t \cdot e^{- \sum_{\ell \in S} x_\ell^{1:(t-1)} - x_j^{1:(t-1)}}
    ~.
\end{align*}

Finally, we drop $x_j^{1:(t-1)}$ from the exponent term of the last term.
We get that:
\begin{align*}
    u_S^t
    &
    ~\le~ 
    \bigg( 1 \:- \: \sum_{j\in S} \mass_j^t \: - \sum_{\{j,k\} : j, k\in S} \mass_{\{j,k\}}^t - \:\frac{1}{2} \sum_{\{j,k\}: j\in S, k\notin S} \mass_{\{j,k\}}^t \bigg) \cdot e^{- \sum_{j \in S} x_j^{1:(t-1)}} \\
    &
    ~=~
    \bigg( 1 - \sum_{j\in S} y_j^t \bigg) \cdot e^{- \sum_{j \in S} y_j^{1:(t-1)}}
    ~.
\end{align*}

The claim now follows by $y_j^{1:t} = y_j^{1:(t-1)} + y_j^t$ and $1-y \le e^{-y}$.

\section{Multi-Way OCS and AdWords: Proof of Lemma~\ref{lem:adwords-competitive}}
\label{app:adwords-competitive}

We first define the following parameters:
\begin{align*}
	\Gamma
	&
	~=~
	1 - \int_0^\infty g(z) e^{-z} \,\dif z
	~;
	\label{eqn:adwords-ratio}
	\\
	\beta(y)
	&
	~=~ - e^y \int_y^\infty g'(z) e^{-z} \, \dif z
	= g(y) - e^y \int_y^\infty g(z) e^{-z} dz
	~;
	\\[1ex]
	\alpha(y)
	&
	~=~
	- g'(y) - \beta(y) 
	~.
\end{align*}

We next explain the Balance-OCS algorithm with the above parameters.
For each time step $t$, imagine that we allocate the item by infinitesimal pieces.
When we allocate an (infinitesimal) $\varepsilon$ amount of the item at time $t$ to an offline vertex $j$, the used portion of $j$'s budget $y_j$ increases by:
\[
	\dif y_j = \varepsilon \cdot \frac{b^t_j}{B_j}
\]

Hence, its expected contribution to the objective, according to the convergence rate of the OCS, increases by:
\[
	\varepsilon \cdot \frac{b^t_j}{B_j} \cdot g'(y_j) \cdot B_j = \dif y_j \cdot g'(y_j) \cdot B_j
	~.
\]

We will distribute this increment between offline vertex $j$ and online vertex $t$.
Concretely, offline vertex $j$ gets:
\[
	\varepsilon \cdot \frac{b^t_j}{B_j} \cdot \alpha(y_j) \cdot B_j = \dif y_j \cdot \alpha(y_j) \cdot B_j
	~,
\]
and online vertex $t$ gets:
\[
	\varepsilon \cdot \frac{b^t_j}{B_j} \cdot \beta(y_j) \cdot B_j = \dif y \cdot \beta(y_j) \cdot B_j
	~.
\]

We will denote the cumulative gain/utility of offline vertices $j$ and online vertices $t$ from the above gain splitting process by $u_j$ and $u_t$.
By the definition of this gain splitting rule, we have the following invariants.

\begin{lemma}
	\label{lem:adwords-objective}
	The expected objective of the rounded solution given by the OCS is at least:
	\[
		\sum_{j \in J} u_j + \sum_{t=1}^T u_t
		~.
	\]	
\end{lemma}

\begin{lemma}
	\label{lem:adwords-offline-invariant}
	For any offline vertex $j \in J$, we have:	
	\[
		u_j = \int_0^{y_j} \alpha(y) \:\dif y \cdot B_j
		~.
	\]
\end{lemma}

The allocation rule will greedily maximize the gain distributed to online vertex $t$, by allocating each infinitesimal piece of the item to the offline vertex with the maximum:
\[
	b^t_j \cdot \beta(y_j)
	~.
\]

Note that the marginal return for allocating pieces over an online vertex $t$ is non-increasing in this continuous allocation process.
An equivalent way to define this algorithm is to find a threshold marginal return $\theta > 0$, such that the total amount of allocation to the offline vertices subject to having marginal return at least $\theta$ is equal to $1$.
More precisely, given $\theta$, the amount of allocation to offline vertex $j$ can be computed as:
\[
	\matchrate^t_j(\theta) = \frac{B_j}{b^t_j} \left( \beta^{-1} \Big( \frac{\theta}{b^t_j} \Big) - y_j \right)^+
	~.
\]
We will artificially define $\matchrate^t_j(\theta) = 0$ if online vertex $t$ bids zero for offline vertex $j$, i.e., if $b^t_j = 0$.

For any offline vertex $j$ with a positive bid $b^t_j > 0$, this is decreasing in $\theta$, and satisfies that $\lim_{\theta \to 0} \matchrate^t_j(\theta) = \infty$ and $\lim_{\theta \to \infty} \matchrate^t_j(\theta) = 0$.
Hence, there exists a $\theta$ for which:
\begin{equation}
	\label{eqn:adwords-balance-ocs-allocation}
	\sum_{j \in J} \matchrate^t_j(\theta) = 1
	~.
\end{equation}

\bigskip

\begin{tcolorbox}
	\textbf{Balance-OCS for AdWords}\\[2ex]
	For each time step $t \in [T]$:
	\begin{enumerate}
		\item Find a threshold $\theta \in (0, \infty)$ that satisfies Equation~\eqref{eqn:adwords-balance-ocs-allocation}.
		\item Let $\matchrate^t_j(\theta)$ be the fractional allocation of item $t$ to each offline vertex $j$.
		\item Let the OCS for AdWords select based on fractional allocation $\matchrate^t(\theta)$.
	\end{enumerate}
\end{tcolorbox}

\bigskip

We will now prove that this algorithm is $\Gamma$ competitive.
Consider any offline vertex $j$, and the subset of online vertices allocated to it in the optimal allocation, denoted as $S_j$.
It suffices to prove that the total gains of offline vertex $j$ and online vertices in $t$ sum to at least:
\[
	u_j + \sum_{t \in S_j} u_t \ge \Gamma \cdot \min \Big\{ \sum_{t \in S_j} b^t_j, B_j \Big\}
	~.
\]

We have already characterized Offline vertex $j$'s contribution in Lemma~\ref{lem:adwords-offline-invariant}.
It remains to analyze the contribution of online vertices $t \in S_j$.
By the greedy fractional allocation rule, the marginal return per unit of online vertex $t$ is at least:
\[
	b^t_j \cdot \beta(y_j)
\]
because allocating to offline vertex $j$ at time $t$ would yield at least as much marginal return.
Hence, we get that:
\[
	u_t \ge b^t_j \cdot \beta(y_j)
	~.
\]

Summing over the online vertices $t \in S_j$, we have:
\[
	\sum_{t \in S_j} u_t \ge \sum_{t \in S_J} b^t_j \cdot \beta(y_j)
	~.
\]

Combining with Lemma~\ref{lem:adwords-offline-invariant} gives:
\[
	u_j + \sum_{t \in S_j} u_t \ge \left( \int_0^{y_j} \alpha(y) \:\dif y + \beta(y_j) \right) \cdot \min \Big\{ \sum_{t \in S_j} b^t_j, B_j \Big\}
	~.
\]

Finally, we verify that for any $y_j$:
\begin{equation}
	\label{eqn:adwords-alpha-beta-invariant}
	\int_0^{y_j} \alpha(y) \:\dif y + \beta(y_j) = \Gamma
	~.
\end{equation}

For $y_j = 0$, it holds because:
\begin{align*}
	\beta(0)
	&
	= g(0) - \int_0^\infty g(z) e^{-z} \,\dif z \\
	&
	= 1 - \int_0^\infty g(z) e^{-z} \,\dif z  
	= \Gamma
	~.
\end{align*}

Further, note that the derivative of $\beta(y)$ is:
\[
	- e^{y} \int_y^\infty g'(z) e^{-z} \,\dif z + e^y g'(y) e^{-y} = \beta(y) + g'(y)
	~.
\]

Hence, its derivative of Equation~\eqref{eqn:adwords-alpha-beta-invariant} w.r.t.\ $y_j$ equals:
\[
	\alpha(y_j) + \beta'(y_j) = 0
	~.
\]

\section{Converse Jensen Inequality: Proof of Lemma~\ref{lem:discrete-converse-jensen}}
\label{app:converse-jensen}

Consider all $i \in I'$ and any $t \in [T]$ such that $x_{ij}^t \ge \frac{1}{2} \mass_i^t$.
Without loss of generality, we may consider the case when for each $t$ there is only one such $i$; 
otherwise, we can merge them into a single type in the following argument. 
Hence, we will omit the subscript $i$ and rename these steps as $t = 1, 2, 3, \dots, n$.
We will fix $x_j^t$’s and consider $\mass^t$’s as variables.
We next argue that subject to the above constraints, $\sum_{t \in [n]} \mass^t$ is minimized when:
\begin{equation}
    \label{eqn:discrete-converse-jensen-characterize}
    \mass^1 = x_j^1 
    ~,\quad
    \mass^2 = \frac{x_j^2}{1-x_j^1} 
    ~,\quad
    \mass^n = \frac{x_j^n}{1-x_j^1-x_j^2-\dots-x_j^{n-1}}
    = \frac{x_j^n}{1-x_j^{1:(n-1)}}
    ~.
\end{equation}
for some order of the rounds. 
In fact, it is maximized when $x_j^1 \le x_j^2 \le \dots \le x_j^n$ but we do not need this for proving our claim.

First, it is easy to verify that the above $\mass(t)$’s are feasible. For any subset $S \subseteq [n]$, we have:
\begin{align*}
    \prod_{t \in S} \big(1 - \mass^t\big)
    &
    = \prod_{t \in S} \frac{1-x_j^{1:t}}{1-x_j^{1:(t-1)}} \\
    &
    \le \prod_{t \in S} \frac{1-\sum_{t' \in S : t' \le t} x_j^{t'}}{1-\sum_{t' \in S : t' < t} x_j^{t'}} \\[1ex]
    &
    = 1 - x_j^S
    ~.
\end{align*}

Next, we prove its optimality.
We change variables by letting $\alpha^t = \ln \big(1- \mass^t\big)$.
The problem becomes maximizing a convex function:
\[
    \sum_{t \in [n]} e^{\alpha^t}
\]
subject to the constraints that for any subset of steps $S \subseteq [n]$:
\[
    \alpha^S = \sum_{t \in S} \alpha^t \le \ln \Big( 1 - x_j^S\Big)
    ~.
\]

This is maximized at a vertex. 
Each vertex corresponds to $n$ tight constraints.
The constraints must correspond to $n$ subsets $S_1 \subset S_2 \subset \dots \subset S_n$ satisfying proper containments. Otherwise, suppose the constraints are tight for two subsets $S, T$ but $S \not\subset T$ and $T \not\subset S$. We have:
\begin{align*}
    \ln \Big( 1 - x_j^S \Big) + \ln \Big( 1 - x_j^T \Big)
    &
    = \alpha^S+ \alpha^T \\
    &
    = \alpha^{S \cup T} + \alpha^{S \cap T} \\
    &
    \le \ln \Big( 1 - x_j^{S \cup T} \Big) + \ln \Big( 1 - x_j^{S \cap T} \Big)
    ~.
\end{align*}
This is a contradiction to the concavity of function $\ln (1-x)$, because $x^S + x^T = x^{S \cup T} + x^{x \cap T}$ and $x^{S \cup T} \ge x^S, x^T$, and $x^S, x^T \ge x^{S \cap T}$.

Further, the adjacent sets' sizes differ by $1$, i.e., $|S_i| = |S_{i-1}| + 1$ for any $1 < i \le n$.
Otherwise, there must be $t \in S_i \setminus S_{i-1}$ such that $x_j(t) = 0$.
However, the constraint \eqref{eqn:discrete-time-lp} is then violated for $S = S_i - t$.

Finally, we prove the stated inequality subject to the characterization in Eqn.~\eqref{eqn:discrete-converse-jensen-characterize}.
Hence, the left-hand-side of the inequality equals:
\[
    \sum_{t \in [n]} \big( 2 x_j^t - \mass^t \big) = \sum_{t\in[n]} \Big( 2 x_j^t - \frac{x_j^t}{1-x_j^{1:(t-1)}} \Big) 
    ~.
\]

We will compare:
\[
    2 x_j^t - \frac{x_j^t}{1-x^{1:(t-1)}}
\]
with:
\[
    \int_{x_j^{1:(t-1)}}^{x_j^{1:t}} \Big( 2 - \frac{1}{1-z} \Big)^+ \,\dif z
\]
and bound the difference by $2 x_j(t)^2$.

For $t < n$, we have $x_j^{1:t} \le \frac{1}{2}$, because otherwise we would have $x_j^{t+1} < \frac{1}{2} \mass^{t+1}$, and thus, round $t+1$ would not qualified as one of the $n$ rounds.
Then, the difference for $t$ is:
\[
    - \ln \Big( 1 - \frac{x_j^t}{1-x^{1:(t-1)}} \Big) - \frac{x_j^t}{1-x^{1:(t-1)}} 
    ~.
\]

Since $- \ln (1-y) - y$ is increasing and $x_j^{1:(t-1)} \le \frac{1}{2} - x_j^t$, this is at most:
\[
    - \ln \Big( 1 - \frac{x_j^t}{\frac{1}{2}+x_j^t} \Big) - \frac{x_j^t}{\frac{1}{2}+x_j^t}
    \le 2 \big( x_j^t \big)^2
    ~.
\]
    
For $t = n$, if $x_j^{1:n} \le \frac{1}{2}$, the above argument still works.
If $x_j^{1:n} > \frac{1}{2}$, the benchmark that we compared to does not depend on $x_j^n$ as long as $x_j^n \ge \frac{1}{2} - x_j^{1:(n-1)}$.
We want to show that:
\[
    2 \big( x_j^n \big)^2 - \Big( 2 x_j^n - \frac{x_j^n}{1-x_j^{1:(n-1)}} \Big) + \int_{x_j(1:n-1)}^{\frac{1}{2}} \Big( 2 - \frac{1}{1-z} \Big) \,\dif z \ge 0
    ~.
\]

The derivative with respect to $x_j^n$ is:
\begin{align*}
    4x_j^n - 2 + \frac{1}{1-x_j^{1:(n-1)}} &
    ~ \ge ~ 4 \Big(\frac{1}{2} - x_j^{1:(n-1)}\Big) - 2 + \frac{1}{1 - x_j^{1:(n-1)}} \\
    &
    ~ = ~ \frac{(1-2x_j^{1:(n-1)})^2}{1-x_j^{1:(n-1)}} \ge 0
    ~.
\end{align*}

Hence, it is minimized when $x_j^n = \frac{1}{2} - x^{1:(n-1)}$, and thus, reducing to the case when $x_j^{1:n} \le \frac{1}{2}$ for which the above argument still works.

In sum, the left-hand-side of the inequality is at most:
\[
    \int_0^{x_j^{1:n}} \Big(2 - \frac{1}{1-z}\Big)^+ \,\dif z + 2 \sum_{t \in [n]} \big( x_j^t \big)^2
    \le \int_0^1 \Big(2 - \frac{1}{1-z}\Big)^+ \,\dif z + 2 \sum_{t \in [T]} \big( x_j^t \big)^2
    ~.
\]

The theorem then follows by verifying that the integral equals $1 - \ln 2$.

\addtocontents{toc}{\protect\setcounter{tocdepth}{2}}

\end{document}